%% file: virtual_clocks.tex
\documentclass[manyauthors]{fundam}

\publyear{26}
\papernumber{1}
\volume{196}
\issue{2}
\theDOI{10.46298/fi.15143}

\versionForARXIV

\usepackage{amsfonts}
\usepackage{mathrsfs}
\usepackage{amssymb}

\usepackage{tikz}
\usetikzlibrary{automata, arrows.meta, positioning}
\usetikzlibrary{backgrounds,automata}
\usetikzlibrary{calc}

\usepackage[section]{algorithm}
\usepackage{algpseudocode}
\numberwithin{algorithm}{section}
\algnewcommand{\IfThenElse}[3]{
  \State \algorithmicif\ #1\ \algorithmicthen\ #2 \algorithmicelse\ #3}
\algnewcommand{\IfThen}[2]{
  \State \algorithmicif\ #1\ \algorithmicthen\ #2}

\usepackage{subcaption}
\usepackage{hyperref}
\usepackage{cancel}
\usepackage[shortlabels]{enumitem}
\usepackage[toc,page]{appendix}
\usepackage{varwidth}
\usepackage{ifthen}
\usepackage{mathtools}

\usepackage{stackengine,scalerel,graphicx,trimclip,amsmath,wasysym}

\usepackage{alphalph}

\usepackage{mymacros}

\usepackage{listings}

\appto\appendix{
}

\ifdefined\SHORTSECOND
\fi

\begin{document}
\title{\papertitle}
\author{Alexander Lieb \thanks{This work has been funded by the Deutsche Forschungsgemeinschaft (DFG, German Research Foundation) – Project-ID 210487104 - SFB 1053.} \\ \tudarmstadt \\alexander.lieb\esmailending%
\and Hendrik Göttmann\\ \tudarmstadt \\hendrik.goettmann\esmailending%
\and Malte Lochau\\University of Siegen, Germany\\malte.lochau@uni-siegen.de%
\and Lars Luthmann\\Accso -- Accelerated Solutions GmbH\\Darmstadt, Germany\\lars.luthmann@accso.de%
\and Andy Schürr\\ \tudarmstadt \\andy.schuerr\esmailending }%
\maketitle%
\begin{abstract}
Timed automata are a widely used formalism for specifying the discrete-state/continu\-ous-time behavior of time-critical reactive systems. 
For the fundamental verification problem of comparing 
two timed automata,
it has been shown that timed trace equivalence is undecidable, 
while timed bisimulation 
is decidable. 
The corresponding decidability proof uses region graphs, a finite but 
space-consuming characterization of timed automata semantics. 
Most verification tools use zone graphs instead, a symbolic and, on average, 
more space-efficient representation of timed automata semantics. 
However, zone graphs provide correct results only for those verification tasks that are 
reducible to 
reachability problems, and are too imprecise 
for timed bisimilarity checking. 
To the best of our knowledge, there is currently no practical tool 
for automated timed bisimilarity checking. 
In this paper, we propose a new representation of timed automata semantics
that extends zone graphs by so-called virtual clocks.
Our zone-based construction is, on average, significantly smaller than the corresponding region graph representation.
We also present experimental results obtained by applying our tool implementation to timed automata models, which are often used to evaluate timed automata analysis techniques.
\end{abstract}


\runninghead{A. Lieb et al.}{\papertitle}

\input{structure.tex}

\bibliographystyle{fundam}
\bibliography{sections/references}

\end{document}

%% file: structure.tex
\input{sections/introduction}
\input{sections/background}
\input{sections/virtual_clocks}
\input{sections/algorithm}
\input{sections/evaluation}
\input{sections/conclusion}

%% file: sections/introduction.tex
\section{Introduction}
\paragraph{Background and Motivation.}
Timed automata specify discrete-state/continuous-time behavior by means of
labeled state-transition graphs of classical finite automata models, where states 
are called \emph{locations} and transitions are denoted as \emph{switches}~\cite{AlurDill90AutomataForModelingRealTimeSystems}.
Timed automata extend classical automata by a set $C$ of \emph{clocks}
constituting constantly and synchronously increasing, yet independently resettable
numerical variables.
Clock values can be referenced within \emph{clock constraints}
to define boundaries for time intervals in which 
switches are allowed to happen in a timed run of the automaton.
In this way, timed automata act as acceptors 
of languages over (timed) traces denoted as pairs of actions and timestamps.

A fundamental verification problem arises from the comparison of a candidate
implementation against a specification of a real-time system, both specified as timed automata over the same alphabet of actions. 
It has been shown that timed trace inclusion (and therefore also timed trace equivalence) 
is undecidable, whereas timed bisimulation is decidable~\cite{Cerans1992}. 
This makes timed bisimilarity a particularly relevant equivalence notion 
for verifying time-critical behaviors.
The undecidability of (timed) trace inclusion is due to the
limitation of timed traces to externally recognizable behavior 
(i.e., actions and delays) which does, however, not offer enough information about the internal
structure of the respective timed automata to reason about trace inclusion.
In particular, (non-visible) resets of clocks potentially caused as part of a switch
in a timed run may lead to (arbitrarily deferred) semantic effects not immediately recognizable 
in the corresponding timed step.
In contrast, the timed generalization of bisimulation equivalence 
is indeed rich enough to capture the discriminating effects of clock resets.
However the original decidability proof of \u{C}er\={a}ns employs region graphs,
a finite, but often unnecessarily space-consuming representation of timed automata semantics 
(i.e., having $\mathcal{O}(|C|!\cdot k^{|C|})$ many regions, where $k$ is the maximum constant upper time bound 
of all clock constraints). 
Instead, most recent verification tools for timed automata are based on zone graphs which use a symbolic and, on average, 
more space-efficient representation of timed automata semantics than region graphs~\cite{Henzinger1994}. 
However, zone graphs only produce sound results for analysis tasks being reducible to plain location-reachability problems 
thus being too imprecise for checking timed bisimilarity~\cite{WeiseUndLenzke1997}.
In particular, zone graphs may not precisely reflect the possible impact of, by definition invisible, 
clock resets on the branching behavior in some subsequent step of a timed run as long as 
the reset does not affect reachability of locations.
\paragraph{Contributions.} 
In this manuscript, we propose a new characterization of
timed automata semantics. 
We extend zone graphs to carry additional information
required for timed bisimilarity checking.
To this end, our enriched notion of zones includes not only the 
current time intervals (zones) of the original clocks
of the two automata under comparison, but additionally contains information about possible deviations due to
(by definition non-visible) clock resets in the respective other automaton.
We capture these deviations between clock valuations by adding so-called
\emph{virtual clocks} to zone graphs, which constitute proper, yet hidden clock variables.
Our approach works for the deterministic as well as the non-deterministic case, where
the second case causes substantially more computational effort.
Our constructions and correctness results are currently focused on bisimilarity of timed (safety) automata 
(i.e., timed automata without acceptance states).
\paragraph{Tool Support and Reproducibility.}
Our tool implementation supports the TChecker file format \cite{TCheckerFileFormat}
for input models and is available online \footnote{https://github.com/Echtzeitsysteme/tchecker}. \textsc{Uppaal} \cite{UppaalInANutshell} models can
be converted into the TChecker file format \cite{UppaalToTchecker}.
This web page also contains experimental data sets and further
information for reproducing the evaluation results.
\paragraph{Related Work.} 
The notion of timed bisimulation goes back to Moller and Tofts~\cite{Moller1990}
as well as Yi~\cite{Yi1990}, both originally defined on real-time extensions of the process algebra
CCS. 
Similarly, Nicollin and Sifakis~\cite{Nicollin1994} define timed bisimulation on ATP (Algebra
of Timed Processes). 
However, none of these works initially incorporated a
technique for effectively checking bisimilarity. 
The pioneering work of \u{C}er\={a}ns~\cite{Cerans1992}
includes the first decidability proof of timed bisimulation on timed labeled transition systems by providing a
finite characterization of bisimilarity-checking using region graphs. 
The improved (i.e., less space-consuming) approach of Weise and Lenzkes~\cite{WeiseUndLenzke1997} employs a variation
of zone graphs, called full backward stable graphs. Their work is most closely related to our approach.
However, the approach lacks a description of how to solve essential problems such as clock resets, does not include an effective operationalisation
of timed-bisimilarity checking, and, unfortunately, no tool implementation is provided.
Guha et al.~\cite{GuhaEtAl1} also utilize a zone-based approach for
bisimilarity-checking on timed automata as well as the weaker notion of timed prebisimilarity,
by employing so-called zone-valuation graphs and the notion of spans.
Again, the description of the approach lacks essential details about
the construction and no tool implementation is available.
Finally, Tanimoto et al.~\cite{Tanimoto2004} use timed bisimulation to
check whether a given behavioral abstraction preserves time-critical system behaviors,
without, however, showing how to check for timed bisimulation. 
The only currently available tool for checking timed bisimilarity we are aware of is called \textsc{Caal} \footnote{https://caal.cs.aau.dk/},
which is, however, inherently incomplete as it does not guarantee a correct representation of timed automata semantics.
To summarize, to the best of our knowledge, our approach is the only
effective construction of timed automata semantics based on zone graphs enabling
sound and complete bisimilarity checking for both deterministic and
non-deterministic timed (safety) automata.
In addition, we provide the only tool currently available for checking
timed bisimilarity for timed automata.
\paragraph{Outline.}
The remainder of this manuscript is organized as follows.
In Section~\ref{sec:background}, we first describe the necessary background on
timed automata, timed labeled transition systems, region graphs, zone graph semantics 
as well as timed bisimulation.
We also provide an example to illustrate potentials flaws of 
checking timed bisimilarity on plain zone graphs.
In Section~\ref{chap:virtual_clocks:main}, we describe our novel
construction of zone graph semantics using virtual clocks for
sound and complete timed bisimilarity checking.
Section~\ref{sec:CheckingForBisimulation} shows how we can use our results
of Section~\ref{chap:virtual_clocks:main} to implement an
algorithm to check for timed bisimilarity.
In Section~\ref{sec:evaluation} we evaluate the implementation, concerning computational effort of timed bisimilarity checking.
Finally, in Section~\ref{chap:conclusion} we conclude the manuscript
and provide an outlook on potential future work.

%% file: sections/background.tex
\section{Background}
\label{sec:background}

In this section, we revisit basic definitions of timed automata including 
syntax and semantics as well as timed bisimulation.
These notions build the foundation for the remainder of this manuscript.

\subsection{Timed Automata}

\textit{Timed automata (TA)} were introduced by Alur and Dill \cite{AlurDill}. 
In the following, we consider timed automata in the form of timed \emph{safety} automata 
according to Henzinger et al. \cite{Henzinger1994}.

Similar to classical finite automata, a TA consists 
of a finite graph whose nodes are called \emph{locations} 
and directed edges between nodes are called \emph{switches}. 
Switches are labeled by symbols from a finite alphabet $\Sigma$
of \emph{actions}.
In addition, a TA comprises a finite set $C$ of \emph{clocks}, representing
numerical variables defined over a time domain $\TimeDomain$. 
We limit our considerations to $\TimeDomain \in \{\mathbb{N}^{\geq 0}, \mathbb{R}^{\geq 0}\}$, where
$\TimeDomain = \mathbb{N}^{\geq 0}$ is used for \emph{discrete} time modeling, whereas
$\TimeDomain = \mathbb{R}^{\geq 0}$ represents \emph{dense} time. 
In both cases, the valuations of all clocks \emph{increase synchronously} over time
during a TA run, but are \emph{independently resettable} to zero. 
The current valuation of each clock measures 
the time elapsed since its last reset. 
By $\mathcal{B}(C)$, we denote the set of \emph{clock constraints} over set $C$.
Clock constraints are used in TA as \emph{location invariants} as well as \emph{switch guards}
to restrict the timing behavior of TA runs.
\begin{definition}[Timed Automaton]
\label{def:background:Timed-Automata:Timed-Automaton}
A \emph{timed automaton} (TA) $A$ is a tuple $(L, l_0, \Sigma, C, I, E)$, where $L$ is a finite set of \emph{locations}, with $l_0 \in L$ being an \emph{initial location}, $\Sigma$ is a finite set of (action) \emph{symbols}, $C$ is a finite set of \emph{clocks} such that $C \cap \Sigma = \emptyset$, $I: L \rightarrow \mathcal{B}(C)$ is a function assigning \emph{invariants} to locations, and $E \subseteq L \times \mathcal{B}(C) \times \Sigma \times 2^{C} \times L$ is a relation containing \emph{switches}.
The set $\mathcal{B}(C)$ of clock constraints $\phi$ over $C$ is inductively defined as
\begin{equation*}
\phi := \text{true} \ | \ c \sim n \ | \ c - c' \sim n \ | \ \phi \land \phi, \text{ with } n \in \mathbb{N}^{\geq 0}, \sim \ \in \{<, \leq, >, \geq\}, \text{ and } c, c' \in C.
\end{equation*}
We also introduce $c = n$ as shorthand notation for $c \leq n \land c \geq n$ and $(c-c') = n$ as shorthand notation for $(c-c') \leq n \land (c-c') \geq n$. Instead of $(l_1, g, \sigma, R, l_2) \in E$ we often denote $\TATrans{1}{g}{\sigma}{R}{2}$.
\end{definition}

\newcommand{\backgroundExampleAutomataScalingFactor}{0.35}
\newcommand{\backgroundExampleAutomataSpaceBetween}{20mm}
\newcommand{\backgroundTikzFontSize}{\Large}
\newcommand{\backgroundExampleTAArrowDesc}{\draw[-{Latex[length=3mm]}]}
\newcommand{\backgroundExampleTABendFactor}{80}
\newcommand{\backgroundExampleTAVerticalNodeDistance}{1.3cm}

\begin{example}
Figure~\ref{fig:examples} shows six different
TA, $A_{1}$ to $A_{6}$, defined over the same
alphabet $\Sigma = \{a,b,c\}$.
$A_{1}$ to $A_{4}$ consist of a set $L$ of three locations, labeled
$l_{0}$, $l_1$, and $l_2$, respectively, while $A_5$ and $A_6$ additionally contain the locations $l_1'$ and $l_2'$.
The set $C$ of clocks of the TA $A_{1}$, $A_{2}$, and $A_6$ comprise only one clock $x$, whereas TA 
$A_{3}$, $A_{4}$, and $A_5$ comprise two clocks, $x$ and $y$.
Switches $(\ell, g, \sigma, R, \ell')\in E$ are visually
depicted as arrows $\ell\xrightarrow{g, \sigma, R}\ell'$ leading from
source location $\ell\in L$ to target location $\ell'\in L$. The labels of transitions
consist of three components: guard $g\in\mathcal{B}(C)$, action $\sigma\in\Sigma$ and
reset set $R\subseteq C$, denoted as assignment $x:=0$ for each $x\in R$.
Guard $g$ is a clock constraint that must be satisfied for the switch to
be taken in a timed step to perform action $\sigma$ and to proceed from
location $\ell$ to location $\ell'$.
In addition, if the switch is taken, then all clocks $x\in R$ are set to zero.
Note that the trivial guard or invariant $\textit{true}$ and the empty reset set $\emptyset$
are omitted in the visual representation.
Similarly, locations $\ell$ may be also labeled with a clock constraint $I(\ell)$
denoting a location invariant (e.g. $I(l_1) = x < 2$ holds for any TA of Figure~\ref{fig:examples}).
$A_5$ and $A_6$ are non-deterministic TA, as the initial
location is the source location of two different 
switches labeled with the same action $a$, which are enabled at the same time but 
lead to different locations. 

\begin{figure}
\centering
\begin{subfigure}{0.1\textwidth}
\centering
\scalebox{\backgroundExampleAutomataScalingFactor}{
\begin{tikzpicture}
\tikzstyle{every node}=[font=\backgroundTikzFontSize]
\tikzstyle{state} = [draw,circle,minimum size=2cm,inner sep=0pt,semithick]

\node[state, align=center, initial above, initial text=] (0) {$l_0$};
\node[state, align=center, below = \backgroundExampleTAVerticalNodeDistance of 0] (1) {$l_1$\\$x<2$};
\node[state, align=center, below = \backgroundExampleTAVerticalNodeDistance of 1] (2) {$l_2$};
\backgroundExampleTAArrowDesc (0) --node[left, align=center, xshift=-0.2cm]{$a$\\$x\leq 0$} (1);
\backgroundExampleTAArrowDesc (1) --node[left, align=center]{$b$} (2);
\backgroundExampleTAArrowDesc (2) edge [bend right = \backgroundExampleTABendFactor] node[left, align=center] {$c$\\$x>3$\\$x:=0$} (0);
\end{tikzpicture}
}
    \caption{$A_{1}$}
    \label{fig:exampleA1}
\end{subfigure}
\hspace{\backgroundExampleAutomataSpaceBetween}
\begin{subfigure}{0.1\textwidth}
\centering
\scalebox{\backgroundExampleAutomataScalingFactor}{
\begin{tikzpicture}
\tikzstyle{every node}=[font=\backgroundTikzFontSize]
\tikzstyle{state} = [draw,circle,minimum size=2cm,inner sep=0pt,semithick]

\node[state, align=center, initial above, initial text=] (0) {$l_0$};
\node[state, align=center, below = \backgroundExampleTAVerticalNodeDistance of 0] (1) {$l_1$\\$x<2$};
\node[state, align=center, below = \backgroundExampleTAVerticalNodeDistance of 1] (2) {$l_2$};
\backgroundExampleTAArrowDesc (0) --node[left, align=center, xshift=-0.2cm]{$a$\\$x:=0$} (1);
\backgroundExampleTAArrowDesc (1) --node[left, align=center]{$b$} (2);
\backgroundExampleTAArrowDesc (2) edge [bend right = \backgroundExampleTABendFactor] node[left, align=center] {$c$\\$x>3$\\$x:=0$} (0);
\end{tikzpicture}
}
    \caption{$A_{2}$}
    \label{fig:exampleA2}
\end{subfigure}
\hspace{\backgroundExampleAutomataSpaceBetween}
\begin{subfigure}{0.1\textwidth}
\centering
\scalebox{\backgroundExampleAutomataScalingFactor}{
\begin{tikzpicture}
\tikzstyle{every node}=[font=\backgroundTikzFontSize]
\tikzstyle{state} = [draw,circle,minimum size=2cm,inner sep=0pt,semithick]

\node[state, align=center, initial above, initial text=] (0) {$l_0$};
\node[state, align=center, below = \backgroundExampleTAVerticalNodeDistance of 0] (1) {$l_1$\\$x<2$};
\node[state, align=center, below = \backgroundExampleTAVerticalNodeDistance of 1] (2) {$l_2$};
\backgroundExampleTAArrowDesc (0) --node[left, xshift=-.3cm, align=center]{$a$\\$x, y:=0$} (1);
\backgroundExampleTAArrowDesc (1) --node[left, align=center]{$b$} (2);
\backgroundExampleTAArrowDesc (2) edge [bend right = \backgroundExampleTABendFactor] node[left, align=center] {$c$\\$y>3$\\$x,y:=0$} (0);
\end{tikzpicture}
}
    \caption{$A_{3}$}
    \label{fig:exampleA3}
\end{subfigure}
\hspace{\backgroundExampleAutomataSpaceBetween}
\begin{subfigure}{0.1\textwidth}
\centering
\scalebox{\backgroundExampleAutomataScalingFactor}{
\begin{tikzpicture}
\tikzstyle{every node}=[font=\backgroundTikzFontSize]
\tikzstyle{state} = [draw,circle,minimum size=2cm,inner sep=0pt,semithick]

\node[state, align=center, initial above, initial text=] (0) {$l_0$};
\node[state, align=center, below = \backgroundExampleTAVerticalNodeDistance of 0] (1) {$l_1$\\$x<2$};
\node[state, align=center, below = \backgroundExampleTAVerticalNodeDistance of 1] (2) {$l_2$};
\backgroundExampleTAArrowDesc (0) --node[left, align=center, xshift=-0.2cm]{$a$\\$x:=0$} (1);
\backgroundExampleTAArrowDesc (1) --node[left, align=center, xshift=-0.2cm]{$b$\\$y:=0$} (2);
\backgroundExampleTAArrowDesc (2) edge [bend right = \backgroundExampleTABendFactor] node[left, align=center] {$c$\\$y>3$\\$x,y:=0$} (0);
\end{tikzpicture}
}
    \caption{$A_{4}$}
    \label{fig:exampleA4}
\end{subfigure}
\begin{subfigure}{0.3\textwidth}
\centering
\scalebox{\backgroundExampleAutomataScalingFactor}{
\begin{tikzpicture}
\tikzstyle{every node}=[font=\backgroundTikzFontSize]
\tikzstyle{state} = [draw,circle,minimum size=2cm,inner sep=0pt,semithick]

\node[state, align=center, initial above, initial text=] (0) {$l_0$};
\node[state, align=center, below  = \backgroundExampleTAVerticalNodeDistance of 0] (1) {$l_1$\\$x<2$};
\node[state, align=center, below = \backgroundExampleTAVerticalNodeDistance of 1] (2) {$l_2$};
\node[state, align=center, left = \backgroundExampleTAVerticalNodeDistance of 1] (3) {$l_1'$\\$x<2$};
\node[state, align=center, left = \backgroundExampleTAVerticalNodeDistance of 2] (4) {$l_2'$};
\backgroundExampleTAArrowDesc (0) --node[left, align=center]{$a$\\$x:=0$} (1);
\backgroundExampleTAArrowDesc (1) --node[left, align=center]{$b$\\$y:=0$} (2);
\backgroundExampleTAArrowDesc (2) edge [bend right = 50] node[right, align=center] {$c$\\$y>3$\\$x,y:=0$} (0);
\backgroundExampleTAArrowDesc (0) --node[left, xshift=-.3cm, yshift=.5cm, align=center]{$a$\\$x, y:=0$} (3);
\backgroundExampleTAArrowDesc (3) --node[left, align=center]{$b$} (4);
\backgroundExampleTAArrowDesc (4) edge [bend left = 100] node[left, yshift=.9cm, align=center] {$c$\\$y>3$\\$x, y:=0$} (0);
\end{tikzpicture}
}
    \caption{$A_5$}
    \label{fig:nondet-example}
\end{subfigure}
\begin{subfigure}{0.3\textwidth}
\centering
\scalebox{\backgroundExampleAutomataScalingFactor}{
\begin{tikzpicture}
\tikzstyle{every node}=[font=\backgroundTikzFontSize]
\tikzstyle{state} = [draw,circle,minimum size=2cm,inner sep=0pt,semithick]

\node[state, align=center, initial above, initial text=] (0) {$l_0$};
\node[state, align=center, below right  = \backgroundExampleTAVerticalNodeDistance and 1cm of 0] (1) {$l_1$\\$x<2$};
\node[state, align=center, below = \backgroundExampleTAVerticalNodeDistance of 1] (2) {$l_2$};
\node[state, align=center, below left = \backgroundExampleTAVerticalNodeDistance and 1cm of 0] (3) {$l_1'$\\$x<2$};
\node[state, align=center, below = \backgroundExampleTAVerticalNodeDistance of 3] (4) {$l_2'$};
\backgroundExampleTAArrowDesc (0) --node[right, xshift=.3cm, yshift=.5cm, align=center]{$a$\\$x \leq 1$\\$x:=0$} (1);
\backgroundExampleTAArrowDesc (1) --node[left, align=center]{$b$} (2);
\backgroundExampleTAArrowDesc (2) edge [bend right = 90] node[right, xshift=.5cm, align=center] {$c$\\$x>3$\\$x:=0$} (0);
\backgroundExampleTAArrowDesc (0) --node[left, xshift=-.3cm, yshift=.5cm, align=center]{$a$\\$x \geq 1$\\$x:=0$} (3);
\backgroundExampleTAArrowDesc (3) --node[left, align=center]{$b$} (4);
\backgroundExampleTAArrowDesc (4) edge [bend left = 90] node[left, xshift=-.5cm, align=center] {$c$\\$x>3$\\$x:=0$} (0);
\end{tikzpicture}
}
    \caption{$A_6$}
    \label{fig:nondet-bisim-example}
\end{subfigure}
\caption{Examples of Timed Automata}
\label{fig:examples}
\end{figure}
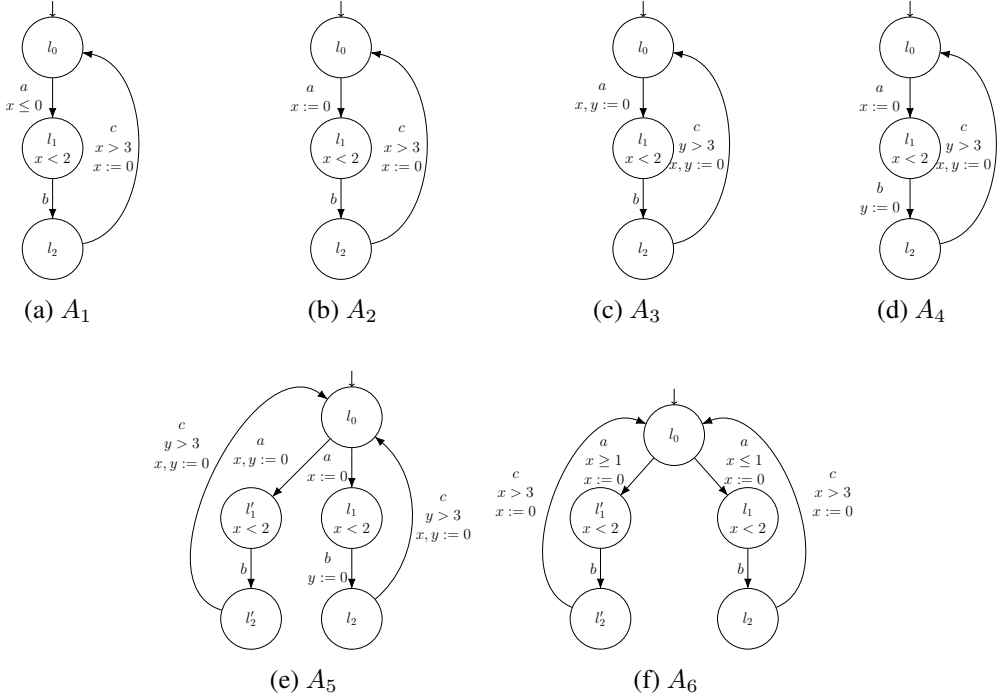
\end{example}

We assume a TA to be \emph{diagonal-free} thus 
containing only atomic clock constraints of the form $c \sim n$. 
This assumption eases many of the following proofs and 
does not constitute a restriction as for every non-diagonal-free TA, a language-equivalent 
diagonal-free TA can be constructed~\cite{Berard1998}.
Clock constraints occurring in a TA are evaluated
by means of clock valuations $u: C \rightarrow \TimeDomain$. 

\begin{definition}[Clock Valuation]
\label{definition:background:clock-valuation}
Let $C$ be a set of clocks over time domain $\TimeDomain$. 
A \emph{clock valuation} $u: C \rightarrow \TimeDomain$ assigns a value $u(c)$ 
to each clock $c \in C$. For a value $d \in \TimeDomain$, the clock valuation $\ClockValuation{} + d$ assigns to each clock $c$ 
the value $\ClockValuation{}(c)+d$ and for a set of clocks $R \subseteq C$, the clock valuation $\ClockValuation{}' = [R \rightarrow 0]\ClockValuation{}$ 
assigns to each clock $c \in R$ the value $\ClockValuation{}'(c) = 0$ and for each clock $c' \in C$ with $c' \not\in R$ 
the value $\ClockValuation{}'(c') = \ClockValuation{}(c')$. A clock valuation $\ClockValuation{}$ satisfies a clock constraint $\phi\in\mathcal{B}(C)$, denoted $\ClockValuation{} \models \phi$, if the replacement of all occurrences of any clock $c \in C$ in $\phi$ by its value $\ClockValuation{}(c)$ yields a true statement.

\end{definition}

If a clock valuation $\ClockValuation{}$ does not satisfy a clock constraint $\phi$, we denote $\ClockValuation{} \not\models \phi$. We call a clock constraint $\phi$ to be \emph{stronger} than a clock constraint $\phi'$, 
if and only if $\ClockValuation{} \models \phi$ implies $\ClockValuation{} \models \phi'$ and there exists a $\ClockValuation{}'$ with $\ClockValuation{}' \models \phi'$ and $\ClockValuation{} \not\models \phi$. 
To rule out trivial cases, we assume the location invariant $I(l_0)$ to be always satisfied 
by the initial clock valuation $[C \rightarrow 0]$. 

\subsection{Timed Labeled Transition Systems}

\textit{Labeled Transition Systems (LTS)} are used to 
model processes with finite or infinite numbers of states.
To capture the operational semantics of TA, \textit{Timed Labeled Transition Systems (TLTS)} extend LTS by the notion
of time.
A state of a TLTS is a pair $(l, u)$, with $l \in L$ being the currently active location and 
$u: C \rightarrow \TimeDomain$ being the current valuation of clocks. 
TLTS contain two types of transitions, one corresponding to the passage of time of duration $d \in \TimeDomain$ 
without any occurrence of a switch, leading to an updated clock valuation $u' = u+d$, 
and one corresponding to the instantaneous execution of a switch $\TATrans{1}{g}{\sigma}{R}{2}$, leading to the state $(l_2, [R \rightarrow 0] u)$. 

\begin{definition}[TLTS]
\label{def:background:TLTS}
Let $A = (l, l_0, \Sigma, C, I, E)$ be a TA. 
The \emph{TLTS} $\TLTS{A}$ of $A$ is a tuple $(\TLTSAllStates{}, \TLTSState{0, A}, \Sigma \cup \TimeDomain, \TLTSTrans{})$, where $\TLTSAllStates{}= L \times (C \rightarrow \TimeDomain)$ is a set of states with $\TLTSState{0, A} = (l_0, [C \rightarrow 0]) \in \TLTSAllStates{}$ being the initial state, $\Sigma \cup \TimeDomain$ is the set of transition labels, where $\Sigma \cap \TimeDomain = \emptyset$, $\TLTSTrans{} \ \subseteq \TLTSAllStates{} \times \Sigma \cup \TimeDomain \times \TLTSAllStates{}$ is a set of transitions, which is the least relation satisfying
\begin{itemize}
\item $\TLTSFullState{}{} \TLTSTrans{d} (l, \ClockValuation{}+d)$ if $(\ClockValuation{}+d) \models I(l)$ for $d \in \TimeDomain$, and
\item $\TLTSFullState{1}{} \TLTSTrans{\sigma} \TLTSFullState{2}{}$ if $\TATrans{1}{g}{\sigma}{R}{2}$ with $\ClockValuation{1} \models g$, $\ClockValuation{2} = [R \rightarrow 0]\ClockValuation{1}$, $\ClockValuation{2} \models I(l_2)$, and $\sigma \in \Sigma$. 
\end{itemize}
\end{definition}

Since we assume $[C \rightarrow 0] \models I(l_0)$ and we check for each TLTS transition 
if the clock valuation of the target state satisfies the invariant of the corresponding
target location, $u \models I(l)$ holds for all reachable states $\TLTSFullState{}{}$. 

\begin{example}
\newcommand{\backgroundExTLTSDistance}{3}
\newcommand{\backgroundExTLTSFourVertDistance}{3}
\newcommand{\backgroundExTLTSFourHorDistance}{3}
\newcommand{\backgroundExTLTSArrowDesc}{\draw[-{Stealth[length=2mm]}]}
\begin{figure}
\centering
\begin{subfigure}{0.3\textwidth}
\centering
\scalebox{\backgroundExampleAutomataScalingFactor}{
\begin{tikzpicture}
\tikzstyle{every node}=[font=\backgroundTikzFontSize]
\node (l00) [rectangle, initial, initial text = {}, align=center] at (0, 1) {$l_0$\\$x = 0$};
\node (l01) [rectangle, align=center] at (\backgroundExTLTSDistance, 1) {$l_0$\\$x = 1$};
\node (l02) [rectangle, align=center] at (1.75*\backgroundExTLTSDistance, 1) {$l_0$\\$x = 2$};
\node (l03) [rectangle, align=center] at (2.5*\backgroundExTLTSDistance, 1) {...};
\node (l10) [rectangle, align=center] at (0, -\backgroundExTLTSDistance) {$l_1$\\$x = 0$};
\node (l11) [rectangle, align=center] at (\backgroundExTLTSDistance, -\backgroundExTLTSDistance) {$l_1$\\$x = 1$};
\node (l20) [rectangle, align=center] at (0, -2*\backgroundExTLTSDistance) {$l_2$\\$x = 0$};
\node (l21) [rectangle, align=center] at (\backgroundExTLTSDistance, -2*\backgroundExTLTSDistance) {$l_2$\\$x = 1$};
\node (l22) [rectangle, align=center] at (1.75*\backgroundExTLTSDistance, -2*\backgroundExTLTSDistance) {$l_2$\\$x = 2$};
\node (l23) [rectangle, align=center] at (2.5*\backgroundExTLTSDistance, -2*\backgroundExTLTSDistance) {$l_2$\\$x = 3$};
\node (l24) [rectangle, align=center] at (3.25*\backgroundExTLTSDistance, -2*\backgroundExTLTSDistance) {...};
\backgroundExTLTSArrowDesc (l00) -- (l01) node[anchor=south, midway]{1};
\backgroundExTLTSArrowDesc (l01) -- (l02) node[anchor=south, midway]{1};
\backgroundExTLTSArrowDesc (l02) -- (l03) node[anchor=south, midway]{1};
\backgroundExTLTSArrowDesc (l00) -- (l10) node[anchor=east, midway]{a};
\backgroundExTLTSArrowDesc (l10) -- (l11) node[anchor=south, midway]{1};
\backgroundExTLTSArrowDesc (l10) -- (l20) node[anchor=east, midway]{b};
\backgroundExTLTSArrowDesc (l11) -- (l21) node[anchor=east, midway]{b};
\backgroundExTLTSArrowDesc (l20) -- (l21) node[anchor=south, midway]{1};
\backgroundExTLTSArrowDesc (l21) -- (l22) node[anchor=south, midway]{1};
\backgroundExTLTSArrowDesc (l22) -- (l23) node[anchor=south, midway]{1};
\backgroundExTLTSArrowDesc (l23) -- (l24) node[anchor=south, midway]{1};
\draw[-] (l24) -- (3.25*\backgroundExTLTSDistance, 2) node[anchor=east, midway]{c};
\draw[-] (3.25*\backgroundExTLTSDistance, 2) -- (0, 2);
\backgroundExTLTSArrowDesc (0, 2) -- (l00);
\end{tikzpicture}
}
\caption{$A_{\text{TLTS}, 1}$}
\label{subfigure:background:TLTS:example:A-1}
\end{subfigure}
\begin{subfigure}{0.3\textwidth}
\centering
\scalebox{\backgroundExampleAutomataScalingFactor}{
\begin{tikzpicture}
\tikzstyle{every node}=[font=\backgroundTikzFontSize]
\node (l00) [rectangle, initial, initial text = {}, align=center] at (0, 1) {$l_0$\\$x = 0$\\$y=0$};
\node (l01) [rectangle, align=center] at (\backgroundExTLTSDistance, 1) {$l_0$\\$x = 1$\\$y=1$};
\node (l02) [rectangle, align=center] at (1.75*\backgroundExTLTSDistance, 1) {$l_0$\\$x = 2$\\$y=2$};
\node (l03) [rectangle, align=center] at (2.5*\backgroundExTLTSDistance, 1) {...};
\node (l10) [rectangle, align=center] at (0, -\backgroundExTLTSDistance) {$l_1$\\$x = 0$\\$y=0$};
\node (l11) [rectangle, align=center] at (\backgroundExTLTSDistance, -\backgroundExTLTSDistance) {$l_1$\\$x = 1$\\$y=1$};
\node (l20) [rectangle, align=center] at (0, -2*\backgroundExTLTSDistance) {$l_2$\\$x=0$\\$y = 0$};
\node (l21) [rectangle, align=center] at (\backgroundExTLTSDistance, -2*\backgroundExTLTSDistance) {$l_2$\\$x=1$\\$y = 1$};
\node (l22) [rectangle, align=center] at (1.75*\backgroundExTLTSDistance, -2*\backgroundExTLTSDistance) {$l_2$\\$x=2$\\$y = 2$};
\node (l23) [rectangle, align=center] at (2.5*\backgroundExTLTSDistance, -2*\backgroundExTLTSDistance) {$l_2$\\$x=3$\\$y = 3$};
\node (l24) [rectangle, align=center] at (3.25*\backgroundExTLTSDistance, -2*\backgroundExTLTSDistance) {...};
\backgroundExTLTSArrowDesc (l00) -- (l01) node[anchor=south, midway]{1};
\backgroundExTLTSArrowDesc (l01) -- (l02) node[anchor=south, midway]{1};
\backgroundExTLTSArrowDesc (l02) -- (l03) node[anchor=south, midway]{1};
\backgroundExTLTSArrowDesc (l00) -- (l10) node[anchor=east, midway]{a};
\backgroundExTLTSArrowDesc (l01) -- (l10) node[anchor=east, midway]{a};
\backgroundExTLTSArrowDesc (l02) -- (l10) node[anchor=south, midway]{a};
\backgroundExTLTSArrowDesc (l03) -- (l10) node[anchor=south, midway]{a};
\backgroundExTLTSArrowDesc (l10) -- (l11) node[anchor=north, midway]{1};
\backgroundExTLTSArrowDesc(l10) -- (l20) node[anchor=east, midway]{b};
\backgroundExTLTSArrowDesc (l11) -- (l21) node[anchor=east, midway]{b};
\backgroundExTLTSArrowDesc (l20) -- (l21) node[anchor=south, midway]{1};
\backgroundExTLTSArrowDesc (l21) -- (l22) node[anchor=south, midway]{1};
\backgroundExTLTSArrowDesc (l22) -- (l23) node[anchor=south, midway]{1};
\backgroundExTLTSArrowDesc (l23) -- (l24) node[anchor=south, midway]{1};
\draw[-] (l24) -- (3.25*\backgroundExTLTSDistance, 2.5) node[anchor=east, midway]{c};
\draw[-] (3.25*\backgroundExTLTSDistance, 2.5) -- (0, 2.5);
\backgroundExTLTSArrowDesc (0, 2.5) -- (l00);
\end{tikzpicture}
}
\caption{$A_{\text{TLTS}, 3}$}
\label{subfigure:background:TLTS:example:A-3}
\end{subfigure}
\begin{subfigure}{0.33\textwidth}
\centering
\scalebox{\backgroundExampleAutomataScalingFactor}{
\begin{tikzpicture}
\tikzstyle{every node}=[font=\backgroundTikzFontSize]
\node (l00) [rectangle, initial, initial text = {}, align=center] at (0, 0) {$l_0$\\$x, y = 0$};
\node (l01) [rectangle, align=center] at (\backgroundExTLTSDistance, 0) {$l_0$\\$x, y = 1$};
\node (l02) [rectangle, align=center] at (2*\backgroundExTLTSDistance, 0) {$l_0$\\$x, y = 2$};
\node (l02Follow) [rectangle, align=center] at (2*\backgroundExTLTSDistance, -0.5*\backgroundExTLTSFourVertDistance) {...};
\node (l03) [rectangle, align=center] at (3*\backgroundExTLTSDistance, 0) {...};
\node (l10Equal) [rectangle, align=center] at (0, -1.5*\backgroundExTLTSFourVertDistance) {$l_1$\\$x, y = 0$};
\node (l11Equal) [rectangle, align=center] at (2*\backgroundExTLTSDistance, -1.5*\backgroundExTLTSFourVertDistance) {$l_1$\\$x, y = 1$};
\node (l10NotEqual) [rectangle, align=center] at (\backgroundExTLTSDistance, -1*\backgroundExTLTSFourVertDistance) {$l_1$\\$x = 0$\\$y = 1$};
\node (l11NotEqual) [rectangle, align=center] at (3*\backgroundExTLTSDistance, -1*\backgroundExTLTSFourVertDistance) {$l_1$\\$x = 1$\\$y = 2$};
\node (l20Equal) [rectangle, align=center] at (\backgroundExTLTSDistance, -2.5*\backgroundExTLTSFourVertDistance) {$l_2$\\$x, y = 0$};
\node (l21Equal) [rectangle, align=center] at (\backgroundExTLTSDistance, -3*\backgroundExTLTSFourVertDistance) {...};
\node (l22Equal) [rectangle, align=center] at (\backgroundExTLTSDistance, -3.5*\backgroundExTLTSFourVertDistance) {$l_2$\\$x, y = 4$};
\node (l23Equal) [rectangle, align=center] at (0, -3.5*\backgroundExTLTSDistance) {...};
\node (l20NotEqual) [rectangle, align=center] at (3*\backgroundExTLTSDistance, -2*\backgroundExTLTSFourVertDistance) {$l_2$\\$x = 1$\\$y = 0$};
\node (l21NotEqual) [rectangle, align=center] at (3*\backgroundExTLTSDistance, -3*\backgroundExTLTSFourVertDistance) {...};
\node (l22NotEqual) [rectangle, align=center] at (2*\backgroundExTLTSDistance, -3*\backgroundExTLTSFourVertDistance) {$l_2$\\$x = 5$\\$y = 4$};
\node (l23NotEqual) [rectangle, align=center] at (2*\backgroundExTLTSDistance, -2*\backgroundExTLTSFourVertDistance) {...};
\backgroundExTLTSArrowDesc (l00) -- (l01) node[anchor=south, midway]{1};
\backgroundExTLTSArrowDesc (l01) -- (l02) node[anchor=south, midway]{1};
\backgroundExTLTSArrowDesc (l02) -- (l03) node[anchor=south, midway]{1};
\backgroundExTLTSArrowDesc (l02) -- (l02Follow) node[anchor=east, midway]{a};
%
%
\backgroundExTLTSArrowDesc (l00) -- (l10Equal) node[anchor=east, midway]{a};
\backgroundExTLTSArrowDesc (l10Equal) -- (l11Equal) node[anchor=south, near start]{1};
\backgroundExTLTSArrowDesc (l10Equal) -- (l20Equal) node[anchor=east, midway]{b};
\backgroundExTLTSArrowDesc (l11Equal) -- (l20NotEqual) node[anchor=east, midway, yshift=-0.2cm]{b};
%
%
\backgroundExTLTSArrowDesc (l01) -- (l10NotEqual) node[anchor=east, midway]{a};
\backgroundExTLTSArrowDesc (l10NotEqual) -- (l11NotEqual) node[anchor=south, midway]{1};
\backgroundExTLTSArrowDesc (l10NotEqual) -- (l20Equal) node[anchor=east, midway]{b};
\backgroundExTLTSArrowDesc (l11NotEqual) -- (l20NotEqual) node[anchor=east, midway]{b};
%
%
\backgroundExTLTSArrowDesc (l20Equal) -- (l21Equal) node[anchor=east, midway]{1};
\backgroundExTLTSArrowDesc (l21Equal) -- (l22Equal) node[anchor=east, midway]{1};
\backgroundExTLTSArrowDesc (l20NotEqual) -- (l21NotEqual) node[anchor=east, midway]{1};
\backgroundExTLTSArrowDesc (l21NotEqual) -- (l22NotEqual) node[anchor=south, midway]{1};
\draw[-] (2*\backgroundExTLTSDistance + 0.7, -3*\backgroundExTLTSFourVertDistance - 0.5) -- (3.75*\backgroundExTLTSDistance, -3*\backgroundExTLTSFourVertDistance - 0.5) node[anchor=south, near end]{c};
\draw[-] (3.75*\backgroundExTLTSDistance, -3.5*\backgroundExTLTSFourVertDistance) -- (3.75*\backgroundExTLTSDistance, 1);
\draw[-] (3.75*\backgroundExTLTSDistance, 1) -- (0, 1);
\backgroundExTLTSArrowDesc (0, 1) -- (l00);
\draw[-] (l22Equal.east) -- (3.75*\backgroundExTLTSDistance, -3.5*\backgroundExTLTSFourVertDistance) node[anchor=south, midway]{c};
\draw[-] (3.75*\backgroundExTLTSDistance, -3.5*\backgroundExTLTSFourVertDistance) edge (3.75*\backgroundExTLTSDistance, -3.5*\backgroundExTLTSFourVertDistance);
\backgroundExTLTSArrowDesc (l22NotEqual) -- (l23NotEqual) node[anchor=east, midway]{1};
\backgroundExTLTSArrowDesc (l22Equal) -- (l23Equal) node[anchor=south, midway]{1};
\end{tikzpicture}
}
\caption{$A_{\text{TLTS}, 4}$}
\label{subfigure:background:TLTS:example:A-4}
\end{subfigure}
\caption{The TLTS of some TA of Figure \ref{fig:examples}}
\label{figure:background:TLTS:example}
\end{figure}
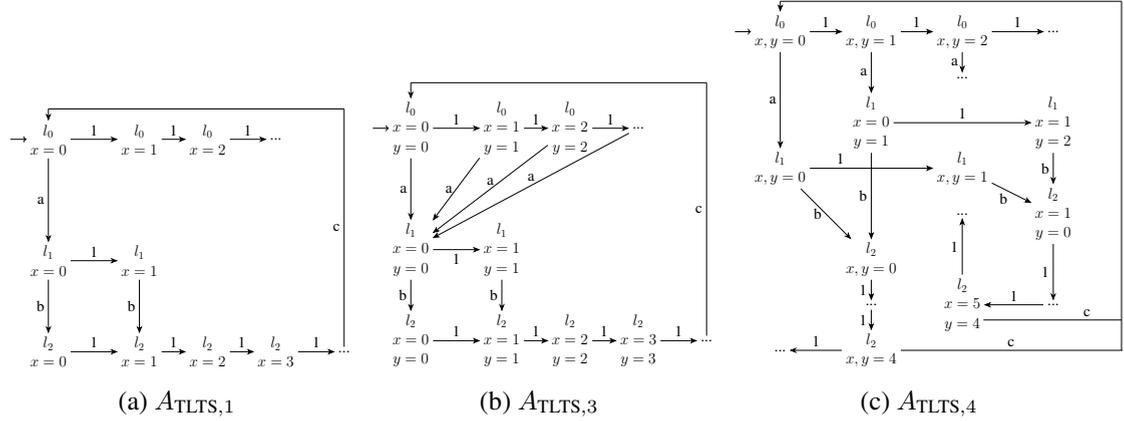

Figure \ref{figure:background:TLTS:example} shows the TLTS of $A_{1}$ (Figure \ref{fig:exampleA1}), $A_{3}$ (Figure \ref{fig:exampleA3}), and $A_{4}$ (Figure \ref{fig:exampleA4}). For reasons of readability, all delay transitions with the exception of $d=1$ have been removed. In the discrete-time model, the states can have an infinite number of outgoing transitions labeled with natural numbers. In the dense-time model the number of outgoing transitions is uncountably infinite and the transitions are labeled with real numbers.

The TLTS of $A_2$ (Figure \ref{fig:exampleA2}) is identical to $A_{\text{TLTS}, 3}$, with the exception of the absence of clock $y$. The TLTS of $A_5$ (Figure \ref{fig:nondet-example}) is a combination of $A_{\text{TLTS}, 3}$ and $A_{\text{TLTS}, 4}$, while the TLTS of $A_6$ (Figure \ref{fig:nondet-bisim-example}) is two times the TLTS of $A_2$ (the first part is reachable from $l_0$ if the guard $x \leq 1$ is satisfied and the second part is reachable from $l_0$ if the guard $x \geq 1$ is satisfied). Time delays are always enabled, unless a location has an invariant that is no longer satisfied after this delay. Action-based transitions are enabled if and only if there exists a corresponding switch $\ell \TATransShort{g}{\sigma}{R}\ell'$ in the TA, for which the clock valuation $\ClockValuation{}$ of the current state fulfills $g$ and the clock valuation $[R \rightarrow 0]\ClockValuation{}$, which is the clock valuation of the target state, fulfills the invariant of $\ell'$.
\end{example}

A fundamental decision problem in the context of
behavioral modeling is to compare the behavior
of two models against each other (e.g.,
one representing the specification and the other one
constituting a candidate implementation).
The most obvious approach is to compare the
languages characterized by both models against each other.
Unfortunately, it is well-known that timed language inclusion is undecidable for timed automata~\cite{AlurDill90AutomataForModelingRealTimeSystems} 
and this result can easily be extended to show that language equivalence is undecidable, too. 

In contrast, the stronger notion of timed bisimulation of two TA models is decidable~\cite{Cerans1992}.

\begin{definition}[Timed Bisimulation \cite{Lynch1992, Cerans1992}]
\label{def:background:sim-and-bisim:Strong-Timed-Bisimulation}
Let $A=(L_A, l_{0, A}, \Sigma, C_A, I_A, E_A)$ and $B=(L_B, l_{0, B}, \Sigma, C_B, I_B, E_B)$ be TA over the same alphabet $\Sigma$ with $C_A \cap C_B = \emptyset$. $A$ has the corresponding TLTS $\TLTS{A} = (\TLTSAllStates{A}, \TLTSState{0, A}, \Sigma \cup \TimeDomain, \TLTSTrans{})$ and $B$ has the corresponding TLTS $\TLTS{B} = (\TLTSAllStates{B}, \TLTSState{0, B}, \Sigma \cup \TimeDomain, \TLTSTrans{})$.
A Relation $R \subseteq \TLTSAllStates{B} \times \TLTSAllStates{A}$ is called a \textit{Timed Bisimulation}, if and only if for all pairs $(\TLTSState{1, B}, \TLTSState{1, A}) \in R$ it holds that 
\begin{enumerate}
\item if there is a transition $\TLTSState{1, B} \TLTSTrans{\mu} \TLTSState{2, B}$ for $\mu \in \Sigma \cup \TimeDomain$, then a transition $\TLTSState{1, A} \TLTSTrans{\mu} \TLTSState{2, A}$ with $(\TLTSState{2, B}, \TLTSState{2, A}) \in R$, exists, and
\item if there is a transition $\TLTSState{1, A} \TLTSTrans{\mu} \TLTSState{2, A}$ for $\mu \in \Sigma \cup \TimeDomain$, then a transition $\TLTSState{1, B} \TLTSTrans{\mu} \TLTSState{2, B}$ with $(\TLTSState{2, B}, \TLTSState{2, A}) \in R$, exists.
\end{enumerate}
\end{definition}


$A$ is timed bisimilar to $B$, denoted $B \sim A$, if and only if there exists a timed bisimulation $R$ with $(\TLTSState{0, B}, \TLTSState{0, A}) \in R$. In the following, we write that 
$A$ and $B$ are bisimilar, if and only if $A$ and $B$ are timed bisimilar. 
It is well-known that (timed) bisimulation is an equivalence relation.

\begin{example}
\label{ex:background:timed-bisim}
We compare the previously described examples with respect to timed bisimulation.

Due to the state $(l_0, x=1)$, which does not have an outgoing transition labeled with $a$, $A_{\text{TLTS}, 1}$ (Figure~\ref{subfigure:background:TLTS:example:A-1}) is not bisimilar to any other TLTS shown. For instance, if we compare $A_{\text{TLTS}, 1}$ and $A_{\text{TLTS}, 3}$ (Figure~\ref{subfigure:background:TLTS:example:A-3}), we find that if we use the transition labeled with 1, the reached states are not bisimilar since $A_{\text{TLTS}, 3}$ has an outgoing transition labeled with $a$ here and $A_{\text{TLTS}, 1}$ does not. 

Since $A_{\text{TLTS}, 2}$ and $A_{\text{TLTS}, 3}$ have an isomorphic structure, they are bisimilar, despite the fact that the sets of clocks are not equivalent. Due to the fact that bisimulation is an equivalence relation, we know that $A_{\text{TLTS}, 2}$ is bisimilar to any other TLTS if and only if that TLTS is bisimilar to $A_{\text{TLTS}, 3}$. Therefore, we proceed with $A_{\text{TLTS}, 3}$. 

To compare $A_{\text{TLTS}, 3}$ and $A_{\text{TLTS}, 4}$ (Figure~\ref{subfigure:background:TLTS:example:A-4}), we consider the sequence $[a, 1, b, 3]$. In $A_{\text{TLTS}, 3}$, this sequence results into the state $(l_2, (x=4, y=4))$, which has an outgoing transition labeled with $c$. In $A_{\text{TLTS}, 4}$, this sequence results into the state $(l_2, (x=4, y=3))$, which does not have such an outgoing transition. Therefore, the TLTS are not bisimilar.

$A_{\text{TLTS}, 4}$ and $A_{\text{TLTS}, 5}$ are not bisimilar since $A_{\text{TLTS}, 4}$ is not bisimilar to $A_{\text{TLTS}, 3}$, which implies that the corresponding branch of $A_{\text{TLTS}, 5}$ is also not bisimilar to $A_{\text{TLTS}, 4}$.

The comparison of $A_{\text{TLTS}, 3}$ and $A_{\text{TLTS}, 6}$ shows that a non-deterministic TA can be timed bisimilar to a deterministic TA. Since $A_{\text{TLTS}, 6}$ is essentially two times $A_{\text{TLTS}, 2}$ we already know that both individual parts of $A_{\text{TLTS}, 6}$ have an identical structure to $A_{\text{TLTS}, 3}$ and we only have to check the state $(l_0, x = 1)$, which has two outgoing transitions labeled with a in $A_{\text{TLTS}, 6}$. The first one leads to $(l_1, x=0)$. Since the following states have an identical structure to $A_{\text{TLTS}, 3}$, we conclude that the state $(l_1, x=0)$, reachable in $A_{\text{TLTS}, 3}$ from $(l_0, x = 1)$ via a transition labeled with a, is bisimilar to the state of $A_{\text{TLTS}, 6}$. The same holds for the other outgoing transition in $A_{\text{TLTS}, 6}$, which leads to $(l_1', x=0)$. Since the following structure is the same as in $A_{\text{TLTS}, 3}$, we conclude that $(l_1, x=0)$ from $A_{\text{TLTS}, 3}$ is bisimilar to $(l_1', x=0)$ from $A_{\text{TLTS}, 6}$ and, therefore, the states $(l_0, x=1)$ from both TLTS are bisimilar.
\end{example}

Although bisimulation of two TA is defined
with respect to their TLTS models, they do not provide
a feasible basis for effectively checking for bisimulation due to their infinite number of states.
In the following we revisit two alternative finite representations 
of timed automata semantics, namely region graphs and zone graphs.

\subsection{Region Graphs}

Region graphs were introduced by Alur and Dill \cite{ALUR1994183}. The main idea is based on the observation that if two clock valuations agree on the integer part of all clock values then every constraint of a diagonal-free TA is satisfied either by both of the clock valuations or none of them (remember: Clock constraints use numbers from $\mathbb{N}^{\geq 0}$ only). Furthermore, if the order of the clocks in which the integer part changes is equal for two clock valuations (i.e., both will lead to similar behavior in the future), the states can be considered to be in the same equivalence class, called \textit{region}.

\begin{example}
\label{ex:background:region-graphs}
Consider TLTS $A_{\text{TLTS}, 4}$ (Figure \ref{subfigure:background:TLTS:example:A-4}) and the sequences $[0.5, a, 0.6]$, $[0.9, a, 0.5]$, $[0.9, a, \allowbreak 1.3]$, and $[0, a, 0]$ (we assume dense-time model here). All these sequences describe a delay, followed by the input $a$ and another delay. The sequences all lead to location $l_1$, but with different clock valuations: $\ClockValuation{1}$, with $\ClockValuation{1}(x) = 0.6$ and $\ClockValuation{1}(y) = 1.1$, $\ClockValuation{2}$, with $\ClockValuation{2}(x) = 0.5$ and $\ClockValuation{2}(y) = 1.4$, $\ClockValuation{3}$, with $\ClockValuation{3}(x) =1.3$ and $\ClockValuation{3}(y) = 2.2$, and $\ClockValuation{4}$, with $\ClockValuation{4}(x) = 0$ and $\ClockValuation{4}(y) = 0$.

$\ClockValuation{1}(x)$ and $\ClockValuation{2}(x)$ have the integer part zero, $\ClockValuation{1}(y)$ and $\ClockValuation{2}(y)$ have the integer part one, and the fractional part of the value of $x$ is both times larger than the fractional part of the value of $y$ and, therefore, clock $x$ changes its integer part before clock $y$. This implies that $\ClockValuation{1}$ and $\ClockValuation{2}$ are in the same region. $\ClockValuation{3}$ is in a different region since the integer part of $\ClockValuation{3}(x)$ is one.

Since the fractional part of $\ClockValuation{4}(x)$ is equal to the fractional part of $\ClockValuation{4}(y)$, which means that $x$ and $y$ will change its integer part at the same time, $\ClockValuation{4}$ is in a different region than the other clock valuations.
\end{example}

Given a diagonal-free TA $A$ over a set of clocks $C$, we define $k : C \rightarrow \mathbb{N}^{\geq 0}$ to map 
each clock to a natural number, such that for any clock constraint $x \sim n$ with $x \in C$ and $n \in \mathbb{N}^{\geq 0}$ occurring 
in $A$, it holds that $n < k(x)$. We are allowed to consider all values for $x$ which are larger than $k(x)$ to be equivalent, since every constraint of a diagonal-free TA of the form $x \sim n$ is either satisfied or not satisfied by all these values~\cite{AlurDill}. By $\lfloor z \rfloor$ we denote the integral part of a number and by $\text{frac}(z)$ we denote the fractional part of a number. 
Region equivalence is defined as follows.

\begin{definition}[Region Equivalence \cite{ALUR1994183}]
\label{def:background:k-equivalence}
Let $C$ be a set of clocks, $k : C \rightarrow \mathbb{N}^{\geq 0}$ a function, and $\ClockValuation{A}, \ClockValuation{B} : C \rightarrow \TimeDomain$ two clock valuations. 
$\ClockValuation{A}$ and $\ClockValuation{B}$ are \emph{region equivalent}, denoted $\ClockValuation{A} \simeq_k \ClockValuation{B}$, if and only if
\begin{enumerate}
\item $\forall c \in C: ((\lfloor \ClockValuation{A}(c) \rfloor = \lfloor \ClockValuation{B}(c) \rfloor) \lor (\ClockValuation{A}(c) > k(c) \land \ClockValuation{B}(c) > k(c)))$,
\item $\forall c, c' \in C : (\ClockValuation{A}(c) \leq k(c) \land \ClockValuation{A}(c') \leq k(c')) \text{ implies } (\text{frac}(\ClockValuation{A}(c)) \leq \text{frac}(\ClockValuation{A}(c')) \Leftrightarrow \text{frac}(\ClockValuation{B}(c)) \leq \text{frac}(\ClockValuation{B}(c')))$, and
\item $\forall c \in C : (\ClockValuation{A}(c) \leq k(c)) \text{ implies } (\text{frac}(\ClockValuation{A}(c)) = 0 \Leftrightarrow \text{frac}(\ClockValuation{B}(c)) = 0)$
\end{enumerate}
holds. We define two states $\TLTSFullState{1}{}$, $\TLTSFullState{2}{}$ to be region equivalent, if and only if $l_1 = l_2$ and $\ClockValuation{1} \simeq_k \ClockValuation{2}$ hold.
\end{definition}

Using region equivalence, it is possible to create a finite region-graph representation of timed automata semantics. This construction has been used by \u{C}er\={a}ns to prove that timed bisimulation is decidable~\cite{Cerans1992}. However, region graphs may become unnecessarily large containing many regions, especially if large constants appear within clock constraints~\cite{Pettersson}.
For any clock constraint  $\phi \in \mathcal{B}(C)$ of the form $\phi = \phi_1 \land \phi_2 \land ...$ with $\phi_i = c \sim n$ and $n < k(c)$ and any two clock valuations $\ClockValuation{1}, \ClockValuation{2}$ of the same region, either $\ClockValuation{1}, \ClockValuation{2} \models \phi$ or $\ClockValuation{1}, \ClockValuation{2} \not\models \phi$ holds. Nevertheless, the distinction into two different regions might not be relevant for the analysis of the TA and, therefore, the symbolic analysis introduced by regions can be further improved. This is the main idea of \textit{zones}.

\subsection{Zone Graphs}

Zone graphs define a symbolic representation of timed automata semantics. Like regions, zones also represent a set of clock valuations, but in typical cases zone graphs are significantly smaller.
\begin{definition}[Zone]
\label{def:background:zones}
Let $C$ be a set of clocks. 
A zone $D\in 2^{C \rightarrow \TimeDomain}$ is a set of clock
valuations corresponding to a clock constraint $\phi\in\mathcal{B}(C)$
such that $D = \{u \ | \ u \in (C \rightarrow \TimeDomain) \wedge \ClockValuation{} \models \phi\}$.
Let zones $D, D' \in \mathcal{D}(C)$ and clock constraint $\phi$ be 
defined over the same set $C$, and $R \subseteq C$.
We define the operations:
\begin{align*}
D \cap D' &= \{u \ | \ u \in D \land u \in D'\} && \text{($\cap$-operation)},\\
 D \land \phi &= \{u \ | \ u \in D \land u \models \phi\} && \text{($\land$-operation)},\\
 D^{\uparrow} &= \{u+d \ | \ d \in \TimeDomain \land u \in D\} && \text{(future-operation), and} \\
 R(D) &= \{[R \rightarrow 0] u \ | \ u \in D\} && \text{(reset-operation)}.
 \end{align*}
\end{definition}

In the symbolic semantics of TA, clock constraints $\phi$ over clocks $C$ are used to 
characterize the set of clock valuations $u$ for $C$ contained in a zone $D$
in a finite way (even if $D$ comprises an infinite number of clock valuations).
By $[\phi]$ we denote the maximum set of clock valuations for $C$
satisfying clock constraint $\phi$ over $C$. 
The operations on zones in Definition~\ref{def:background:zones}
can be lifted to the set $\mathcal{B}(C)$ of clock constraints as
$\mathcal{B}(C)$ is closed under these operations and 
the operator $[ \ ]$ commutes with these operations~\cite{Yi1995}.

\begin{example}
The states of Example~\ref{ex:background:region-graphs} all have outgoing delay transitions as long as the clock valuation of the resulting state evaluates $x$ to a value smaller than 2. Moreover, they all have an outgoing action transition, labeled with $b$. Therefore, we summarize all clock valuations of Example~\ref{ex:background:region-graphs} into the zone $[x < 2]$. Therefore, we only need to analyze a single symbolic state in the zone graph, while we have three different symbolic states in the region graph.
\end{example}

In most practical TA analysis tools,
\textit{Difference Bound Matrices (DBM)} are used to represent zones~\cite{Belman57, Dill90}. 
A DBM is a conjunction of $(|C|+1)^2$ basic constraints of the form $c - c' \preccurlyeq z$, $c \preccurlyeq z$, or $-c \preccurlyeq z$, respectively, with $\preccurlyeq \ \in \{<, \leq\}$ and $z \in \mathbb{Z} \cup \{\infty\}$. 
Given a set of clocks $C$, we refer to the set of all DBM over $C$ by $\mathcal{M}(C)$. 
An element of a DBM, which represents a basic constraint, is called \emph{strong}, if the corresponding clock constraint cannot be replaced by a stronger clock constraint 
without changing the set of clock valuations contained in the zone. 
For any DBM, there exists a canonical form in which all elements are strong.
For more details on DBM we refer the reader 
to the comprehensive overview of Bengtsson and Yi~\cite{Bengtsson2004}. 

A \emph{zone graph} defines the operational semantics of a TA as a transition system.
If a symbolic state $\ZGFullState{}{}$ is reachable in a zone graph, 
then every state $\TLTSFullState{}{}$ with $\ClockValuation{} \in \Zone{}$ is reachable in the corresponding TLTS. 
Conversely, if a state $\TLTSFullState{}{}$ of the TLTS is reachable, then there exists a reachable symbolic state $\ZGFullState{}{}$ with $\ClockValuation{} \in \Zone{}$ in the corresponding zone graph.
The following definition of zone graphs is a slightly adapted version to the zone graphs of Yi et al.~\cite{Yi1995}. We denote $\Zone{0} = \{[C \rightarrow 0]\}$.
\begin{definition}[Zone Graph]
\label{def:background:zonegraphs}
Let $A = (L, l_0, \Sigma, C, I, E)$ be a TA. 
A \emph{zone graph} $\ZG{A} = (\ZGAllStates{}, \ZGFullState{0}{}, \Sigma \cup \{\varepsilon\}, \ZGTrans{})$ 
of $A$ is a transition system, where $\ZGAllStates{} = L \times \mathcal{D}(C)$ is the set of symbolic states with $\ZGFullState{0}{} \in \ZGAllStates{}$ being the initial symbolic state, and $\ZGTrans{} \ \subseteq \ZGAllStates{} \times \Sigma \cup \{\varepsilon\} \times \ZGAllStates{}$ is the least relation satisfying the rules:
\begin{itemize}
\item $\ZGFullState{}{} \ZGTrans{\varepsilon} (l, \Zone{}^{\uparrow} \land I(l))$, and
\item $\ZGFullState{1}{} \ZGTrans{\sigma} (l_2, R(\Zone{1} \land g) \land I(l_2))$, if $\TATrans{1}{g}{\sigma}{R}{2}$ and $R(\Zone{1} \land g) \land I(l_2) \neq \emptyset$.
\end{itemize}
\end{definition}

Note that there always exists a clock constraint $\phi_0$ for which $[\phi_0] = \{[C \rightarrow 0]\}$ holds.
As all zones of a zone graph are constructed by using the previously described operations, 
each zone $D$ reachable from $\Zone{0}$ in a zone graph can be represented by 
a (finite) clock constraint $\phi$ with $[\phi] = D$. 
Since $D_0 \land I(l_0) = D_0$ holds for the initial symbolic state $(l_0, D_0)$ (as we assume $[C \rightarrow 0] \models I(l_0)$) 
and since we intersect the target zone $D$ of a transition by the invariant of the target location $l$, 
$D \land I(l) = D$ holds for all reachable states $(l, D)$.
Moreover, due to $D \subseteq D^{\uparrow}$ and since we exclude empty symbolic states after an action-labeled transition,
it follows that no reachable zone is ever empty. 
This last property, which will become useful later on, is the only difference of our zone graph as compared to Yi et al.~\cite{Yi1995}. 
For locations $l$, $l'$, clock valuation $u$, and zone $D$, we denote by $(l, u) \in (l', D)$ that $l = l'$ and $u \in D$ holds. 
The following propositions show that a zone graph has two crucial properties.
\begin{proposition}[Backward Stability of Zone Graphs]
\label{prop:background:backward-stability}
Assume a TA $A$ with set of clocks $C_A$, locations $l_1, l_2$, TLTS $\TLTS{A}$, and zone graph $\ZG{A}$. Let $\Zone{1}, \Zone{2} \in \mathcal{D}(C_A)$ be zones. If there is a transition $\ZGFullState{1}{1} \ZGTrans{\varepsilon} \ZGFullState{1}{2}$, then for any $\ClockValuation{2} \in \Zone{2}$ exists a $\ClockValuation{1} \in \Zone{1}$ and a $d \in \TimeDomain$ such that $\TLTSFullState{1}{} \TLTSTrans{d} \TLTSFullState{1}{2}$.
%
%
If there is a transition $\ZGFullState{1}{1} \ZGTrans{\sigma} \ZGFullState{2}{2}$ with $\sigma \in \Sigma$ then for any $\ClockValuation{2} \in \Zone{2}$ exists a $\ClockValuation{1} \in \Zone{1}$ such that $\TLTSFullState{1}{} \TLTSTrans{\sigma} \TLTSFullState{2}{}$.
%
%
\end{proposition}

\begin{shortproof}
This is a direct consequence of the results by Yi et al.~\cite{Yi1995}.
\end{shortproof}

Forward Stability is defined in a similar way. 
We remind the reader of the fact that any symbolic state has exactly one outgoing $\varepsilon$-transition.

\begin{proposition}[Forward Stability of Zone Graphs]
\label{prop:background:forward-stability}
Assume a TA $A$ with set of clocks $C_A$, locations $l_1, l_2$, TLTS $\TLTS{A}$, and zone graph $\ZG{A}$. Let $\ClockValuation{1}, \ClockValuation{2} \in (C_A \rightarrow \TimeDomain)$. If there is a transition $\TLTSFullState{1}{} \TLTSTrans{d} \TLTSFullState{1}{2}$ with $d \in \TimeDomain$, then for any zone $\Zone{1}$ with $\ClockValuation{1} \in \Zone{1}$ and transition $\ZGFullState{1}{} \ZGTrans{\varepsilon} \ZGFullState{1}{2}$, $\ClockValuation{2} \in \Zone{2}$ holds.
%
%
If there is a transition $\TLTSFullState{1}{} \TLTSTrans{\sigma} \TLTSFullState{2}{}$ with $\sigma \in \Sigma$, then then for any zone $\Zone{1}$ with $\ClockValuation{1} \in \Zone{1}$ there is a transition $\ZGFullState{1}{} \ZGTrans{\sigma} \ZGFullState{2}{}$ with $\ClockValuation{2} \in \Zone{2}$.
\end{proposition}

\begin{shortproof}
This is a direct consequence of the results by Yi et al.~\cite{Yi1995}.
\end{shortproof}

These properties of zone graphs are usually represented 
in a more compact way. 
Nevertheless, we use this representation to point out the relationship between transitions in the TLTS and their
counter-part in zone graphs.
%
%
\newcommand{\backgroundExampleZGVertDistance}{1cm}
Zone graphs do not necessarily yield a finite representation of TA semantics. To solve this issue, the same idea that we saw in Definition~\ref{def:background:k-equivalence} is used in recent works.
For every TA, a (sound and complete) finite zone graph representation
can be obtained using \textit{k-normalization}~\cite{Bengtsson2004, Rokicki}. Let $D$ be a zone, which can be described by a clock constraint $\phi = \phi_0 \land \phi_1 \land ..$. The main idea is to replace all $\phi_i$ of the form $c \preccurlyeq n$, $c - c' \preccurlyeq n$ where ${\preccurlyeq} \in \{<, \leq\}$ and $n > k(c)$ with true, and to replace all $\phi_i$ of the form $c \succcurlyeq n$, $c - c' \succcurlyeq n$ where ${\succcurlyeq} \in \{>, \geq\}$ and $n > k(c)$ with $c > k(c)$ or $c - c' > k(c)$ respectively~\cite{Bengtsson2004}. 
We use the notation $\text{norm}(\Zone{}, k)$ to indicate that k-normalization is used. 
In the worst case, zone graphs become exponentially large in the number of reachable symbolic states, as shown in the following example.

\newcommand{\backgroundExampleZGHorDistance}{0.9cm}

\begin{example}
\label{ex:background:alternating-sequences}
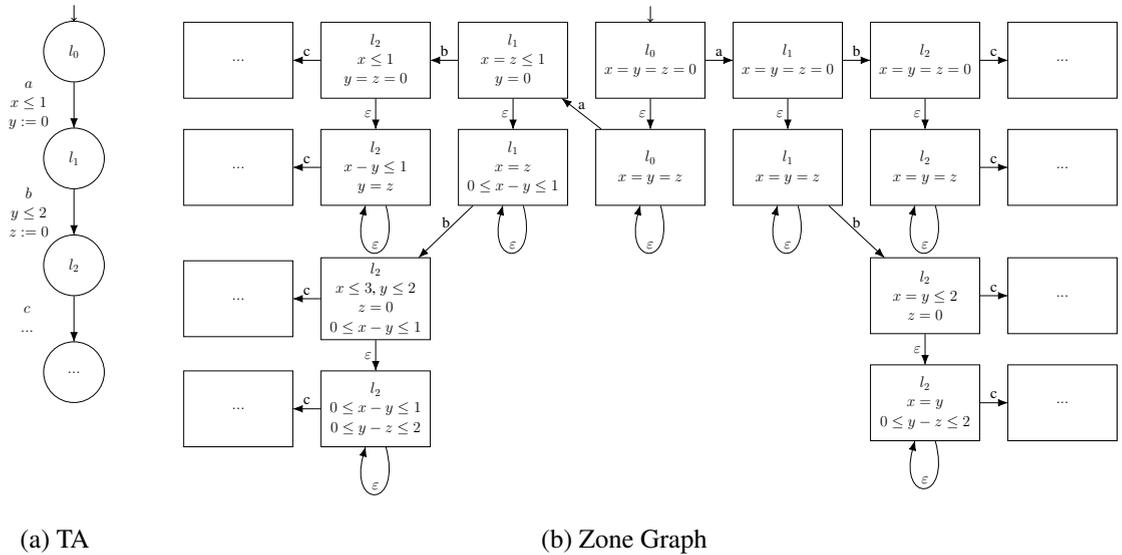
\begin{figure}[h!]
\begin{subfigure}{0.2\textwidth}
\centering
\scalebox{\backgroundExampleAutomataScalingFactor}{
\begin{tikzpicture}
\tikzstyle{every node}=[font=\backgroundTikzFontSize]
\tikzstyle{state} = [draw,circle,minimum size=2cm,inner sep=0pt,semithick]

\node[state, align=center, initial above, initial text=] (0) {$l_0$};
\node[state, align=center, below = 1.5cm of 0] (1) {$l_1$};
\node[state, align=center, below = 1.5cm of 1] (2) {$l_2$};
\node[state, align=center, below = 1.5cm of 2] (3) {$...$};
\node[state, draw = none, align=center, below = 1.5cm of 3] (4) {};
\backgroundExampleTAArrowDesc (0) --node[left = 0.7cm, align=center]{$a$\\$x \leq 1$\\$y:=0$} (1);
\backgroundExampleTAArrowDesc (1) --node[left = 0.7cm, align=center]{$b$\\$y \leq 2$\\$z:=0$} (2);
\backgroundExampleTAArrowDesc (2) --node[left = 1.2cm, align=center]{$c$\\$...$} (3);
\end{tikzpicture}
}
\caption{TA}
\label{subfig:background:zonegraph:example-alternating-seq:ta}
\end{subfigure}
\begin{subfigure}{0.75\textwidth}
\centering
\scalebox{\backgroundExampleAutomataScalingFactor}{
\begin{tikzpicture}
\tikzstyle{every node}=[font=\backgroundTikzFontSize]
\tikzstyle{symstate} = [draw,rectangle,minimum width=3.6cm,minimum height=2.5cm,inner sep=5pt,thick]
%
\node[symstate, align=center, initial above, initial text=] (0) {$l_0$\\$x=y=z=0$};
\node[symstate, align=center, below = \backgroundExampleZGVertDistance of 0] (1) {$l_0$\\$x=y=z$};
%
\node[symstate, align=center, right = \backgroundExampleZGHorDistance of 0] (2) {$l_1$\\$x=y=z=0$};
\node[symstate, align=center, below = \backgroundExampleZGVertDistance of 2] (3) {$l_1$\\$x=y=z$};
%
\node[symstate, align=center, left = \backgroundExampleZGHorDistance of 0] (4) {$l_1$\\$x = z\leq 1$\\$y=0$};
\node[symstate, align=center, below = \backgroundExampleZGVertDistance of 4] (5) {$l_1$\\$x = z$\\$0 \leq x-y \leq 1$};
%
%
%
\node[symstate, align=center, right = \backgroundExampleZGHorDistance of 2] (6) {$l_2$\\$x=y=z=0$};
\node[symstate, align=center, below = \backgroundExampleZGVertDistance of 6] (7) {$l_2$\\$x=y=z$};
%
\node[symstate, align=center, right = \backgroundExampleZGHorDistance of 6] (6target) {$...$};
\node[symstate, align=center, right = \backgroundExampleZGHorDistance of 7] (7target) {$...$};
%
\node[symstate, align=center, below = 1.7*\backgroundExampleZGVertDistance of 7] (8) {$l_2$\\$x = y\leq2$\\$z=0$};
\node[symstate, align=center, below = \backgroundExampleZGVertDistance of 8] (9) {$l_2$\\$x = y$\\$0 \leq y-z \leq 2$};
%
\node[symstate, align=center, right = \backgroundExampleZGHorDistance of 8] (8target) {$...$};
\node[symstate, align=center, right = \backgroundExampleZGHorDistance of 9] (9target) {$...$};
%
\node[symstate, align=center, left = \backgroundExampleZGHorDistance of 4] (10) {$l_2$\\$x \leq 1$\\$y=z=0$};
\node[symstate, align=center, below = \backgroundExampleZGVertDistance of 10] (11) {$l_2$\\$x - y \leq 1$\\$y=z$};
%
\node[symstate, align=center, left = \backgroundExampleZGHorDistance of 10] (10target) {$...$};
\node[symstate, align=center, left = \backgroundExampleZGHorDistance of 11] (11target) {$...$};
%
\node[symstate, align=center, below = 1.7*\backgroundExampleZGVertDistance of 11] (12) {$l_2$\\$x \leq 3$, $y \leq 2$\\$z=0$\\$0 \leq x-y \leq 1$};
\node[symstate, align=center, below = \backgroundExampleZGVertDistance of 12] (13) {$l_2$\\$0 \leq x-y \leq 1$\\$0 \leq y-z \leq 2$};
%
\node[symstate, align=center, left = \backgroundExampleZGHorDistance of 12] (12target) {$...$};
\node[symstate, align=center, left = \backgroundExampleZGHorDistance of 13] (13target) {$...$};
%
%
%
%
%
\backgroundExampleTAArrowDesc (0) --node[left, align=center]{$\varepsilon$} (1);
\backgroundExampleTAArrowDesc (1) to[loop below] node[above, align=center]{$\varepsilon$} (1);
\backgroundExampleTAArrowDesc (0) --node[above, align=center]{a} (2);
\backgroundExampleTAArrowDesc (1) --node[above, align=center]{a} (4);
%
%
%
\backgroundExampleTAArrowDesc (2) --node[left, align=center]{$\varepsilon$} (3);
\backgroundExampleTAArrowDesc (3) to[loop below] node[above, align=center]{$\varepsilon$} (3);
\backgroundExampleTAArrowDesc (2) --node[above, align=center]{b} (6);
\backgroundExampleTAArrowDesc (3) --node[above, align=center]{b} (8);
%
\backgroundExampleTAArrowDesc (4) --node[left, align=center]{$\varepsilon$} (5);
\backgroundExampleTAArrowDesc (5) to[loop below] node[above, align=center]{$\varepsilon$} (5);
\backgroundExampleTAArrowDesc (4) --node[above, align=center]{b} (10);
\backgroundExampleTAArrowDesc (5) --node[above, align=center]{b} (12);
%
%
%
\backgroundExampleTAArrowDesc (6) --node[left, align=center]{$\varepsilon$} (7);
\backgroundExampleTAArrowDesc (7) to[loop below] node[above, align=center]{$\varepsilon$} (7);
\backgroundExampleTAArrowDesc (6) --node[above, align=center]{c} (6target);
\backgroundExampleTAArrowDesc (7) --node[above, align=center]{c} (7target);
%
\backgroundExampleTAArrowDesc (8) --node[left, align=center]{$\varepsilon$} (9);
\backgroundExampleTAArrowDesc (9) to[loop below] node[above, align=center]{$\varepsilon$} (9);
\backgroundExampleTAArrowDesc (8) --node[above, align=center]{c} (8target);
\backgroundExampleTAArrowDesc (9) --node[above, align=center]{c} (9target);
%
\backgroundExampleTAArrowDesc (10) --node[left, align=center]{$\varepsilon$} (11);
\backgroundExampleTAArrowDesc (11) to[loop below] node[above, align=center]{$\varepsilon$} (11);
\backgroundExampleTAArrowDesc (10) --node[above, align=center]{c} (10target);
\backgroundExampleTAArrowDesc (11) --node[above, align=center]{c} (11target);
%
\backgroundExampleTAArrowDesc (12) --node[left, align=center]{$\varepsilon$} (13);
\backgroundExampleTAArrowDesc (13) to[loop below] node[above, align=center]{$\varepsilon$} (13);
\backgroundExampleTAArrowDesc (12) --node[above, align=center]{c} (12target);
\backgroundExampleTAArrowDesc (13) --node[above, align=center]{c} (13target);
\end{tikzpicture}
}
\caption{Zone Graph}
\label{subfig:background:zonegraph:example-alternating-seq:zg}
\end{subfigure}
\caption{TA with exponentially large zone graph}
\label{fig:background:zonegraph:example-alternating-seq}
\end{figure}
\begin{figure}
\centering
\scalebox{\backgroundExampleAutomataScalingFactor}{
\begin{tikzpicture}
\tikzstyle{every node}=[font=\backgroundTikzFontSize]
\tikzstyle{symstate} = [draw,rectangle,minimum width=3.5cm,minimum height=2.6cm,inner sep=5pt,thick]
%
\node[symstate, align=center, initial left, initial text=] (0) {$l_0$\\$x=y=z=0$};
\node[symstate, align=center, right = 1.5cm of 0] (1) {$l_0$\\$x=y=z$};
%
\node[symstate, align=center, right = 1.5cm of 1] (4) {$l_1$\\$x = z\leq 1$\\$y=0$};
\node[symstate, align=center, right = 1.5cm of 4] (5) {$l_1$\\$x = z$\\$0 \leq x-y \leq 1$};
%
\node[symstate, align=center, right = 1.5cm of 5] (12) {$l_2$\\$x \leq 3$, $y \leq 2$\\$z=0$\\$0 \leq x-y \leq 1$};
\node[symstate, align=center, right = 1.5cm of 12] (13) {$l_2$\\$0 \leq x-y \leq 1$\\$0 \leq y-z \leq 2$};
%
\node[symstate, align=center, right = 1.5cm of 13] (13target) {$...$};
%
%
%
%
%
\backgroundExampleTAArrowDesc (0) --node[above, align=center]{$\varepsilon$} (1);
\backgroundExampleTAArrowDesc (1) --node[above, align=center]{a} (4);
%
\backgroundExampleTAArrowDesc (4) --node[above, align=center]{$\varepsilon$} (5);
\backgroundExampleTAArrowDesc (5) --node[above, align=center]{b} (12);
%
\backgroundExampleTAArrowDesc (12) --node[above, align=center]{$\varepsilon$} (13);
\backgroundExampleTAArrowDesc (13) --node[above, align=center]{c} (13target);
\end{tikzpicture}
}
\caption{Zone Graph with alternating sequences}
\label{subfig:background:zonegraph:example-alternating-seq:short-zg}
\end{figure}
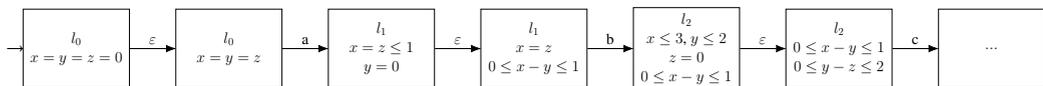
Starting from the initial symbolic state, we can see in Figure \ref{subfig:background:zonegraph:example-alternating-seq:zg} that the zone graph has an exponential size. While $l_0$ appears only twice in the zone graph, $l_1$ appears four times and $l_2$ eight times. The reader may notice that this problem becomes even worse when regions are used.
\end{example}

Yi et al. \cite{Yi1995} address this problem by considering only sequences that alternate between transitions labeled with actions from $\Sigma$ and transitions labeled with $\varepsilon$. For TLTS, this is allowed since two action transitions can be interleaved by a zero-delay transition and two delay transitions can be combined to the sum of them. Therefore, every reachable state in a TLTS can also be reached by an alternating sequence. This implies that it is feasible to only use alternating sequences in zone graphs (i.e., we only allow an action transition after an $\varepsilon$-transition and an $\varepsilon$-transition after an action transition). Figure \ref{subfig:background:zonegraph:example-alternating-seq:short-zg} shows how this avoids the exponential growth of zone graphs.

Unfortunately, any notion of equivalence based on zone graphs
is too coarse-grained to provide a sound basis
for timed bisimilarity checking as we illustrate by the following example.

\begin{example}
\label{ex:background:zone-graphs-do-not-contain-sufficient-information}
The attempt to check timed bisimilarity with the help of zone graphs basically has two main problems. The first problem concerns the case in which the set of clocks are not equal. How should the elapsed time be measured and compared between both? The second problem, illustrated by this example, is the non-observability of clock resets. To illustrate the non-observability, we consider TA $A_1$ and $A_2$ of Figure \ref{fig:examples}. Both have exactly the same zone graph, shown in Figure \ref{fig:background:zone-graphs:zone-graphs-do-not-contain-sufficient-information}. As we have already seen in Example \ref{ex:background:timed-bisim}, $A_1$ and $A_2$ are not timed bisimilar, due to the state $(l_0, x=1)$. Therefore, there are two non-bisimilar TA with the same zone graph and no algorithm that is based solely on zone graphs is able to return $\text{true}$ if and only if the corresponding TA are timed bisimilar.

\begin{figure}
\centering
\scalebox{\backgroundExampleAutomataScalingFactor}{
\begin{tikzpicture}
\tikzstyle{every node}=[font=\backgroundTikzFontSize]
\tikzstyle{symstate} = [draw,rectangle,minimum width=3.5cm,inner sep=5pt,thick]
%
\node[symstate, align=center, initial above, initial text=] (0) {$l_0$\\$x=0$};
\node[symstate, align=center, below = \backgroundExampleZGVertDistance of 0] (1) {$l_0$\\$x<\infty$};
%
\node[symstate, align=center, right = 1.5cm of 1] (2) {$l_1$\\$x=0$};
\node[symstate, align=center, above = \backgroundExampleZGVertDistance of 2] (3) {$l_1$\\$x<2$};
\node[symstate, align=center, right = 1.5cm of 2] (3inter) {$l_2$\\$x=0$};
\node[symstate, align=center, right = 1.5cm of 3inter] (5) {$l_2$\\$x<\infty$};
\node[symstate, align=center, above =\backgroundExampleZGVertDistance of 5] (4) {$l_2$\\$x<2$};
%
%
%
%
\backgroundExampleTAArrowDesc (0) --node[left, align=center]{$\varepsilon$} (1);
\backgroundExampleTAArrowDesc (0) --node[above, align=center]{a} (2);
\backgroundExampleTAArrowDesc (1) --node[above, align=center]{a} (2);
\backgroundExampleTAArrowDesc (1) to[loop left] node[right, align=center]{$\varepsilon$} (1);
\backgroundExampleTAArrowDesc (2) --node[left, align=center]{$\varepsilon$} (3);
\backgroundExampleTAArrowDesc (2) --node[above, align=center]{b} (3inter);
\backgroundExampleTAArrowDesc (3) to[loop above] node[below, align=center]{$\varepsilon$} (3);
\backgroundExampleTAArrowDesc (3) --node[above, align=center]{b} (4);
\backgroundExampleTAArrowDesc (3inter) --node[above, align=center]{$\varepsilon$} (5);
\backgroundExampleTAArrowDesc (4) --node[left, align=center]{$\varepsilon$} (5);
\backgroundExampleTAArrowDesc (5) to[loop right] node[left, align=center]{$\varepsilon$} (5);
\draw[-] (5) -- (15.1, -3.5);
\draw[-] (15.1, -3.5) -- (-4.5, -3.5);
\draw[-] (-4.5, -3.5) -- (-4.5, 0) node[anchor=east, midway]{c};
\backgroundExampleTAArrowDesc (-4.5, 0) -- (0);
\end{tikzpicture}
}
\caption{Zone Graph of $A_1$ (Figure \ref{fig:exampleA1}) and $A_2$ (Figure \ref{fig:exampleA2})}
\label{fig:background:zone-graphs:zone-graphs-do-not-contain-sufficient-information}
\end{figure}
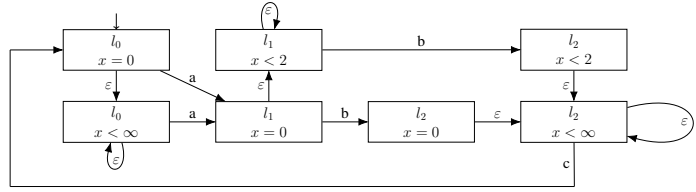

A region-based approach avoids this problem, as it splits the pair of symbolic states $((l_0, x < \infty), (l_0, x < \infty))$ (the first from $A_1$, the second from $A_2$) into pairs of regions, such as $((l_0, x=0), (l_0, x=0))$, $((l_0, 0 < x < 1), (l_0, 0 < x < 1))$, $((l_0, x = 1), (l_0, x = 1))$, and so on. For some of these pairs, the second component has an outgoing transition labeled \texttt{a} while the first does not, making the difference observable. This precision comes at the cost of a potentially large number of regions.
\end{example}


%% file: sections/virtual_clocks.tex
\section{Virtual Clocks}
\label{chap:virtual_clocks:main}
\input{sections/virtual_clock_chapter/virtual_clock_intro}
\input{sections/virtual_clock_chapter/bounded_bisim_TLTS}
\input{sections/virtual_clock_chapter/TSVCs}
\input{sections/virtual_clock_chapter/VCGs}

%% file: sections/virtual_clock_chapter/virtual_clock_intro.tex
\begin{figure}[tp]
\centering
\scalebox{0.9}{
\begin{tikzpicture}
\node (00) [rectangle, align=center] at (0, 0) {$A_{\textit{ZG}}$};
\node (03) [rectangle, align=center] at (0, -2.5) {$B_{\textit{ZG}}$};
\node (035) [rectangle, align=center] at (0, -3) {(finite)};
\draw[<->] (00) -- (03) node[midway,above,sloped] {bisimilar};
\node (0mr) [rectangle] at (0.5, -1.25) {};
\node (2ml) [rectangle] at (2, -1.25) {};
\draw[<->] (0mr) -- (2ml) node[midway, above] {\cancel{iff}};
\node (20) [rectangle, align=center] at (2, 0) {$A_{\textit{TLTS}}$};
\node (23) [rectangle, align=center] at (2, -2.5) {$B_{\textit{TLTS}}$};
\node (235) [rectangle, align=center] at (2, -3) {(infinite)};
\draw[<->] (20) -- (23) node[midway,above,sloped] {bisimilar};
\node (2mr) [rectangle] at (2.5, -1.25) {};
\node (5ml) [rectangle] at (4.5, -1.25) {};
\draw[<->] (2mr) -- (5ml) node[midway, above] {iff};
\node (50) [rectangle, align=center] at (4.5, 0) {$A$};
\node (53) [rectangle, align=center] at (4.5, -2.5) {$B$};
\draw[<->] (50) -- (53) node[midway,above,sloped] {bisimilar};
\node (5mr) [rectangle] at (5, -1.25) {};
\node (7ml) [rectangle] at (7, -1.25) {};
\draw[<->] (5mr) -- (7ml) node[midway, above] {iff};
\node (70) [rectangle, align=center] at (7, 0) {$A_{\textit{TSVC}}$};
\node (73) [rectangle, align=center] at (7, -2.5) {$B_{\textit{TSVC}}$};
\node (735) [rectangle, align=center] at (7, -3) {(infinite)};
\draw[<->] (70) -- (73) node[midway,above,sloped] {bisimilar};
\node (5mr) [rectangle] at (7.5, -1.25) {};
\node (7ml) [rectangle] at (9.5, -1.25) {};
\draw[<->] (5mr) -- (7ml) node[midway, above] {iff};
\node (95) [rectangle, align=center] at (9.5, 0) {$A_{\textit{VCG}}$};
\node (98) [rectangle, align=center] at (9.5, -2.5) {$B_{\textit{VCG}}$};
\node (985) [rectangle, align=center] at (9.5, -3) {(finite)};
\draw[<->] (95) -- (98) node[midway,above,sloped] {bisimilar};
\end{tikzpicture}
}
\caption{Road map for chapter \ref{chap:virtual_clocks:main}}
\label{figure:virtual_clocks:roadmap-theoretical-contribution}
\end{figure}
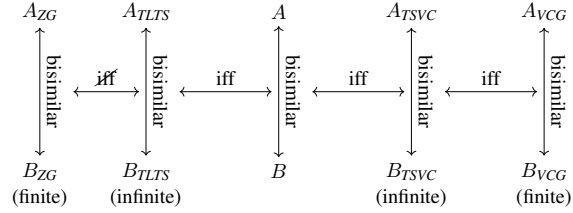


In this section, we propose an extension of zone graphs which
enables sound and complete timed bisimilarity checking of TA.
The road map towards this goal is shown in Figure \ref{figure:virtual_clocks:roadmap-theoretical-contribution}. 
Two TA, $A$ and $B$, are timed bisimilar if and only if their TLTS, $\TLTS{A}$ and $\TLTS{B}$, 
are bisimilar.
However, as TLTS have an infinite number of states 
(countably many states in the discrete time model, uncountably many states in the dense-time model), TLTS
do not permit effective timed bisimilarity checking.
Although region graphs contain enough information to check for bisimulation, their large size makes them unusable in practice.
In contrast, the corresponding zone graphs $\ZG{A}$ and $\ZG{B}$ of $A$ and $B$
can be pruned to be finite using $k$-normalization and are, on average, much smaller than region graphs. Nevertheless, zone graphs
do not contain enough information to ensure correctness of the bisimilarity 
check.

To solve both problems described in Example~\ref{ex:background:zone-graphs-do-not-contain-sufficient-information}, we introduce fresh clocks (one for each original clock)
that do not occur in the surface syntax of the corresponding TA (i.e., neither in guards and resets, nor in invariants).
We refer to these as \emph{virtual} clocks. The TLTS constructed over the enlarged set of clocks is called TSVC and the corresponding zone graph is called VCG. Importantly, this operation does not introduce a new kind of transition structure, but enriches the information captured by zones. Based on virtual clocks, we are able to introduce virtual bisimulation, an equivalence notion over VCG that holds if and only if the original TA are timed bisimilar.

%% file: sections/virtual_clock_chapter/bounded_bisim_TLTS.tex
\subsection{Bounded Timed Bisimulation}

In this subsection, we first consider
an alternative way to define timed bisimulation
analogously to Milner in the context of CCS~\cite{MilnerCCS}.
Two TLTS states are \emph{bounded timed bisimilar} 
if both behave similar up to the next $n$ steps. 
This definition eases many proof steps for the subsequent 
constructions as it enables inductive proof steps
with explicit mentioning of $n$.
\begin{definition}[Bounded Timed Bisimulation for TLTS]
\label{def:virtual-clocks:bounded-timed-bisim-for-TLTS}
Let $\TLTS{A} = (\TLTSAllStates{A}, \TLTSState{0, A}, \Sigma \cup \TimeDomain, \TLTSTrans{})$ and $\TLTS{B} = (\TLTSAllStates{B}, \TLTSState{0, B}, \Sigma \cup \TimeDomain, \TLTSTrans{})$ be two TLTS. 
Any state $\TLTSState{A}$ of $\TLTS{A}$ is timed bisimilar in order zero to any state $\TLTSState{B}$ of $\TLTS{B}$, 
denoted $\TLTSState{A} \sim_{0} \TLTSState{B}$.
$\TLTSState{A}$ is timed bisimilar in order $(n+1)$ (with $n \in \mathbb{N}^{\geq 0}$) to $\TLTSState{B}$, denoted $\TLTSState{A} \sim_{n+1} \TLTSState{B}$, if and only if it holds that

\eject

\begin{enumerate}
\item the existence of a transition $\TLTSState{A} \TLTSTrans{\mu} \TLTSState{1, A}$ for $\mu \in \Sigma \cup \TimeDomain$ implies the existence of a transition $\TLTSState{B} \TLTSTrans{\mu} \TLTSState{1, B}$, with $\TLTSState{1, A} \sim_n \TLTSState{1, B}$, and
\item the existence of a transition $\TLTSState{B} \TLTSTrans{\mu} \TLTSState{1, B}$ for $\mu \in \Sigma \cup \TimeDomain$ implies the existence of a transition $\TLTSState{A} \TLTSTrans{\mu} \TLTSState{1, A}$, with $\TLTSState{1, A} \sim_n \TLTSState{1, B}$.
\end{enumerate}
\end{definition}

$\TLTS{A}$ is \textit{timed bisimilar in order $n$} to $\TLTS{B}$, written $\TLTS{A} \sim_n \TLTS{B}$, if and only if $\TLTSState{0, A} \sim_n \TLTSState{0, B}$ holds. A TA is \textit{timed bisimilar in order $n$} to another TA, if this holds for their respective TLTS. 
Monotonicity states that $\TLTSBisimulation{\TLTSState{A}}{\TLTSState{B}}{n+1}$ implies $\TLTSBisimulation{\TLTSState{A}}{\TLTSState{B}}{n}$.
\begin{proposition}[Monotonicity of Bounded Bisimulation]
\label{prop:virtual-clocks:bisimInOrderPlusImpliesBisimInOrder}
For any $n \in \mathbb{N}^{\geq 0}$, $\TLTSBisimulation{\TLTSState{A}}{\TLTSState{B}}{n+1}$ implies $\TLTSBisimulation{\TLTSState{A}}{\TLTSState{B}}{n}$.
\end{proposition}
\begin{shortproof}
Proof by induction in $n$.
\end{shortproof}
Bounded timed bisimulation converges against timed bisimulation 
for increasing $n$ and ultimately gets equivalent to timed bisimulation in case of $n \rightarrow \infty$. 

\begin{proposition}[Bounded Bisimulation and Bisimulation]
\label{prop:virtual-clocks:BoundedBisimVsUnboundedBisim}
There exists a timed bisimulation $R$ with $(\TLTSState{A}, \TLTSState{B}) \in R$, if and only if $\forall n \in \mathbb{N}^{\geq 0} : \TLTSBisimulation{\TLTSState{A}}{\TLTSState{B}}{n}$ holds.
\end{proposition}
\begin{proof}
We define the relation $R_n = \{(\TLTSState{A}, \TLTSState{B}) \ | \  \TLTSBisimulation{\TLTSState{A}}{\TLTSState{B}}{n}\}$. 
It is well-known that the power set $\mathcal{P}(\TLTSAllStates{A} \times \TLTSAllStates{B})$ with the relation $\subseteq$ is a complete lattice. Moreover, there is a function $f: \mathcal{P}(\TLTSAllStates{A} \times \TLTSAllStates{B}) \rightarrow \mathcal{P}(\TLTSAllStates{A} \times \TLTSAllStates{B})$ with $f(R_0) = R_1$, $f(R_1) = R_2$, ... This function is monotonic due to Proposition~\ref{prop:virtual-clocks:bisimInOrderPlusImpliesBisimInOrder}. Due to the Knaster-Tarski theorem~\cite{tarski1955lattice}, there exists a greatest fixed point which is the intersection of all $R_n$ and which we name $\lim_{n \to \infty} R_n$. Since $\lim_{n \to \infty} R_n$ is the intersection of all $R_n$, $(\TLTSState{A}, \TLTSState{B}) \in \lim_{n \to \infty} R_n$ if and only if $\forall n \in \mathbb{N}^{\geq 0} : \TLTSBisimulation{\TLTSState{A}}{\TLTSState{B}}{n}$ and we only have to show that $(\TLTSState{A}, \TLTSState{B}) \in \lim_{n \to \infty} R_n$ fulfills the conditions of Definition~\ref{def:background:sim-and-bisim:Strong-Timed-Bisimulation}. $(\TLTSState{A}, \TLTSState{B}) \in \lim_{n \to \infty} R_n$ implies $\forall n \in \mathbb{N}^{\geq 0} : \TLTSBisimulation{\TLTSState{A}}{\TLTSState{B}}{n}$. If there is a transition $\TLTSState{A} \TLTSTrans{\mu} \TLTSState{1, A}$, then there is a transition $\TLTSState{B} \TLTSTrans{\mu} \TLTSState{1, B}$ with $\forall n \in \mathbb{N}^{> 0} : \TLTSBisimulation{\TLTSState{1, A}}{\TLTSState{1, B}}{n-1}$, which implies $(\TLTSState{1, A}, \TLTSState{1, B}) \in \lim_{n \to \infty} R_n$ and analogously for the second condition.
\end{proof}
Moreover, since both propositions include the case in which $\TLTSState{A} = \TLTSState{0, A}$, respectively $\TLTSState{B} = \TLTSState{0, B}$, 
the proposition also holds for $\TLTS{A}$ and $\TLTS{B}$ itself.

%% file: sections/virtual_clock_chapter/TSVCs.tex
\subsection{Timed Labeled Transition Systems with Virtual Clocks}
\label{subsection:virtual-clocks:TSVCs}

To prove correctness of the following constructions, we extend,
as a purely theoretical base-line, the TLTS semantics
by virtual clocks. Two TSVC are bisimilar if and only if their corresponding TLTS are bisimilar.

%
%
We supplement $\TLTS{A}$ and $\TLTS{B}$ 
by a \emph{common} set of fresh clock variables $\chi$. We call these clocks virtual because they are only used for checking timed bisimilarity but never occur in any clock constraint of the TA.
As we will show, it suffices to add exactly one \emph{virtual} clock $\chi_i$
for each original clock $c_i\in C_A \cup C_B$.
%
%
To ensure freshness, we assume that $\forall i \in \mathbb{N}^{\geq 0} : \chi_i \not\in C$ 
for any set of clocks $C$ used in a TA under consideration which is
easily achievable by applying alpha conversion of the original clock names.
\begin{definition}[TA with Virtual Clocks]
\label{def:virtual_clocks:Extend-TA}
Given two TA $A = (L_A, l_{0, A}, \Sigma, C_A, I_A, E_A)$ and $B = (L_B, l_{0, B}, \Sigma, C_B, I_B, E_B)$, we define 
the TA \emph{$A$ regarding $B$} as $\text{add-virtual}(A, B) = (L_A, l_{0, A}, \Sigma, C_A \cup \{\chi_0, ..., \chi_{|C_A| + |C_B| - 1}\}, I_A, E_A)$.
\end{definition}

The TA $\text{add-virtual}(A, B)$ is a proper TA according to Definition~\ref{def:background:Timed-Automata:Timed-Automaton}
which equals TA $A$ except for the newly added set of virtual clocks (one for each clock in $C_{A}$ and in $C_{B}$).
These newly added clocks, however, never occur in any clock constraint or clock reset
of the TA $\text{add-virtual}(A, B)$.
%
%
As the newly added clocks behave like proper clocks, 
the TLTS semantics of TA $\text{add-virtual}(A, B)$ with virtual clocks (TSVC) is defined 
according to Definition~\ref{def:background:TLTS}.
%
\begin{definition}[TSVC]
\label{def:virtal_clocks:TSVC}
Given two TA $A$ and $B$ defined over $C_A$ and $C_B$, respectively, 
we define the \emph{TSVC} \emph{of} $A$ \emph{regarding} $B$ 
to be the TLTS of $\text{add-virtual}(A, B)$. 
Given a clock valuation 
$u : C_A \cup \{\chi_0, ..., \chi_{|C_A| + |C_B| - 1}\} \rightarrow \TimeDomain$, we
use the following notations:
\begin{align*}
\text{orig} &: (C_A \cup \{\chi_0, ..., \chi_{|C_A| + |C_B| - 1}\} \rightarrow \TimeDomain) \rightarrow (C_A \rightarrow \TimeDomain), \\
\text{virt} &: (C_A \cup \{\chi_0, ..., \chi_{|C_A| + |C_B| - 1}\} \rightarrow \TimeDomain) \rightarrow (\{\chi_0, ..., \chi_{|C_A| + |C_B| - 1}\} \rightarrow \TimeDomain), \\
\text{virt-A} &: (C_A \cup \{\chi_0, ..., \chi_{|C_A| + |C_B| - 1}\} \rightarrow \TimeDomain) \rightarrow (\{\chi_0, ..., \chi_{|C_A| - 1}\} \rightarrow \TimeDomain), \text{ and } \\
\text{virt-B} &: (C_A \cup \{\chi_0, ..., \chi_{|C_A| + |C_B| - 1}\} \rightarrow \TimeDomain) \rightarrow (\{\chi_{|C_A|}, ..., \chi_{|C_A| + |C_B| - 1}\} \rightarrow \TimeDomain),
\end{align*}
such that:
\begin{align*}
\forall c \in C_A &: \text{orig}(u)(c) = u(c),\\ 
\forall c \in \{\chi_0, ..., \chi_{|C_A| + |C_B| - 1}\} &: \text{virt}(u)(c) = u(c), \\ 
\forall c \in \{\chi_0, ..., \chi_{|C_A| - 1}\} &: \text{virt-A}(u)(c) = u(c), \text{ and } \\
\forall c \in \{\chi_{|C_A|}, ..., \chi_{|C_A| + |C_B| - 1}\} &: \text{virt-B}(u)(c) = u(c).
\end{align*}
\end{definition}

Given a clock valuation $u$, $\text{orig}(u)$ restricts $u$ to the valuation of
the original clocks from TA $A$, whereas $\text{virt}$, $\text{virt-A}$, $\text{virt-B}$
restricts $u$ to all virtual clocks, to those related to clocks of $A$ and to those related to clocks
of $B$, respectively.
%
We presume some canonical ordering on a set of clocks $C$ and 
refer to the $j$th clock in $C$ by $C[j]$, where $j$ ranges from 0 to $|C|-1$.

While the virtual clocks solve the first main problem, namely that $C_A$ and $C_B$ are disjoint, we can also use them to solve the second main problem, namely that resets can be masked by clock constraints. To solve this issue, we define two main classes of TSVC states. In \textit{\synchronized} states, all virtual clock values are equal to the corresponding original clock values. \textit{Semi-\synchronized} states are former \synchronized \ states, in which a (possibly empty) set of original clocks but not their related virtual clocks has been reset.
\begin{definition}[Classes of TSVC States]
\label{def:virtual-clocks:TSVC:Proper-and-syncd-States}
Let $\TSVC{A}$ be the TSVC of TA $A$ regarding $B$
and $\TSVC{B}$ be the TSVC of TA $B$ regarding $A$.
\begin{itemize}
	\item State $\TSVCFullState{A}{}{}$ of $\TSVC{A}$ is \textit{AB-\semisynchronized}, if and only if
\begin{equation*}
\forall i \in [0, |C_A| - 1] : (\ClockValuation{A}(C_A[i]) = \ClockValuation{A}(\chi_i) \lor \ClockValuation{A}(C_A[i]) = 0)\text{.}
\end{equation*}
	\item State $\TSVCFullState{A}{}{}$ of $\TSVC{A}$ is \textit{AB-\synchronized}, if and only if 
\begin{equation*}
\forall i \in [0, |C_A| - 1] : \ClockValuation{A}(C_A[i]) = \ClockValuation{A}(\chi_i)\text{.}
\end{equation*}
	\item State $\TSVCFullState{B}{}{}$ of $\TSVC{B}$ is \textit{AB-\semisynchronized}, if and only if
\begin{equation*}
\forall i \in [0, |C_B| - 1] :( \ClockValuation{B}(C_B[i]) = \ClockValuation{B}(\chi_{i + |C_A|}) \lor \ClockValuation{B}(C_B[i]) = 0)\text{.}
\end{equation*}
	\item State $\TSVCFullState{B}{}{}$ of $\TSVC{B}$ is \textit{AB-\synchronized}, if and only if 
\begin{equation*}
\forall i \in [0, |C_B| - 1] : \ClockValuation{B}(C_B[i]) = \ClockValuation{B}(\chi_{i + |C_A|})\text{.}
\end{equation*}

\end{itemize}
\end{definition}

Thus, if a state is AB-\synchronized \ it is also AB-\semisynchronized{} but not vice-versa.
Let $(\TSVCFullState{A}{}, \TSVCFullState{B}{})$ be a pair of states, where $\TSVCFullState{A}{}$ is an AB-\synchronized \ state of $\TSVC{A}$ 
and $\TSVCFullState{B}{}$ is an AB-\synchronized \ state of $\TSVC{B}$.
Then from $\text{virt}(\ClockValuation{A}) = \text{virt}(\ClockValuation{B})$ it follows that 
\[
\forall i \in [0, |C_A| - 1] : \ClockValuation{A}(C_A[i]) = \ClockValuation{A}(\chi_i) = \ClockValuation{B}(\chi_i)
\]
%
%
and, analogously, 
\[
\forall i \in [0, |C_B| - 1] : \ClockValuation{B}(C_B[i]) = \ClockValuation{B}(\chi_{i + |C_A|}) = \ClockValuation{A}(\chi_{i + |C_A|}).
\]
%
%
In other words, given a pair of AB-\synchronized \ states with equal virtual clock values, the first state
contains the current values of the original clocks of the second state in the shared set of virtual clocks and vice versa.
%
In the following proposition, we show that any transition 
from an AB-\synchronized \ state in a TSVC leads to an AB-\semisynchronized \ state.

\begin{proposition}[Outgoing Transitions of AB-\synchronized \ States]
\label{prop:virtual_clocks:TSVC:Outgoing-transitions-of-syncd-states}
Let $\TSVC{A}$ be the TSVC of TA $A$ regarding $B$, $\TSVC{B}$ be the TSVC of TA $B$ regarding $A$,
$\TSVCState{A}$ be an AB-\synchronized \ state of $\TSVC{A}$, and $\TSVCState{B}$ be an AB-\synchronized \ state of $\TSVC{B}$.
\begin{itemize}
    \item If $\TSVCState{A} \TSVCTrans{d} \TSVCState{d, A}$ with $d \in \TimeDomain$, then $\TSVCState{d, A}$ is AB-\synchronized.
	\item If $\TSVCState{A} \TSVCTrans{\sigma} \TSVCState{\sigma, A}$ with $\sigma \in \Sigma$, then  $\TSVCState{\sigma, A}$ is AB-\semisynchronized.
    \item If $\TSVCState{B} \TSVCTrans{d} \TSVCState{d, B}$ with $d \in \TimeDomain$, then $\TSVCState{d, B}$ is AB-\synchronized.
	\item If $\TSVCState{B} \TSVCTrans{\sigma} \TSVCState{\sigma, B}$ with $\sigma \in \Sigma$, then $\TSVCState{\sigma, B}$ is AB-\semisynchronized.
\end{itemize}
\end{proposition}
\begin{proof}
  Given any original clock $c$ with corresponding virtual clock $\chi$ and synchronized state $\TSVCState{}$. For any delay transition $\TSVCState{} \TSVCTrans{d} \TSVCState{d}$ the statement $\ClockValuation{d}(c) = \ClockValuation{}(c) + d = \ClockValuation{}(\chi) + d = \ClockValuation{d}(\chi)$ holds by Definition~\ref{def:background:TLTS} and Definition~\ref{def:virtal_clocks:TSVC}. Therefore, $\TSVCState{d}$ is AB-\synchronized. For any action transition $\TSVCState{} \TSVCTrans{\sigma} \TSVCState{\sigma}$ there exists a corresponding transition $\TATrans{}{g}{\sigma}{R}{\sigma}$ by Definition~\ref{def:background:TLTS} and Definition~\ref{def:virtal_clocks:TSVC}. If $c \in R$, then $\ClockValuation{\sigma}(c) = 0$, otherwise $\ClockValuation{\sigma}(c) = \ClockValuation{}(c) = \ClockValuation{}(\chi) = \ClockValuation{\sigma}(\chi)$. Therefore, $\TSVCState{\sigma}$ is AB-\semisynchronized.
\end{proof}

Given an AB-\semisynchronized \ state, we can obtain an AB-\synchronized{} state by resetting the virtual clocks that correspond to the resetted original clocks. We call this procedure \textit{synchronization}.

\begin{definition}[Synchronization of AB-\semisynchronized \ States]
\label{def:virtual_clocks:TSVC:sync-for-TSVC}
Let $\TSVC{A}$ be the TSVC of TA $A$ regarding $B$ with set of states $\TSVCAllStates{A}$, $\TSVC{B}$ be the TSVC of TA $B$ regarding $A$ with set of states $\TSVCAllStates{B}$,
$\TSVCFullState{A}{}$ be an AB-\semisynchronized \ state of $\TSVC{A}$ and $\TSVCFullState{B}{}$ be an AB-\semisynchronized \ state of $\TSVC{B}$
such that $\text{virt}(\ClockValuation{A}) = \text{virt}(\ClockValuation{B})$ holds.
We define the $\syncfunction$ function for states
\begin{equation*}
\syncfunction : \TSVCAllStates{A} \times \TSVCAllStates{B} \rightarrow \TSVCAllStates{A} \times \TSVCAllStates{B}
\end{equation*}
such that
\begin{equation*}
\syncfunction(\TSVCFullState{A}{}, \TSVCFullState{B}{}) = (\TSVCFullState{A}{e, A}, \TSVCFullState{B}{e, B})
\end{equation*}
with
\begin{alignat*}{2}
& \ClockValuation{e, A} : C_A \cup \{\chi_0, ..., \chi_{|C_A| + |C_B| - 1}\} \rightarrow \TimeDomain && \ClockValuation{e, B} : C_B \cup \{\chi_0, ..., \chi_{|C_A| + |C_B| - 1}\} \rightarrow \TimeDomain \\
& \text{orig}(\ClockValuation{e, A}) = \text{orig}(\ClockValuation{A}), && \text{orig}(\ClockValuation{e, B}) = \text{orig}(\ClockValuation{B})\\
& \text{virt-A}(\ClockValuation{e, A})(\chi_i) = \left\{
\begin{array}{ll}
\ClockValuation{A}(\chi_i), \\
\ \ \text{ if } \ClockValuation{A}(C_A[i]) \neq 0, \\
0, \text{ else, }
\end{array}
\right. 
&& \text{virt-A}(\ClockValuation{e, B})(\chi_i) = \left\{
\begin{array}{ll}
\ClockValuation{B}(\chi_i), \\
\ \ \text{ if } \ClockValuation{A}(C_A[i]) \neq 0, \\
0, \text{ else, }
\end{array}
\right. \\
& \text{virt-B}(\ClockValuation{e, A})(\chi_{|C_A| + i}) = \left\{
\begin{array}{ll}
\ClockValuation{A}(\chi_{|C_A| + i}), \\
\ \ \text{ if } \ClockValuation{B}(C_B[i]) \neq 0, \quad\\
0, \text{ else, }
\end{array}
\right. 
&& \text{virt-B}(\ClockValuation{e, B})(\chi_{|C_A| + i}) = \left\{
\begin{array}{ll}
\ClockValuation{B}(\chi_{|C_A| + i}), \\
\ \ \text{ if } \ClockValuation{B}(C_B[i]) \neq 0, \\
0, \text{ else. }
\end{array}
\right.
\end{alignat*}

\end{definition}

The following proposition shows that the $\syncfunction$ function converts AB-\semisynchronized \ states into AB-\synchronized \ states as expected.

\begin{proposition}[Properties of $\syncfunction$ for States]
\label{prop:virtual_clocks:TSVC:Properties-of-sync-for-states}
Let $\TSVC{A}$ be the TSVC of TA $A$ regarding $B$, $\TSVC{B}$ be the TSVC of TA $B$ regarding $A$,
let $\TSVCFullState{A}{}$ be an AB-\semisynchronized \ state of $\TSVC{A}$ and $\TSVCFullState{B}{}$ be an AB-\semisynchronized \ state of $\TSVC{B}$ such that $\text{virt}(\ClockValuation{A}) = \text{virt}(\ClockValuation{B})$ holds.
Then, it holds that 
\[
\syncfunction(\TSVCFullState{A}{}, \TSVCFullState{B}{}) = (\TSVCFullState{A}{e, A}, \TSVCFullState{B}{e, B})
\]
is a pair of AB-\synchronized \ states with $\TSVCFullState{A}{e, A}$ being a state of $\TSVC{A}$, $\TSVCFullState{B}{e, B}$ being a state of $\TSVC{B}$, and $\text{virt}(\ClockValuation{e, A}) = \text{virt}(\ClockValuation{e, B})$.
\end{proposition}

\begin{proof}
We first show that $\TSVCFullState{A}{e, A}$ and $\TSVCFullState{B}{e, B}$ are AB-\synchronized. Since we assumed $\TSVCFullState{A}{}$ to be AB-\semisynchronized, we know for any $i \in [0, |C_A| - 1]$ either $\ClockValuation{A}(C_A[i]) = \ClockValuation{A}(\chi_i)$ or $\ClockValuation{A}(C_A[i]) = 0$. We need to show that for any $i \in [0, |C_A| - 1]$ the statement $\ClockValuation{e, A}(C_A[i]) = \ClockValuation{e, A}(\chi_i)$ holds. If $\ClockValuation{A}(C_A[i]) \neq 0$, the $\syncfunction$ does not change the value of the corresponding virtual clock and, therefore, $\ClockValuation{e, A}(C_A[i]) = \ClockValuation{A}(C_A[i]) = \ClockValuation{A}(\chi_i) = \ClockValuation{e, A}(\chi_i)$ holds. If $\ClockValuation{A}(C_A[i]) = 0$, then $\ClockValuation{e, A}(C_A[i]) = 0$ and $\ClockValuation{e, A}(\chi_i) = 0$ hold by Definition~\ref{def:virtual_clocks:TSVC:sync-for-TSVC}. Therefore, $\TSVCFullState{A}{e, A}$ is AB-synchronized and the same can be shown for $\TSVCFullState{B}{e, B}$.

We now show $\text{virt}(\ClockValuation{e, A}) = \text{virt}(\ClockValuation{e, B})$. For any $i \in [0, |C_A| - 1]$, either $\ClockValuation{A}(C_A[i]) = 0$, which implies $\ClockValuation{e, A}(\chi_i) = 0 = \ClockValuation{e, B}(\chi_i)$ or $\ClockValuation{A}(C_A[i]) \neq 0$, which implies $\ClockValuation{e, A}(\chi_i) = \ClockValuation{A}(\chi_i) = \ClockValuation{B}(\chi_i) = \ClockValuation{e, B}(\chi_i)$ by $\text{virt}(\ClockValuation{A}) = \text{virt}(\ClockValuation{B})$. Analogously for any $\chi_{i + |C_A|}$ with $i \in [0, |C_B| - 1]$.
\end{proof}

If the inputs are already AB-\synchronized, then the $\syncfunction$ returns them unchanged. Since AB-\synchronized \ states already encode the unique virtual clock valuation consistent with their original clocks, no value is changed.
%
%
%
%
%
%
%
The following example shows the impact of the $\syncfunction$ function.
\begin{example}
\label{ex:virtual_clocks:TSVCs:impact-of-sync-function}
\newcommand{\VirtualClockTSVCExTSVCDistance}{6}
\newcommand{\VirtualClockTSVCTikzFontSize}{\Large}
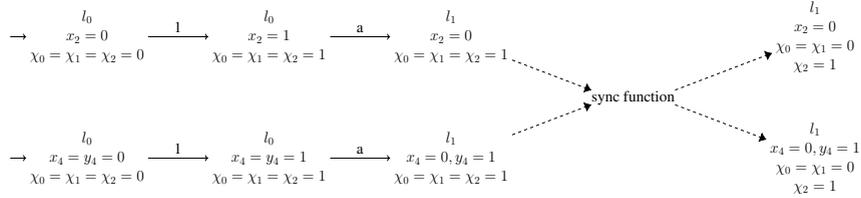
\begin{figure}
\centering
\scalebox{0.5}{
\begin{tikzpicture}
\tikzstyle{every node}=[font=\VirtualClockTSVCTikzFontSize]
\node (l00) [rectangle, initial, initial text = {}, align=center] at (0, 1.4) {$l_0$\\$x_2 = 0$\\$\chi_0 = \chi_1 = \chi_2 = 0$};
\node (l01) [rectangle, align=center] at (\VirtualClockTSVCExTSVCDistance, 1.4) {$l_0$\\$x_2 = 1$\\$\chi_0 = \chi_1 = \chi_2 = 1$};
\node (l02) [rectangle, align=center] at (2*\VirtualClockTSVCExTSVCDistance, 1.4) {$l_1$\\$x_2 = 0$\\$\chi_0 = \chi_1 = \chi_2 = 1$};
\node (l03) [rectangle, align=center] at (4*\VirtualClockTSVCExTSVCDistance, 1.4) {$l_1$\\$x_2 = 0$\\$\chi_0 = \chi_1 = 0 \land \chi_2 = 1$};
\node (l10) [rectangle, initial, initial text = {}, align=center] at (0, -1.4) {$l_0$\\$x_4 = y_4 = 0$\\$\chi_0 = \chi_1 = \chi_2 = 0$};
\node (l11) [rectangle, align=center] at (\VirtualClockTSVCExTSVCDistance, -1.4) {$l_0$\\$x_4 = y_4 = 1$\\$\chi_0 = \chi_1 = \chi_2 = 1$};
\node (l12) [rectangle, align=center] at (2*\VirtualClockTSVCExTSVCDistance, -1.4) {$l_1$\\$x_4 = 0, y_4 = 1$\\$\chi_0 = \chi_1 = \chi_2 = 1$};
\node (l13) [rectangle, align=center] at (4*\VirtualClockTSVCExTSVCDistance, -1.4) {$l_1$\\$x_4 = 0, y_4 = 1$\\$\chi_0 = \chi_1 = 0 \land \chi_2 = 1$};
\node (syncnode) [rectangle, thick, align=center, inner sep=0pt, minimum size=12pt] at (3*\VirtualClockTSVCExTSVCDistance, 0) {$\syncfunction$ function};
\draw[->] (l00) -- (l01) node[anchor=south, midway]{1};
\draw[->] (l01) -- (l02) node[anchor=south, midway]{a};
\draw[dashed, -{Latex[length=2mm,width=3mm]}] (l02) -- (syncnode.north west);
\draw[dashed, -{Latex[length=2mm,width=3mm]}] (syncnode.north east) -- (l03);
\draw[->] (l10) -- (l11) node[anchor=south, midway]{1};
\draw[->] (l11) -- (l12) node[anchor=south, midway]{a};
\draw[dashed, -{Latex[length=2mm,width=3mm]}] (l12) -- (syncnode.south west);
\draw[dashed, -{Latex[length=2mm,width=3mm]}] (syncnode.south east) -- (l13);
\end{tikzpicture}
}
\caption{Extracts of the TSVC of $A_2$ regarding $A_4$ and $A_4$ regarding $A_2$}
\label{figure:virtual-clocks:TSVC:example:sync-a2-reg-a4}
\end{figure}

Figure \ref{figure:virtual-clocks:TSVC:example:sync-a2-reg-a4} shows (small) extracts from the TSVC of $A_2$ regarding $A_4$ and the TSVC $A_4$ regarding $A_2$ (Figure \ref{fig:examples}). Since we require the set of clocks to be disjoint, we rename the original clocks to $x_2$, $x_4$, and $y_4$. Since there are three original clocks, we add three virtual clocks, $\chi_0$, $\chi_1$, and $\chi_2$. Initially, all clocks are set to zero, which implies that both initial states are $A_2A_4$-\synchronized \ and $\text{virt}(\ClockValuation{0, A_2}) = \text{virt}(\ClockValuation{0, A_4})$ holds. After a delay transition, all clock values are increased and the previously described properties still hold. After the action transition, $x_2$ and $x_4$ are reset. While the virtual parts of the clock valuations are still equivalent, due to $\chi_0 = 1 \neq 0 = x_2$ and $\chi_1 = 1 \neq 0 = x_4$ the states are not $A_2A_4$-\synchronized \ but only $A_2A_4$-\semisynchronized. Finally, the $\syncfunction$ function is used to reset the corresponding virtual clocks and the result is $A_2A_4$-\synchronized, again.
\end{example}

Since any invariant, guard, or reset of a TA refers to an original clock, but the $\syncfunction$ changes only virtual clocks, the $\syncfunction$ function does not have any impact on the overall structure of the TLTS, but changes carried information only. Therefore, it is intuitively clear that we can define timed bisimulation equivalent to Definition \ref{def:background:sim-and-bisim:Strong-Timed-Bisimulation} but using TSVC instead of TLTS.

\begin{definition}[Virtual Bisimulation using TSVC]
\label{def:virtual-clocks:BoundedTimedBiSimulationForTSVCs}
\newcommand{\TSVCBiSimulateTransitionCommonPart}[5] {%
the existence of a transition $\TSVCFullState{#4}{#1} \TSVCTrans{\mu} \TSVCFullState{\mu, #4}{\mu, #1}$ 
implies the existence of a transition $\TSVCFullState{#5}{#2} \TSVCTrans{\mu} \TSVCFullState{\mu, #5}{\mu, #2}$ with $#3{\TSVCFullState{\mu, #4}{\mu, #1}}{\TSVCFullState{\mu, #5}{\mu, #2}}{n}$%
}%
Assume two TA $A$, $B$ with TSVC $\TSVC{A}$ of $A$ regarding $B$ and TSVC $\TSVC{B}$ of $B$ regarding $A$ using the same alphabet $\Sigma$. Let $\TSVCFullState{A}{}$ be an AB-\semisynchronized \ state of $\TSVC{A}$ and let $\TSVCFullState{B}{}$ be an AB-\semisynchronized \ state of $\TSVC{B}$.

We define $\TSVCState{A}$ to be virtually bisimilar in order $n=0$ to $\TSVCState{B}$, denoted  $\TSVCBisimulation{\TSVCState{A}}{\TSVCState{B}}{0}$, if and only if $\text{virt}(\ClockValuation{A}) = \text{virt}(\ClockValuation{B})$ holds.

We define $\TSVCState{A}$ to be virtually bisimilar in order n+1 (with $n \in \mathbb{N}^{\geq 0}$) to $\TSVCState{B}$, denoted $\TSVCBisimulation{\TSVCState{A}}{\TSVCState{B}}{n+1}$, if and only if $\text{virt}(\ClockValuation{A}) = \text{virt}(\ClockValuation{B})$ and for the AB-\synchronized \ states
\begin{equation*}
(\TSVCFullState{A}{e, A}, \TSVCFullState{A}{e, B}) = \syncfunction(\TSVCState{A}, \TSVCState{B})
\end{equation*}
\TSVCBiSimulateTransitionCommonPart{e, A}{e, B}{\TSVCBisimulation}{A}{B} and \TSVCBiSimulateTransitionCommonPart{e, B}{e, A}{\TSVCBisimulation}{B}{A}.

$\TSVC{A}$ is virtually bisimilar in order n to $\TSVC{B}$, denoted $\TSVCBisimulation{\TSVC{A}}{\TSVC{B}}{n}$, if and only if $\TSVCBisimulation{\TSVCState{0, A}}{\TSVCState{0, B}}{n}$ holds. $\TSVC{A}$ is virtually bisimilar to $\TSVC{B}$, denoted $\TSVCBisimulation{\TSVC{A}}{\TSVC{B}}{}$, if and only if $\forall n \in \mathbb{N}^{\geq 0} : \TSVCBisimulation{\TSVC{A}}{\TSVC{B}}{n}$ holds.
\end{definition}

The following theorem shows that it is sound to use Definition \ref{def:virtual-clocks:BoundedTimedBiSimulationForTSVCs} instead of Definition \ref{def:virtual-clocks:bounded-timed-bisim-for-TLTS} when checking for timed bisimilarity of two TA.
\begin{theorem}[Virtual Bisimulation is equivalent to Timed Bisimulation]
\label{theorem:virtual-clocks:TSVCSimIffTASim}
Let $\TSVC{A}$ be the TSVC of TA $A$ regarding $B$, $\TSVC{B}$ be the TSVC of TA $B$ regarding $A$. $A \sim_n B$ holds if and only if $\TSVCBisimulation{\TSVC{A}}{\TSVC{B}}{n}$ holds. $A \sim B$ is equivalent to $\TSVCBisimulation{\TSVC{A}}{\TSVC{B}}{}$.
\end{theorem}
\begin{proof}
If $A \sim_n B$ holds if and only if $\TSVCBisimulation{\TSVC{A}}{\TSVC{B}}{n}$ holds, the second part of the theorem is trivially implied. Therefore, we prove only this statement by showing two hypothesis regarding the relationship between states of the TLTS and states of the TSVC. We show both hypothesis by induction in $n$.

\noindent \textbf{Hypothesis 1:} Given $\TLTSFullState{A}{A, \text{TS}}$ of $\TLTS{A}$ and $\TLTSFullState{B}{B, \text{TS}}$ of $\TLTS{B}$, such that $\TLTSFullState{A}{A, \text{TS}} \sim_n \TLTSFullState{B}{B, \text{TS}}$, then for any two AB-\semisynchronized \ states of the corresponding TSVC, $\TSVCFullState{A}{A, \text{TV}}$ of $\TSVC{A}$ and $\TSVCFullState{B}{B, \text{TV}}$ of $\TSVC{B}$ with $\ClockValuation{A, \text{TS}} = \text{orig}(\ClockValuation{A, \text{TV}})$, $\ClockValuation{B, \text{TS}} = \text{orig}(\ClockValuation{B, \text{TV}})$, and $\text{virt}(\ClockValuation{A, \text{TV}}) = \text{virt}(\ClockValuation{B, \text{TV}})$, the statement $\TSVCBisimulation{\TSVCFullState{A}{A, \text{TV}}}{\TSVCFullState{B}{B, \text{TV}}}{n}$ holds.

\noindent\textbf{Base Case:} If $n = 0$, then the statement trivially holds due to the fact that any two states with equivalent virtual clock valuations are virtually bisimilar in order $0$.

\noindent\textbf{Induction Step:} To show the induction step, we have to show the two conditions regarding the transitions of the synchronized states (see Definition~\ref{def:virtual-clocks:BoundedTimedBiSimulationForTSVCs}). Due to $\text{virt}(\ClockValuation{A, \text{TV}}) = \text{virt}(\ClockValuation{B, \text{TV}})$, we are allowed to apply the $\syncfunction$ and denote $(\TSVCFullState{A}{e, A}, \TSVCFullState{B}{e, B}) = \syncfunction(\TSVCFullState{A}{A, \text{TV}}, \TSVCFullState{B}{B, \text{TV}})$. To show the first condition, we assume a transition $\TSVCFullState{A}{e, A} \TSVCTrans{\mu} \TSVCFullState{\mu, A}{\mu, A}$ and show that the existence of a transition $\TSVCFullState{B}{e, B} \TSVCTrans{\mu} \TSVCFullState{\mu, B}{\mu, B}$ with $\TSVCBisimulation{\TSVCFullState{\mu, A}{\mu, A}}{\TSVCFullState{\mu, B}{\mu, B}}{n-1}$ is implied. 

If $\mu \in \TimeDomain$, $\ClockValuation{\mu, A} = \ClockValuation{e, A} + \mu$. By Definition~\ref{def:background:TLTS} and Definition~\ref{def:virtal_clocks:TSVC}, this implies $\ClockValuation{e, A} + \mu \models I(l_A)$. Since $I(l_A) \in \mathcal{B}(C_A)$and $\text{orig}(\ClockValuation{e, A}) = \text{orig}(\ClockValuation{A, \text{TV}}) = \ClockValuation{A, \text{TS}}$ by Definition~\ref{def:virtual_clocks:TSVC:sync-for-TSVC}, $\ClockValuation{A, \text{TS}} + \mu \models I(l_A)$. By Definition~\ref{def:background:TLTS}, this implies the existence of a transition $(l_A, \ClockValuation{A, \text{TS}}) \TLTSTrans{\mu} (l_A, \ClockValuation{A, \text{TS}} + \mu)$. Since $\TLTSFullState{A}{A, \text{TS}} \sim_n \TLTSFullState{B}{B, \text{TS}}$, this implies the existence of a transition $(l_B, \ClockValuation{B, \text{TS}}) \TLTSTrans{\mu} (l_B, \ClockValuation{B, \text{TS}} + \mu)$ with $(l_A, \ClockValuation{A, \text{TS}} + \mu) \sim_{n-1} (l_B, \ClockValuation{B, \text{TS}} + \mu)$. By Definition~\ref{def:background:TLTS}, this implies $\ClockValuation{B, \text{TS}} + \mu \models I(l_B)$ and due to $\text{orig}(\ClockValuation{e, B}) = \ClockValuation{B, \text{TS}}$ and $I(l_B) \in \mathcal{B}(C_B)$, this implies $\ClockValuation{e, B} + \mu \models I(l_B)$. Therefore, there is a transition $\TSVCFullState{B}{e, B} \TSVCTrans{\mu} \TSVCFullState{\mu, B}{\mu, B}$ and $\TSVCBisimulation{\TSVCFullState{\mu, A}{\mu, A}}{\TSVCFullState{\mu, B}{\mu, B}}{n-1}$ by the induction hypothesis.

If $\mu \in \Sigma$, there exists a transition $\TATrans{A}{g_A}{\mu}{R_A}{\mu, A}$ with $\ClockValuation{\mu, A, \text{TV}} = [R_A \rightarrow 0]\ClockValuation{e, A}$, $\ClockValuation{e, A} \models g_A$, and $[R_A \rightarrow 0]\ClockValuation{e, A} \models I(l_{\mu, A})$ by Definition \ref{def:virtal_clocks:TSVC} and Definition \ref{def:background:TLTS}. Since $\text{orig}(\ClockValuation{e, A}) = \ClockValuation{A, \text{TS}}$ and $g_A, I(l_{\mu, A}) \in \mathcal{B}(C_A)$, we also know $\ClockValuation{A, \text{TS}} \models g_A$, $[R_A \rightarrow 0]\ClockValuation{A, \text{TS}} \models I(l_{\mu, A})$ and, therefore, there exists a transition $(l_A, \ClockValuation{A, \text{TS}}) \TLTSTrans{\mu} (l_{\mu, A}, \ClockValuation{\mu, A, \text{TS}})$ by Definition \ref{def:background:TLTS}. Due to $(l_A, \ClockValuation{A, \text{TS}}) \sim_{n} (l_B, \ClockValuation{B, \text{TS}})$, there exists a transition $(l_B, \ClockValuation{B, \text{TS}}) \TLTSTrans{\mu} (l_{\mu, B}, \allowbreak \ClockValuation{\mu, B, \text{TS}})$ with $(l_{\mu, A}, \ClockValuation{\mu, A, \text{TS}}) \sim_{n-1} (l_{\mu, B}, \ClockValuation{\mu, B, \text{TS}})$ by Definition \ref{def:virtual-clocks:bounded-timed-bisim-for-TLTS}. This implies the existence of a transition $\TATrans{B}{g_B}{\mu}{R_B}{\mu, B}$. $\ClockValuation{B, \text{TS}} \models g_B$, $\ClockValuation{\mu, B, \text{TS}} \models I(l_{\mu, B})$, and $\ClockValuation{\mu, B, \text{TS}} = [R_B \rightarrow 0]\ClockValuation{B, \text{TS}}$ by Definition \ref{def:background:TLTS}. Due to $\ClockValuation{B, \text{TS}} = \text{orig}(\ClockValuation{e, B})$ and $g_B, I(l_B) \in \mathcal{B}(C_B)$, we know $\ClockValuation{e, B} \models g_B$ and $[R_B \rightarrow 0]\ClockValuation{e, B} \models I(l_B)$. Therefore, there exists a transition $\TSVCFullState{B}{e, B} \TSVCTrans{\mu} \TSVCFullState{\mu, B}{\mu, B, \text{TV}}$ by Definition \ref{def:virtal_clocks:TSVC} and $\TSVCBisimulation{\TSVCFullState{\mu, A}{\mu, A}}{\TSVCFullState{\mu, B}{\mu, B}}{n-1}$ holds by the induction hypothesis.

Since the second condition can be shown analogously, hypothesis 1 is true. Now we show the second hypothesis.

\noindent\textbf{Hypothesis 2:} Given the AB-\semisynchronized{} states $\TSVCFullState{A}{A, \text{TV}}$ of $\TSVC{A}$ and $\TSVCFullState{B}{B, \text{TV}}$ of $\TSVC{B}$, such that $\TSVCBisimulation{\TSVCFullState{A}{A, \text{TV}}}{\TSVCFullState{B}{B, \text{TV}}}{n}$, then for any states of the corresponding TLTS, $\TLTSFullState{A}{A, \text{TS}}$ of $\TLTS{A}$ and $\TLTSFullState{B}{B, \text{TS}}$ of $\TLTS{B}$, with $\ClockValuation{A, \text{TS}} = \text{orig}(\ClockValuation{A, \text{TV}})$ and $\ClockValuation{B, \text{TS}} = \text{orig}(\ClockValuation{B, \text{TV}})$, the statement $\TLTSFullState{A}{A, \text{TS}} \sim_n \TLTSFullState{B}{B, \text{TS}}$ holds.

\noindent\textbf{Base Case:} If $n = 0$, then the statement trivially holds due to the fact that any two states are bisimilar in order $0$.

\noindent\textbf{Induction Step:} To show the induction step, we have to show the two conditions of Definition~\ref{def:virtual-clocks:bounded-timed-bisim-for-TLTS}. To show the first condition, we assume a transition $(l_A, \ClockValuation{A, \text{TS}}) \TLTSTrans{\mu} (l_{\mu, A}, \ClockValuation{\mu, A, \text{TS}})$ and show that the existence of a transition $(l_B, \ClockValuation{B, \text{TS}}) \TLTSTrans{\mu} (l_{\mu, B}, \ClockValuation{\mu, B, \text{TS}})$ with $\TLTSFullState{\mu, A}{\mu, A, \text{TS}} \sim_{n-1} \TLTSFullState{\mu, B}{\mu, B, \text{TS}}$ is implied. We denote $(\TSVCFullState{A}{e, A}, \TSVCFullState{B}{e, B}) = \text{sync}(\TSVCFullState{A}{A, \text{TV}}, \TSVCFullState{B}{B, \text{TV}})$ (remember that $\TSVCBisimulation{\TSVCFullState{A}{A, \text{TV}}}{\TSVCFullState{B}{B, \text{TV}}}{n}$ implies $\text{virt}( \allowbreak \ClockValuation{A, \text{TV}}) = \text{virt}(\ClockValuation{B, \text{TV}})$).

If $\mu \in \TimeDomain$, $\ClockValuation{\mu, A, \text{TS}} = \ClockValuation{A, \text{TS}} + \mu$, and $\ClockValuation{A, \text{TS}} + \mu \models I(l_A)$ by Definition \ref{def:background:TLTS}. Due to $\ClockValuation{A, \text{TS}} = \text{orig}(\ClockValuation{A, \text{TV}}) = \text{orig}(\ClockValuation{e, A})$, $\ClockValuation{e, A} + \mu \models I(l_A)$ also holds, which implies a transition $\TSVCFullState{A}{e, A} \TSVCTrans{\mu} (l_A, \ClockValuation{e, A} + \mu)$ by Definition \ref{def:virtal_clocks:TSVC}. Therefore, there exists a transition $\TSVCFullState{B}{e, B} \TSVCTrans{\mu} (l_B, \ClockValuation{e, B} + \mu)$ with $\TSVCBisimulation{(l_A, \ClockValuation{e, A} + \mu)}{(l_B, \ClockValuation{e, B} + \mu)}{n-1}$ by Definition \ref{def:virtual-clocks:BoundedTimedBiSimulationForTSVCs}. By Definition \ref{def:virtal_clocks:TSVC}, this implies $\ClockValuation{B, \text{TV}} + \mu \models I(l_B)$. Due to $\ClockValuation{B, \text{TS}} = \text{orig}( \allowbreak \ClockValuation{B, \text{TV}}) = \text{orig}(\ClockValuation{e, B})$, $\ClockValuation{B, \text{TS}} + \mu \models I(l_B)$ and, therefore, there exists a transition $(l_B, \ClockValuation{B, \text{TS}}) \allowbreak \TLTSTrans{\mu} (l_{B}, \ClockValuation{B, \text{TS}} + \mu)$ by Definition \ref{def:background:TLTS}. $(l_{A}, \ClockValuation{A, \text{TS}} + \mu) \sim_{n-1} (l_{B}, \ClockValuation{B, \text{TS}} + \mu)$ holds by the induction hypothesis.

If $\mu \in \Sigma$, there exists a transition $\TATrans{A}{g_A}{\mu}{R_A}{\mu, A}$ with $\ClockValuation{\mu, A, \text{TS}} = [R_A \rightarrow 0]\ClockValuation{A, \text{TS}}$, $\ClockValuation{A, \text{TS}} \models g_A$, and $\ClockValuation{\mu, A, \text{TS}} \models I(l_{\mu, A})$ by Definition \ref{def:background:TLTS}. Due to $\ClockValuation{A, \text{TS}} = \text{orig}(\ClockValuation{A, \text{TV}}) = \text{orig}(\ClockValuation{e, A})$, $\ClockValuation{e, A} \models g_A$ and $[R_A \rightarrow 0]\ClockValuation{e, A} \models I(l_A)$ also hold. Therefore, there exists a transition $\TSVCFullState{A}{e, A} \TSVCTrans{\mu} \TSVCFullState{\mu, A}{\mu, A, \text{TV}}$ with $\ClockValuation{\mu, A, \text{TV}} = [R_A \rightarrow 0]\ClockValuation{e, A}$ by Definition \ref{def:virtal_clocks:TSVC}, which implies the transition $\TSVCFullState{B}{e, B} \TSVCTrans{\mu} \TSVCFullState{\mu, B}{\mu, B, \text{TV}}$ with $\TSVCBisimulation{\TSVCFullState{\mu, A}{\mu, A, \text{TV}}}{\TSVCFullState{\mu, B}{\mu, B, \text{TV}}}{n-1}$ by Definition \ref{def:virtual-clocks:BoundedTimedBiSimulationForTSVCs}, which implies the existence of a transition $\TATrans{B}{g_B}{\mu}{R_B}{\mu, B}$, $\ClockValuation{\mu, B, \text{TV}} = [R_B \rightarrow 0]\ClockValuation{e, B}$, $\ClockValuation{e, B} \models g_B$, and $\ClockValuation{\mu, B, \text{TV}} \models I(l_B)$. Due to $\ClockValuation{B, \text{TS}} = \text{orig}(\ClockValuation{B, \text{TV}})$ and $g_B, I(l_B) \in \mathcal{B}(C_B)$, $\ClockValuation{B, \text{TS}} \models g_B$, $[R_B \rightarrow 0]\ClockValuation{B, \text{TS}} \models I(l_B)$ hold and, therefore, there exists a transition $(l_B, \ClockValuation{B, \text{TS}}) \TLTSTrans{\mu} (l_{\mu, B}, \ClockValuation{\mu, B, \text{TS}})$ with $\ClockValuation{\mu, B, \text{TS}} = [R_B \rightarrow 0]\ClockValuation{B, \text{TS}}$ by Definition \ref{def:background:TLTS}. $\TLTSFullState{\mu, A}{\mu, A, \text{TS}} \sim_{n-1} (l_{B}, \ClockValuation{B, \text{TS}} + \mu)$ holds by the induction hypothesis.

Since the second condition can be shown analogously, hypothesis 2 is also true. For the initial states of TSVC and TLTS, the previous conditions obviously hold since the values of all clocks are set to zero. Therefore, Theorem~\ref{theorem:virtual-clocks:TSVCSimIffTASim} is implied by hypothesis 1, hypothesis 2, Definition~\ref{def:virtual-clocks:bounded-timed-bisim-for-TLTS}, and Definition~\ref{def:virtual-clocks:BoundedTimedBiSimulationForTSVCs}.
\end{proof}

Analogously to TLTS, TSVC are also infinite and cannot be used to effectively check for timed bisimulation. However, the following example shows how checking for timed bisimulation using TSVC works in theory.
\begin{example}
We reuse the extracts of the TSVCs shown in Figure \ref{figure:virtual-clocks:TSVC:example:sync-a2-reg-a4}. Some states that appear in the figure actually have an infinite number of outgoing transitions. However, we only consider the transitions shown to illustrate the construction.

To check for bisimulation of the initial states, we first have to show that the used clock valuations have the same virtual part. This is the case since all virtual clocks are set to zero. Since both initial states are $A_2A_4$-\synchronized, applying the $\syncfunction$ function does not change the states. Therefore, we only have to check, whether for each of their outgoing transitions there exists an outgoing transition at the other state using the same label such that the target states are bisimilar. Since we only consider the transitions shown in the extract, we proceed by using the transitions labeled with 1 in both TSVC.

Once again, the virtual parts of the clock valuations are the same and since both states are $A_2A_4$-\synchronized, we proceed by using the transition labeled with $a$.

We first recognize that the virtual parts of the clock valuations are equal. However, the states are $A_2A_4$-\semisynchronized \ but not $A_2A_4$-\synchronized. Therefore, we have to check the outgoing transitions of the output of the $\syncfunction$ function. Since the extract shows none of them, and we only consider those transitions shown in the extract, we are done.
Of course, if we would consider all transitions of the TSVC, we would recognize that some of them would violate Definition \ref{def:virtual-clocks:BoundedTimedBiSimulationForTSVCs} due to the same reasons we have already discussed in Example \ref{ex:background:timed-bisim}.
\end{example}

We next analyze the effect of virtual clocks on zone graphs.

%% file: sections/virtual_clock_chapter/VCGs.tex
\subsection{Zone Graphs with Virtual Clocks}

This section shows how virtual clocks are integrated into zone graphs.

\begin{definition}[Virtual Clock Graph (VCG)]
\label{def:virtual-clocks:VCG}
Let $A$ and $B$ be two TA. 
We define the \textit{Virtual Clock Graph (VCG)} of $A$ regarding $B$ to be the zone graph of $\text{add-virtual}(A, B)$.
\end{definition}

Since VCGs are zone graphs, all properties described in Section~\ref{sec:background} hold. Since TSVC are the TLTS of $\text{add-virtual}(A, B)$ and VCGs are the zone graphs of $\text{add-virtual}(A, B)$, VCGs are forward- and backward-stable to TSVC as described in Proposition~\ref{prop:background:backward-stability} and Proposition~\ref{prop:background:forward-stability}.

Analogous to the definition in Section \ref{subsection:virtual-clocks:TSVCs}, we define two classes of symbolic states. AB-\synchronized \ symbolic states represent symbolic states where all included states are AB-\synchronized.
AB-\semisynchronized \ symbolic states represent former AB-\synchronized \ symbolic states where a set of original clocks has been reset. Therefore, for any original clock the value is either equivalent to the corresponding virtual clock for all contained clock valuations, or the original clock has been reset for all contained clock valuations.

\begin{definition}[Classes of Symbolic States]
\label{def:virtual-clocks:VCG:Important-classes-of-symbolic-states}
Let $\VCG{A}$ be the VCG of TA $A$ regarding $B$ and $\VCG{B}$ be the VCG of TA $B$ regarding $A$. We assume a symbolic state $\VCGFullState{A}{}$ of $\VCG{A}$ with $\Zone{A} \neq \emptyset$ of $\VCG{A}$ and a symbolic state $\VCGFullState{B}{}$ of $\VCG{B}$ with $\Zone{B} \neq \emptyset$.
\begin{itemize}
	\item $\VCGFullState{A}{}$ is \textit{AB-\semisynchronized}, if and only if
\begin{equation*}
\forall i \in [0, |C_A| - 1] : ((\forall \ClockValuation{A} \in \Zone{A} : \ClockValuation{A}(C_A[i]) = \ClockValuation{A}(\chi_i)) \lor (\forall \ClockValuation{A} \in \Zone{A} : \ClockValuation{A}(C_A[i]) = 0))\text{.}
\end{equation*}

	\item $\VCGFullState{A}{}$ is \textit{AB-\synchronized}, if and only if
\begin{equation*}
\forall i \in [0, |C_A| - 1] : \forall \ClockValuation{A} \in \Zone{A} : \ClockValuation{A}(C_A[i]) = \ClockValuation{A}(\chi_i)\text{.}
\end{equation*}

	\item $\VCGFullState{B}{}$ is \textit{AB-\semisynchronized}, if and only if
\begin{equation*}
\forall i \in [0, |C_B| - 1] : ((\forall \ClockValuation{B} \in \Zone{B} : \ClockValuation{B}(C_B[i]) = \ClockValuation{B}(\chi_{i + |C_A|})) \lor (\forall \ClockValuation{B} \in \Zone{B} : \ClockValuation{B}(C_B[i]) = 0)).
\end{equation*}

	\item $\VCGFullState{B}{}$ is \textit{AB-\synchronized}, if and only if
\begin{equation*}
\forall i \in [0, |C_B| - 1] : \forall \ClockValuation{B} \in \Zone{B} : \ClockValuation{B}(C_B[i]) = \ClockValuation{B}(\chi_{i + |C_A|}).
\end{equation*}
\end{itemize}
\end{definition}

If a symbolic state is AB-\synchronized \ then it is also AB-\semisynchronized. 
Moreover, for any AB-\semisynchronized \ symbolic state $\VCGState{}$, any state $\TSVCState{}$ 
with $\TSVCState{} \in \VCGState{}$ is AB-\semisynchronized \ and the same holds for any AB-\synchronized \ symbolic state. For the outgoing transitions of an AB-\synchronized \ symbolic state, the analog to Proposition \ref{prop:virtual_clocks:TSVC:Outgoing-transitions-of-syncd-states} holds.

\begin{proposition}[Outgoing Transitions of AB-\synchronized \ Symbolic States]
\label{prop:virtual-clocks:VCG:outgoing-transitions-of-symbolic-syncd-states}
Let $\VCG{A}$ be the VCG of $A$ regarding $B$, $\VCG{B}$ be the VCG of $B$ regarding $A$, 
$\VCGState{A}$ be an AB-\synchronized \ symbolic state of $\VCG{A}$
and $\VCGState{B}$ be an AB-\synchronized \ symbolic state of $\VCG{B}$.
\begin{itemize}
  \item If $\VCGState{A} \VCGTrans{\varepsilon} \VCGState{\varepsilon, A}$, then $\VCGState{\varepsilon, A}$ is AB-\synchronized.
	\item If $\VCGState{A} \VCGTrans{\sigma} \VCGState{\sigma, A}$ with $\sigma \in \Sigma$, then  $\VCGState{\sigma, A}$ is AB-\semisynchronized.
  \item If $\VCGState{B} \VCGTrans{\varepsilon} \VCGState{\varepsilon, B}$, then $\VCGState{\varepsilon, B}$ is AB-\synchronized.
	\item If $\VCGState{B} \VCGTrans{\sigma} \VCGState{\sigma, B}$ with $\sigma \in \Sigma$, then $\VCGState{\sigma, B}$ is AB-\semisynchronized.
\end{itemize}
\end{proposition}

\begin{proof}
Given any original clock $c$ with corresponding virtual clock $\chi$ and synchronized symbolic state $(l, \Zone{})$. For the $\varepsilon$-transition $(l, \Zone{}) \VCGTrans{\varepsilon} (l, \Zone{\varepsilon})$, we know that for any $(l, \ClockValuation{\varepsilon}) \in (l, \Zone{\varepsilon})$ exists a $(l, \ClockValuation{}) \in (l, \Zone{})$ and a $d \in \TimeDomain$ such that $(l, \ClockValuation{}) \TSVCTrans{d} (l, \ClockValuation{\varepsilon})$ by Proposition~\ref{prop:background:backward-stability}. Since $(l, \Zone{})$ is synchronized, $(l, \ClockValuation{})$ is also synchronized and by Proposition~\ref{prop:virtual_clocks:TSVC:Outgoing-transitions-of-syncd-states}, $(l, \ClockValuation{\varepsilon})$ is also synchronized. Therefore, for any original clock $c$ with corresponding clock $\chi$ and for any clock valuation $\ClockValuation{\varepsilon} \in \Zone{\varepsilon}$, $\ClockValuation{\varepsilon}(c) = \ClockValuation{\varepsilon}(\chi)$ and the first and the third statement hold. For any action transition $(l, \Zone{}) \VCGTrans{\sigma} (l_\sigma, \Zone{\sigma})$ we know that for any $(l, \ClockValuation{\sigma}) \in (l, \Zone{\sigma})$ exists a $(l, \ClockValuation{}) \in (l, \Zone{})$ such that $(l, \ClockValuation{}) \TSVCTrans{\sigma} (l_{\sigma}, \ClockValuation{\sigma})$ by Proposition~\ref{prop:background:backward-stability}. Moreover, there exists a transition $\TATrans{}{g}{\sigma}{R}{\sigma}$. If $c \in R$, then for all $\ClockValuation{\sigma} \in \Zone{\sigma}$, $\ClockValuation{\sigma}(c) = 0$. Otherwise, $c$ and $\chi$ are unchanged and for all clock valuations $\ClockValuation{\sigma}(c) = \ClockValuation{}(c) = \ClockValuation{}(\chi) = \ClockValuation{\sigma}(\chi)$ holds. Therefore, the second and fourth statement hold.
\end{proof}

Definition~\ref{def:virtual_clocks:TSVC:sync-for-TSVC} has the precondition that the virtual clock values of the two states to be synchronized are equivalent. Thus, before we can apply any $\syncfunction$ function to symbolic states, we have to check that each state included in these symbolic states has a corresponding state with equivalent virtual clock values in the other symbolic state.

\begin{definition}[Virtual Inclusion]
\label{def:virtual-clocks:VCG:VirtualComparisonOfZones}
Let $C_A$, $C_B$ be two sets of clocks, $n \in \mathbb{N}^{\geq 0}$, $\Zone{A} \in \mathcal{D}(C_A \cup \{\chi_0, ..., \chi_{n}\})$, and $\Zone{B} \in \mathcal{D}(C_B \cup \{\chi_0, ..., \chi_{n}\})$ be two zones.
Virtual inclusion of zones, denoted $\VirtualComparison{A}{B}$, holds if and only if $\forall \ClockValuation{A} \in \Zone{A} : \exists \ClockValuation{B} \in \Zone{B} : \text{virt}(\ClockValuation{A}) = \text{virt}(\ClockValuation{B})$.
Virtual equivalence of zones, denoted $\VirtualEquivalence{A}{B}$, holds if and only if $\VirtualComparison{A}{B}$ and $\VirtualComparison{B}{A}$.
\end{definition}
To understand the following $\syncfunction$ function definition for symbolic states, it is crucial that for an AB-\semisynchronized \ symbolic state $\VCGState{}$ with $\TSVCFullState{}{1}, \TSVCFullState{}{2} \in \VCGState{}$, $\text{virt}(\ClockValuation{1}) = \text{virt}(\ClockValuation{2})$ implies $\ClockValuation{1} = \ClockValuation{2}$
%
%
%
In other words, in an AB-\semisynchronized \ symbolic state, the virtual clock values uniquely identify a particular state. We define the $\syncfunction$ function for symbolic states such that the result contains the results of each individual application of the $\syncfunction$ for states. Using virtual equivalence, we ensure that for each included state there exists exactly a single corresponding state in the other symbolic state.
\begin{definition}[$\syncfunction$ Function for Symbolic States]
\label{def:virtual-clocks:sync-function-for-symbolic-states}
Let $\VCG{A}$ be the VCG of TA $A$ regarding $B$, $\VCG{B}$ be the VCG of TA $B$ regarding $A$,
$\VCGFullState{A}{}$ be an AB-\semisynchronized \ symbolic state of $\VCG{A}$, 
$\VCGFullState{B}{}$ be an AB-\semisynchronized \ symbolic state of $\VCG{B}$ with $\VirtualEquivalence{A}{B}$. 
We define the $\syncfunction$ function
\begin{equation*}
\text{sync} : \VCGAllStates{A} \times \VCGAllStates{B} \rightarrow \VCGAllStates{A} \times \VCGAllStates{B}
\end{equation*}
such that 
\begin{equation*}
\syncfunction (\VCGState{A}, \VCGState{B}) = (\VCGFullState{A}{e, A}, \VCGFullState{B}{e, B})
\end{equation*}
with
\begin{align*}
\Zone{e, A} = \{ \ClockValuation{e, A} \in & (C_A \cup \{\chi_0, ..., \chi_{|C_A| + |C_B| -  1}\} \rightarrow \TimeDomain) \ | \ \\
\exists \TSVCState{A} \in & \VCGState{A} : \exists \TSVCState{B} \in \VCGState{B} : \exists \ClockValuation{e, B} \in (C_B \cup \{\chi_0, ... \chi_{|C_A| + |C_B| - 1}\} \rightarrow \TimeDomain) : \\
& (\TSVCFullState{A}{e, A}, \TSVCFullState{B}{e, B}) = \syncfunction (\TSVCState{A}, \TSVCState{B}) \}
\end{align*}
and
\begin{align*}
\Zone{e, B} = \{ \ClockValuation{e, B} \in & (C_B \cup \{\chi_0, ..., \chi_{|C_A| + |C_B| - 1}\} \rightarrow \TimeDomain) \ | \ \\
\exists \TSVCState{A} \in & \VCGState{A} : \exists \TSVCState{B} \in \VCGState{B} : \exists \ClockValuation{e, A} \in (C_A \cup \{\chi_0, ... \chi_{|C_A| + |C_B| - 1}\} \rightarrow \TimeDomain):\\
& (\TSVCFullState{A}{e, A}, \TSVCFullState{B}{e, B}) = \syncfunction (\TSVCState{A}, \TSVCState{B}) \}.
\end{align*}
\end{definition}

The $\text{sync}$ function for states is calculated by checking whether an original clock has been set to zero and to reset the corresponding virtual clock. An analogue procedure is possible to compute the $\syncfunction$ function for symbolic states. We remind the reader of the reset-operator, defined in Definition~\ref{def:background:zones}.
\begin{proposition}[Properties of $\syncfunction$ for Symbolic States]
\label{prop:virtual-clocks:VCG:correctness-sync-function}
Let $\VCG{A}$ be the VCG of TA $A$ regarding $B$, $\VCG{B}$ be the VCG of TA $B$ regarding $A$,
$\VCGFullState{A}{}$ be an AB-\semisynchronized \ symbolic state of $\VCG{A}$, 
$\VCGFullState{B}{}$ be an AB-\semisynchronized \ symbolic state of $\VCG{B}$ with $\VirtualEquivalence{A}{B}$. 
We define $R = R_A \cup R_B$ with
\begin{itemize}
	\item $R_A = \{\chi_i \ | \ i \in [0, |C_A| - 1] \land \forall \ClockValuation{A} \in \Zone{A} : \ClockValuation{A}(C_A[i]) = 0\}$, and
	\item $R_B = \{\chi_{i + |C_A|} \ | \ i \in [0, |C_B| - 1] \land \forall \ClockValuation{B} \in \Zone{B} : \ClockValuation{B}(C_B[i]) = 0\}$, respectively.
\end{itemize}
$\text{sync}(\VCGState{A}, \VCGState{B}) = ((l_A, R(\Zone{A})), (l_B, R(\Zone{B})))$ and $R(\Zone{A}) \VirtualEquivalenceBare R(\Zone{B})$ hold. 
The resulting symbolic states are AB-\synchronized.
\end{proposition}

\begin{proof}
To show the first statement, we denote $\text{sync}(\VCGState{A}, \VCGState{B}) = (\VCGFullState{A}{e, A}, \VCGFullState{B}{e, B})$ and show $R(\Zone{A}) \subseteq \Zone{e, A}$ and $\Zone{e, A} \subseteq R(\Zone{A})$ (as those statements for $B$ can be shown analogously). 

We consider any $\ClockValuation{R} \in R(\Zone{A})$. From the background section, we know that there exists a $\ClockValuation{A} \in \Zone{A}$ such that $R(\ClockValuation{A}) = \ClockValuation{R}$. By $\VirtualEquivalence{A}{B}$, we know that there exists exactly a single $\ClockValuation{B} \in \Zone{B}$ with $\text{virt}(\ClockValuation{A}) = \text{virt}(\ClockValuation{B})$ and denote $((l_A, \ClockValuation{e, A}), (l_B, \ClockValuation{e, B})) = \text{sync}((l_A, \ClockValuation{A}), (l_B, \ClockValuation{B}))$ with $\ClockValuation{e, A} \in \Zone{e, A}$ by Definition~\ref{def:virtual-clocks:sync-function-for-symbolic-states}. To show $\ClockValuation{R} \in \Zone{e, A}$, we show for each clock $c \in C_A \cup \{\chi_0, ..., \chi_{|C_A| + |C_B| - 1}\}$ that $\ClockValuation{R}(c) = \ClockValuation{e, A}(c)$. Since the original clock values are neither changed by the reset (as the reset set consists only of virtual clocks) nor by the sync function (see Definition~\ref{def:virtual_clocks:TSVC:sync-for-TSVC}), the original clock values are the same. For any $\chi_i$ with $i \in [0, |C_A| - 1]$, we distinguish the cases $\forall \ClockValuation{} \in \Zone{A} : \ClockValuation{}(C_A[i]) = 0$ and the opposite case. If the statement holds, $\chi_i$ is element of the reset set and $\chi_i$ is also set to zero by the $\syncfunction$. Therefore, $\ClockValuation{R}(\chi_i) = \ClockValuation{e, A}(\chi_i)$ in this case. Otherwise, $\chi_i$ is not element of $R$ and $\ClockValuation{A}(\chi_i) = \ClockValuation{A}(C_A[i])$ holds. Therefore, $\ClockValuation{e, A}(\chi_i) = \ClockValuation{A}(\chi_i) = \ClockValuation{R}(\chi_i)$ holds. For any $\chi_{i + |C_A|}$ with $i \in [0, |C_B| - 1]$, the equivalence can be shown analogously.

To show $\Zone{e, A} \subseteq R(\Zone{A})$, we consider any $\ClockValuation{e, A} \in \Zone{e, A}$ and show $\ClockValuation{e, A} \in R(\Zone{A})$. By Definition~\ref{def:virtual-clocks:sync-function-for-symbolic-states}, we know that there exists a $\ClockValuation{A} \in \Zone{A}$ and a $\ClockValuation{B} \in \Zone{B}$ with $\text{virt}(\ClockValuation{A}) = \text{virt}(\ClockValuation{B})$ and there exists a $\ClockValuation{e, B}$ such that $((l_A, \ClockValuation{e, A}), (l_B, \ClockValuation{e, B})) = \text{sync}((l_A, \ClockValuation{A}), (l_B, \ClockValuation{B}))$. Since $\ClockValuation{e, A} = R(\ClockValuation{A})$ can be shown analogously to the previous case, $\ClockValuation{e, A} \in R(\Zone{A})$ holds.

To show $R(\Zone{A}) \VirtualEquivalenceBare R(\Zone{B})$, we consider any $\ClockValuation{R, A} \in R(\Zone{A})$ and show that there exists a $\ClockValuation{R, B} \in R(\Zone{B})$ with $\text{virt}(\ClockValuation{R, A}) = \text{virt}(\ClockValuation{R, B})$, which implies $R(\Zone{A}) \leq_\text{virtual} R(\Zone{B})$. From the background section, we know that there exists a $\ClockValuation{A} \in \Zone{A}$ such that $R(\ClockValuation{A}) = \ClockValuation{R, A}$. By $\Zone{A} \VirtualEquivalenceBare \Zone{B}$, we know that there exists a $\ClockValuation{B} \in \Zone{B}$ such that $\text{virt}(\ClockValuation{A}) = \text{virt}(\ClockValuation{B})$. This implies $\text{virt}(R(\ClockValuation{A})) = \text{virt}(R(\ClockValuation{B}))$ (since all resets are exactly the same on both clock valuations) and since $R(\ClockValuation{B}) \in \Zone{B}$, $R(\Zone{A}) \leq_\text{virtual} R(\Zone{B})$ holds. Since the opposite direction can be shown analogously, $R(\Zone{A}) \VirtualEquivalenceBare R(\Zone{B})$ holds.

$R(\Zone{A})$ and $R(\Zone{B})$ are AB-\synchronized{} by the fact that any contained state is AB-\synchronized{} by Proposition~\ref{prop:virtual_clocks:TSVC:Properties-of-sync-for-states}.
\end{proof}

Analogously to AB-\synchronized \ states, also AB-\synchronized \ symbolic states are not affected 
if the $\syncfunction$ function is applied to them. The following example shows the impact of the $\syncfunction$ function, analogously to Example \ref{ex:virtual_clocks:TSVCs:impact-of-sync-function}.

\newcommand{\VCGExampleAutomataScalingFactor}{0.5}
\newcommand{\VCGExampleAutomataSpaceBetween}{10mm}
\newcommand{\VCGTikzFontSize}{\Large}
\newcommand{\VCGExampleArrowDesc}{\draw[-{Latex[length=3mm]}]}

\begin{example}

\begin{figure}
\centering
\scalebox{\VCGExampleAutomataScalingFactor}{
\begin{tikzpicture}
\tikzstyle{every node}=[font=\VCGTikzFontSize]
\tikzstyle{symstate} = [draw,rectangle,minimum width=5cm,inner sep=5pt,thick]
%
\node[symstate, align=center, initial left, initial text=] (02) {$l_0$\\$x_2=0$\\$\chi_0 = \chi_1 = \chi_2 = 0$};
\node[symstate, align=center, right = 1.5cm of 02] (12) {$l_0$\\$x_2<\infty$\\$\chi_0 = \chi_1 = \chi_2 = x_2$};
%
\node[symstate, align=center, right = 1.5cm of 12] (22) {$l_1$\\$x_2=0$\\$\chi_0 = \chi_1 = \chi_2 < \infty$};
\node[symstate, align=center, right = 5cm of 22] (32) {$l_1$\\$x_2=0$\\$\chi_0 = \chi_1 = 0; \chi_2 < \infty$};
%
\node[symstate, align=center, below = 0.7cm of 02, initial left, initial text=] (04) {$l_0$\\$x_4 = y_4 =0$\\$\chi_0 = \chi_1 = \chi_2 = 0$};
\node[symstate, align=center, below = 0.7cm of 12] (14) {$l_0$\\$x_4 = y_4<\infty$\\$\chi_0 = \chi_1 = \chi_2 = x_4$};
%
\node[symstate, align=center, below = 0.7cm of 22] (24) {$l_1$\\$x_4 = 0; y_4 < \infty$\\$\chi_0 = \chi_1 = \chi_2 = y_4$};
\node[symstate, align=center,  below = 0.7cm of 32] (34) {$l_1$\\$x_4 = 0; y_4 < \infty$\\$\chi_0 = \chi_1 = 0; \chi_2 = y_4$};
\node[rectangle, thick, align=center, inner sep=0pt, minimum size=12pt, below right = 0.5cm and 1cm of 22] (syncnode) {$\syncfunction$ function};
%
%
%
\VCGExampleArrowDesc (02) --node[above, align=center]{$\varepsilon$} (12);
\VCGExampleArrowDesc (12) --node[above, align=center]{a} (22);
\draw[dashed, -{Latex[length=2mm,width=3mm]}] (22) -- (syncnode.north west);
\draw[dashed, -{Latex[length=2mm,width=3mm]}] (syncnode.north east) -- (32);
\VCGExampleArrowDesc (04) --node[above, align=center]{$\varepsilon$} (14);
\VCGExampleArrowDesc (14) --node[above, align=center]{a} (24);
\draw[dashed, -{Latex[length=2mm,width=3mm]}] (24) -- (syncnode.south west);
\draw[dashed, -{Latex[length=2mm,width=3mm]}] (syncnode.south east) -- (34);
\end{tikzpicture}
}
\caption{Extracts of the VCG of $A_2$ regarding $A_4$ and $A_4$ regarding $A_2$}
\label{fig:virtual_clocks:VCG:extracts-of-the-vcgs-of-a2-reg-a4-and-a4-reg-a2}
\end{figure}
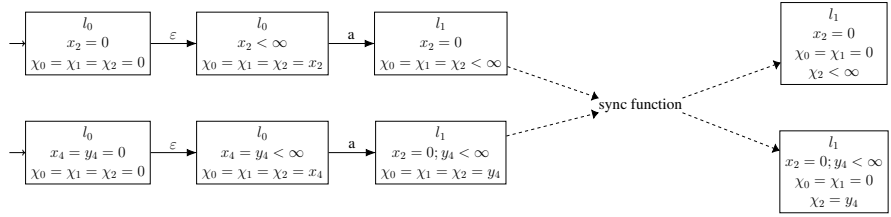

Figure \ref{fig:virtual_clocks:VCG:extracts-of-the-vcgs-of-a2-reg-a4-and-a4-reg-a2} shows (small) extracts from the VCG of $A_2$ regarding $A_4$ and the VCG of $A_4$ regarding $A_2$. We renamed the clocks analogously to Example \ref{ex:virtual_clocks:TSVCs:impact-of-sync-function}. Once again, we add three virtual clocks (since $A_2$ has a single original clock, while $A_4$ has two). Initially, all clocks are set to zero, which implies that both initial symbolic states are $A_2A_4$-\synchronized \ and $\VirtualEquivalence{0, A_2}{0, A_4}$ holds. Both properties also hold after the $\varepsilon$-transition. After the action transition, $x_2$ and $x_4$ are reset. While the virtual clock values of both symbolic states are still the same, the symbolic states are not $A_2A_4$-\synchronized \ anymore. We use the $\syncfunction$ function to reset the corresponding virtual clocks and the result is $A_2A_4$-\synchronized \ again.
\end{example}

We are now able to define a bisimulation equivalence for VCG.

\begin{definition}[Virtual Bisimulation for VCGs]
\label{def:virtual-clocks:VCG:CheckingForBoundedBisim:BoundedBiSim}
\newcommand{\VCGBiSimulateTransitionCommonPart}[3] {%
the existence of an outgoing transition $\VCGFullState{#1}{e, #1} \VCGTrans{\mu} \VCGFullState{\mu, #1}{\mu, #1}$ implies the existence of a finite set of symbolic states $\{\VCGFullState{\mu, #1}{\mu, 0, #1}, \VCGFullState{\mu, #1}{\mu, 1, #1}, ...\}$ with $(\bigcup \Zone{\mu, i, #1}) = \Zone{\mu, #1}$ and for any $\VCGFullState{\mu, #1}{\mu, i, #1}$ exists an outgoing transition $\VCGFullState{#2}{e, #2} \VCGTrans{\mu} \VCGFullState{\mu, #2}{\mu, #2}$ such that there exists a zone $\Zone{\mu, i, #2} \subseteq \Zone{\mu, #2}$ with $#3{\mu, #1}{\mu, i, #1}{\mu, #2}{\mu, i, #2}{n}$%
}%
\newcommand{\VCGBisimulateExtendToVCG}[3]{%
$\VCG{A}$ is #1 in order $n$ #2 $\VCG{B}$, denoted $#3{\VCG{A}}{\VCG{B}}{n}$, if and only if $#3{\VCGState{0, A}}{\VCGState{0, B}}{n}$ holds. $\VCG{A}$ is #1 #2 $\VCG{B}$, denoted $#3{\VCG{A}}{\VCG{B}}{}$, if and only if $\forall n \in \mathbb{N}^{\geq 0} : #3{\VCG{A}}{\VCG{B}}{n}$ holds.
}%
Assume two TA $A$, $B$ with VCG $\VCG{A}$ of $A$ regarding $B$ and VCG $\VCG{B}$ of $B$ regarding $A$, using the same alphabet $\Sigma$. Let $\VCGFullState{A}{}$ be an AB-\semisynchronized \ symbolic state of $\VCG{A}$ and let $\VCGFullState{B}{}$ be an AB-\semisynchronized \ symbolic state of $\VCG{B}$.

We define $\VCGFullState{A}{}$ to be virtually bisimilar in order $n=0$ to $\VCGFullState{B}{}$, denoted $\VCGFullState{A}{} \VCGBisimulationBare{0} \VCGFullState{B}{}$, if and only if $\VirtualEquivalence{A}{B}$ holds.

We define $\VCGFullState{A}{}$ to be virtually bisimilar in order $n+1$ (with $n \in \mathbb{N}^{\geq 0}$) to $\VCGFullState{B}{}$, denoted $\VCGFullState{A}{} \VCGBisimulationBare{n+1} \VCGFullState{B}{}$, if and only if $\VirtualEquivalence{A}{B}$ and for the AB-\synchronized \ symbolic states
\begin{equation*}
(\VCGFullState{A}{e, A}, \VCGFullState{B}{e, B}) = \text{sync}(\VCGFullState{A}{}, \VCGFullState{B}{})
\end{equation*}
\VCGBiSimulateTransitionCommonPart{A}{B}{\VCGBisimulationFullState} and \VCGBiSimulateTransitionCommonPart{B}{A}{\VCGBisimulationFullState}.

$\VCG{A}$ is virtual bisimilar in order $n$ to $\VCG{B}$, denoted $\VCGBisimulation{\VCG{A}}{\VCG{B}}{n}$, if and only if $\VCGBisimulation{\VCGState{0, A}}{\VCGState{0, B}}{n}$ holds. $\VCG{A}$ is virtual bisimilar to $\VCG{B}$, denoted $\VCGBisimulation{\VCG{A}}{\VCG{B}}{}$, if and only if $\forall n \in \mathbb{N}^{\geq 0} : \VCGBisimulation{\VCG{A}}{\VCG{B}}{n}$ holds.
\end{definition}

Virtual bisimulation holds for two VCG, if and only if the corresponding TA are timed bisimilar.

\begin{theorem}[Virtual Bisimulation of VCGs is correct]
\label{theorem:virtual-clocks:CheckingForBoundedBisim:BoundedSimIsCorrect}
Assume two TA $A$, $B$, the VCG $\VCG{A}$ of $A$ regarding $B$, and the VCG $\VCG{B}$ of $B$ regarding $A$.
$A \sim_n B$ holds if and only if $\VCG{A} \sim_n \VCG{B}$ holds. $A \sim B$ is equivalent to $\VCG{A} \sim \VCG{B}$.
\end{theorem}

\begin{proof}
By Theorem~\ref{theorem:virtual-clocks:TSVCSimIffTASim}, $A \sim_n B$ holds if and only if $\TSVCBisimulation{\TSVC{A}}{\TSVC{B}}{n}$. Hence, to show Proposition~\ref{theorem:virtual-clocks:CheckingForBoundedBisim:BoundedSimIsCorrect}, it suffices to show that $\TSVCBisimulation{\TSVC{A}}{\TSVC{B}}{n}$ holds if and only if $\VCG{A} \sim_n \VCG{B}$. To do so, we formulate two hypotheses that relate virtual bisimulation of states to virtual bisimulation of symbolic states.
We show both hypotheses by induction in $n$.

\noindent\textbf{Hypothesis 1:} If $\VCGState{A}$ is non-empty and AB-\semisynchronized \ and $\VCGState{B}$ is non-empty and AB-\semisynchronized, then $\VCGBisimulationState{A}{B}{n}$ implies 
\begin{align*}
& \forall \TSVCState{A} \in \VCGState{A} : \exists \TSVCState{B} \in \VCGState{B} : \TSVCBisimulation{\TSVCState{A}}{\TSVCState{B}}{n} \land \\
& \forall \TSVCState{B} \in \VCGState{B} : \exists \TSVCState{A} \in \VCGState{A} : \TSVCBisimulation{\TSVCState{A}}{\TSVCState{B}}{n}.
\end{align*}

\noindent\textbf{Base Case:}  Assume any $\TSVCFullState{A}{} \in \VCGState{A}$. By Definition \ref{def:virtual-clocks:VCG:CheckingForBoundedBisim:BoundedBiSim}, $\VCGBisimulationState{A}{B}{0}$ implies $\VirtualEquivalence{A}{B}$. By Definition \ref{def:virtual-clocks:VCG:VirtualComparisonOfZones}, this implies the existence of a $\TSVCFullState{B}{} \in \VCGState{B}$, such that $\text{virt}(\ClockValuation{A}) = \text{virt}(\ClockValuation{B})$, which implies $\TSVCBisimulation{\TSVCFullState{A}{}{}}{\TSVCFullState{B}{}{}}{0}$ by Definition \ref{def:virtual-clocks:BoundedTimedBiSimulationForTSVCs}. Since the second part of the statement can be shown analogously, the base case holds.

\noindent\textbf{Induction Step:} Assume any $\TSVCFullState{A}{} \in \VCGState{A}$. $\VCGBisimulationState{A}{B}{n}$ implies $\VirtualEquivalence{A}{B}$. Therefore, there exists exactly a single $\TSVCFullState{B}{}{} \in \VCGState{B}$ such that $\text{virt}(\ClockValuation{A}) = \text{virt}(\ClockValuation{B})$ and we show $\TSVCBisimulation{\TSVCFullState{A}{}}{\TSVCFullState{B}{}}{n}$. To do so, we show that for the states $(\TSVCFullState{A}{e, A}, \TSVCFullState{B}{e, B}) = \text{sync}(\TSVCFullState{A}{}, \TSVCFullState{B}{})$ the existence of a transition $\TSVCFullState{A}{e, A} \TSVCTrans{\mu} \TSVCFullState{\mu, A}{}$ ($\mu \in \Sigma \cup \TimeDomain$) implies the existence of a transition $\TSVCFullState{B}{e, B} \TSVCTrans{\mu} \TSVCFullState{\mu, B}{}$ with $\TSVCBisimulation{\TSVCFullState{\mu, A}{}}{\TSVCFullState{\mu, B}{}}{n-1}$. The second condition can be shown analogously. 

We denote $(\VCGFullState{A}{e, A}, \VCGFullState{B}{e, B}) = \text{sync}(\VCGFullState{A}{}, \VCGFullState{B}{})$ and by Definition \ref{def:virtual-clocks:sync-function-for-symbolic-states}, we know that $\TSVCFullState{A}{e, A} \in \VCGFullState{A}{e, A}$ and $\TSVCFullState{B}{e, B} \in \VCGFullState{B}{e, B}$. Due to forward stability (Proposition \ref{prop:background:forward-stability}), the existence of $\TSVCFullState{A}{e, A} \TSVCTrans{\mu} \TSVCFullState{\mu, A}{}$ implies the existence of a transition $\VCGFullState{A}{e, A} \VCGTrans{\mu_\varepsilon} \VCGFullState{\mu, A}{}$ with $\ClockValuation{\mu, A} \in \Zone{\mu, A}$, $\mu_\varepsilon = \mu$ if $\mu \in \Sigma$ and $\mu_\varepsilon = \varepsilon$, else. Due to Definition \ref{def:virtual-clocks:VCG:CheckingForBoundedBisim:BoundedBiSim} and $\VCGBisimulationState{A}{B}{n}$, this implies the existence of a zone $\Zone{\mu, i, A} \subseteq \Zone{\mu, A}$ with $\TSVCFullState{\mu, A}{} \in \VCGFullState{\mu, A}{\mu, i, A}$ such that there exists an outgoing transition $\VCGState{B} \VCGTrans{\mu_\varepsilon} \VCGFullState{\mu, B}{}$ such that there exists a zone $\Zone{\mu, i, B} \subseteq \Zone{\mu, B}$ with $\VCGBisimulation{\VCGFullState{\mu, A}{\mu, i, A}}{\VCGFullState{\mu, B}{\mu, i, B}}{n-1}$.

Since $\VCGFullState{\mu, A}{\mu, i, A}$ and $\VCGFullState{\mu, B}{\mu, i, B}$ are both non-empty and AB-\semisynchronized \ (due to Proposition \ref{prop:virtual-clocks:VCG:outgoing-transitions-of-symbolic-syncd-states} and since substates of AB-\semisynchronized \ states are AB-\semisynchronized, too), we can apply the induction hypothesis here. Therefore,
\begin{equation*}
\forall \TSVCFullState{\mu, A}{} \in \VCGFullState{\mu, A}{} : \exists \TSVCFullState{\mu, B}{} \in \VCGFullState{\mu, B}{} : \TSVCBisimulation{\TSVCFullState{\mu, A}{}}{\TSVCFullState{\mu, B}{}}{n-1}
\end{equation*}
holds. Therefore, there exists a $\TSVCFullState{\mu, B}{} \in \VCGFullState{\mu, B}{}$ such that $\TSVCBisimulation{\TSVCFullState{\mu, A}{}}{\TSVCFullState{\mu, B}{}}{n-1}$ and we only have to show that there exists a transition $\TSVCFullState{B}{e, B} \TSVCTrans{\mu} \TSVCFullState{\mu, B}{}$.

Due to backward stability (Proposition \ref{prop:background:backward-stability}), the existence of a transition $\VCGFullState{B}{e, B} \VCGTrans{\mu_\varepsilon} \VCGFullState{\mu, B}{}$ with $\TSVCFullState{\mu, B}{} \in \VCGFullState{\mu, B}{}$ implies the existence of a transition $\TSVCFullState{B}{\text{help}, B} \TSVCTrans{\mu} \TSVCFullState{\mu, B}{}$ with $\TSVCFullState{B}{\text{help}, B} \in \VCGFullState{B}{e, B}$ and we have to show $\ClockValuation{\text{help}, B} = \ClockValuation{e, B}$.  

If $\mu \in \TimeDomain$, we know that $\ClockValuation{\text{help}, B} + \mu = \ClockValuation{\mu, B}$. Due to $\text{virt}(\ClockValuation{e, A}) = \text{virt}(\ClockValuation{e, B})$, $\ClockValuation{e, A} + \mu = \ClockValuation{\mu, A}$, and $\text{virt}(\ClockValuation{\mu, A}) = \text{virt}(\ClockValuation{\mu, B})$,  $\text{virt}(\ClockValuation{\text{help}, B}) = \text{virt}(\ClockValuation{e, B})$.

If $\mu \in \Sigma$, we know that $\text{virt}(\ClockValuation{\text{help}, B}(\chi_i)) = \text{virt}(\ClockValuation{\mu, B}(\chi_i))$. Due to $\text{virt}(\ClockValuation{e, A}) = \text{virt}(\ClockValuation{e, B})$, $\text{virt}(\ClockValuation{e, A}) = \text{virt}(\ClockValuation{\mu, A})$, and $\text{virt}(\ClockValuation{\mu, A}) = \text{virt}(\ClockValuation{\mu, B})$,  $\text{virt}(\ClockValuation{\text{help}, B}) = \text{virt}(\ClockValuation{e, B})$.

Since $\VCGFullState{B}{e, B}$ is AB-\synchronized, and $\text{virt}(\ClockValuation{\text{help}, B}) = \text{virt}(\ClockValuation{e, B})$ hold, $\ClockValuation{\text{help}, B} = \ClockValuation{e, B}$ is implied by the fact that a clock valuation within an AB-\semisynchronized{} symbolic state is uniquely identified by its virtual clocks. Since the second part can be shown analogously, the induction step holds.

Now we show the opposite direction.

\noindent\textbf{Hypothesis 2:} If $\VCGState{A}$ is non-empty and AB-\semisynchronized \ and $\VCGState{B}$ is non-empty and AB-\semisynchronized \ with $\VirtualEquivalence{A}{B}$, then
\begin{align*}
& \forall \TSVCState{A} \in \VCGState{A} : \exists \TSVCState{B} \in \VCGState{B} : \TSVCBisimulation{\TSVCState{A}}{\TSVCState{B}}{n} \land \\
& \forall \TSVCState{B} \in \VCGState{B} : \exists \TSVCState{A} \in \VCGState{A} : \TSVCBisimulation{\TSVCState{A}}{\TSVCState{B}}{n}
\end{align*}
implies 
\begin{equation*}
\VCGBisimulationState{A}{B}{n}.
\end{equation*}

\noindent\textbf{Base Case:} The base case is implied by the precondition $\VirtualEquivalence{A}{B}$.

\noindent\textbf{Induction Step:} Since $\VirtualEquivalence{A}{B}$ holds by the precondition, the first condition is fulfilled and we have to show that for the AB-\synchronized \ symbolic states $(\VCGFullState{A}{e, A}, \VCGFullState{B}{e, B}) = \text{sync}(\VCGFullState{A}{}, \VCGFullState{B}{})$, the existence of an outgoing transition $\VCGFullState{A}{e, A} \VCGTrans{\mu} \VCGFullState{\mu, A}{}$ implies the existence of a finite set of symbolic states $\{\VCGFullState{\mu, A}{\mu, 0, A}, \VCGFullState{\mu, A}{\mu, 1, A}, ...\}$ with $(\bigcup \Zone{\mu, i, A}) = \Zone{\mu, A}$ and for any $\VCGFullState{\mu, A}{\mu, i, A}$ exists an outgoing transition $\VCGFullState{B}{e, B} \VCGTrans{\mu} \VCGFullState{\mu, B}{\mu, B}$ such that there exists a zone $\Zone{\mu, i, B} \subseteq \Zone{\mu, B}$ with $\VCGBisimulationFullState{\mu, A}{\mu, i, A}{\mu, B}{\mu, i, B}{n-1}$. Since the second condition can be shown analogously, we skip this part.

Due to backward stability (Proposition \ref{prop:background:backward-stability}), for each $\TSVCFullState{\mu, A}{} \in \VCGFullState{\mu, A}{}$ exists a $\TSVCFullState{A}{e, A} \in \VCGFullState{A}{e, A}$, such that there exists a transition $\TSVCFullState{A}{e, A} \TSVCTrans{\mu_d} \TSVCFullState{\mu, A}{}$ (with $\mu_d \in \TimeDomain$ if $\mu = \varepsilon$ and $\mu_d = \mu$ else). Due to Definition \ref{def:virtual-clocks:sync-function-for-symbolic-states}, we know that there exists a $\TSVCFullState{B}{e, B} \in \VCGFullState{B}{e, B}$ such that there exists a $\TSVCFullState{A}{} \in \VCGFullState{A}{}$ and a $\TSVCFullState{B}{} \in \VCGFullState{B}{}$ such that $\text{virt}(\ClockValuation{A}) = \text{virt}(\ClockValuation{B})$  and $(\TSVCFullState{A}{e, A}, \TSVCFullState{B}{e, B}) = \text{sync}(\TSVCFullState{A}{}, \TSVCFullState{B}{})$ holds. Since $\TSVCFullState{B}{}$ is the only state with $\text{virt}(\ClockValuation{A}) = \text{virt}(\ClockValuation{B})$, we know that $\TSVCBisimulation{\TSVCFullState{A}{}{}}{\TSVCFullState{B}{}{}}{n}$ by the precondition. By Definition \ref{def:virtual-clocks:BoundedTimedBiSimulationForTSVCs}, the existence of the transition $\TSVCFullState{A}{e, A} \TSVCTrans{\mu_d} \TSVCFullState{\mu, A}{}$ implies the existence of a transition $\TSVCFullState{B}{e, B} \TSVCTrans{\mu_d} \TSVCFullState{\mu, B}{}$ with $\TSVCBisimulation{\TSVCFullState{\mu, A}{}{}}{\TSVCFullState{\mu, B}{}{}}{n-1}$. 

If $\mu = \varepsilon$, we define the finite set of symbolic states $\{\VCGFullState{\mu, A}{\mu, 0, A}, \VCGFullState{\mu, A}{\mu, 1, A}, ...\}$ to be equal to $\{\VCGFullState{\mu, A}{}\}$.  For the $\varepsilon$-transition $\VCGFullState{B}{e, B} \VCGTrans{\varepsilon} \VCGFullState{\mu, B}{}$, we know by forward stability $\TSVCFullState{\mu, B}{} \in \VCGFullState{\mu, B}{}$ and, therefore, $\VirtualComparison{\mu, A}{\mu, B}$ holds. Since the same can be shown for $B$ analogously, we can follow $\VirtualEquivalence{\mu, A}{\mu, B}$. Both symbolic states are non-empty and AB-\synchronized \ by Proposition \ref{prop:virtual-clocks:VCG:outgoing-transitions-of-symbolic-syncd-states}, and for any $\TSVCFullState{\mu, A}{}$ exists a $\TSVCFullState{\mu, B}{}$ with $\TSVCBisimulation{\TSVCFullState{\mu, A}{}{}}{\TSVCFullState{\mu, B}{}{}}{n-1}$ and vice versa. Therefore, we can apply the induction hypothesis, which implies $\VCGBisimulation{\VCGFullState{\mu, A}{}}{\VCGFullState{\mu, B}{\mu, B}}{n-1}$ and the induction step holds in this case.

If $\mu \in \Sigma$, we consider the set of regions $\{r_0, ..., r_m\}$ that have a non-empty intersection with $\Zone{\mu, A}$, denote $\Zone{\mu, i, A} = \Zone{\mu, A} \land r_i$, and split $\Zone{\mu, A}$ into the set of symbolic states $\{\VCGFullState{\mu, A}{\mu, 0, A}, ..., \VCGFullState{\mu, A}{\mu, m, A}\}$. Since $(\bigcup \Zone{\mu, i, A}) = \Zone{\mu, A}$ holds, we only have to show that for any $\Zone{\mu, i, A}$ exists a transition $\VCGFullState{B}{e, B} \VCGTrans{\mu} \VCGFullState{\mu, B}{}$ such that there exists a $\Zone{\mu, i, B} \subseteq \Zone{\mu, B}$ with $\VCGBisimulation{\VCGFullState{\mu, A}{\mu, i, A}}{\VCGFullState{\mu, B}{\mu, i, B}}{n-1}$.

We know that for any $\ClockValuation{\mu, A} \in \Zone{\mu, i, A}$ exists a $\ClockValuation{e, A} \in \Zone{e, A}$ with $\TSVCFullState{A}{e, A} \TSVCTrans{\mu} \TSVCFullState{\mu, A}{\mu, A}$ and $\text{virt}(\ClockValuation{e, A}) = \text{virt}(\ClockValuation{\mu, A})$. Moreover, there is only a single $\ClockValuation{e, A}$ fulfilling this property. Due to the precondition, we know that there exists a $\ClockValuation{e, B} \in \Zone{e, B}$ such that $\TSVCBisimulation{\TSVCFullState{A}{e, A}}{\TSVCFullState{B}{e, B}}{n}$. Therefore, there exists a transition $\TSVCFullState{B}{e, B} \TSVCTrans{\mu} \TSVCFullState{\mu, B}{}$ with $\TSVCBisimulation{\TSVCFullState{\mu, A}{}}{\TSVCFullState{\mu, B}{}}{n-1}$. 

By the hypotheses, used to prove Theorem~\ref{theorem:virtual-clocks:TSVCSimIffTASim}, we know that this implies that the corresponding states of the TLTS are timed bisimilar. The induction step is now implied by the results of \u{C}er\={a}ns~\cite{Cerans1992}, which shows that the elements of the regions of those states of the TLTS are, therefore, also pairwise bisimilar. Once again by the hypotheses, used to prove Theorem~\ref{theorem:virtual-clocks:TSVCSimIffTASim}, and the induction hypothesis, we can therefore follow $\VCGBisimulation{\VCGFullState{\mu, A}{\mu, i, A}}{\VCGFullState{\mu, B}{\mu, i, B}}{n-1}$ and the induction step holds.

Since the initial states of the TSVC are the only elements of the initial states of the VCGs, hypothesis 1 and hypothesis 2 imply Theorem~\ref{theorem:virtual-clocks:CheckingForBoundedBisim:BoundedSimIsCorrect}.
\end{proof}

We remind the reader of the two main challenges, described in Example \ref{ex:background:zone-graphs-do-not-contain-sufficient-information}. Using virtual equivalence,  we can compare the timing behavior of two TA even if the set of original clocks differ. The following example shows how the second main challenge, non-observability of clock resets, is solved by using virtual clocks.

\begin{example}
\label{ex:virtual_clocks:VCG:observability-of-clk-resets}
\begin{figure}
\centering
\scalebox{\VCGExampleAutomataScalingFactor}{
\begin{tikzpicture}
\tikzstyle{every node}=[font=\VCGTikzFontSize]
\tikzstyle{symstate} = [draw,rectangle,minimum width=3.5cm,inner sep=5pt,thick]
%
%
\node[symstate, align=center, initial left, initial text=] (02) {$l_0$\\$x_2=0$\\$\chi_0 = \chi_1 = 0$};
\node[symstate, align=center, right = 1.5cm of 02] (12) {$l_0$\\$x_2<\infty$\\$\chi_0 = \chi_1 = x_2$};
%
\node[symstate, align=center, right = 1.5cm of 12] (22) {$l_1$\\$x_2=0$\\$\chi_0 = \chi_1 < \infty$};
%
\node[symstate, align=center, left = 1.5cm of 02] (21) {$l_1$\\$x_1 = 0$\\$\chi_0 = \chi_1 = x_1$};
\node[symstate, align=center, left = 1.5cm of 21] (11) {$l_0$\\$x_1<\infty$\\$\chi_0 = \chi_1 = x_1$};
\node[symstate, align=center, left = 1.5cm of 11, initial left, initial text=] (01) {$l_0$\\$x_1 =0$\\$\chi_0 = \chi_1 = 0$};
\VCGExampleArrowDesc (01) --node[above, align=center]{$\varepsilon$} (11);
\VCGExampleArrowDesc (11) --node[above, align=center]{a} (21);
\VCGExampleArrowDesc (02) --node[above, align=center]{$\varepsilon$} (12);
\VCGExampleArrowDesc (12) --node[above, align=center]{a} (22);
\VCGExampleArrowDesc (11) to[loop above] node[below, align=center]{$\varepsilon$} (11);
\VCGExampleArrowDesc (12) to[loop below] node[above, align=center]{$\varepsilon$} (12);
\end{tikzpicture}
}
\caption{Extracts of the VCG of $A_1$ regarding $A_2$ (left) and the VCG of $A_2$ regarding $A_1$ (right)}
\label{fig:virtual_clocks:VCG:extracts-of-the-vcgs-of-a1-reg-a2-and-a2-reg-a1}
\end{figure}
In Example \ref{ex:background:zone-graphs-do-not-contain-sufficient-information}, we have seen that the zone graphs of $A_1$ and $A_2$ are the same despite the fact that $A_1$ and $A_2$ are not timed bisimilar. In Figure \ref{fig:virtual_clocks:VCG:extracts-of-the-vcgs-of-a1-reg-a2-and-a2-reg-a1} (small) extracts of the VCG $A_{1, \text{VCG}}$ of $A_1$ regarding $A_2$ and the VCG $A_{2, \text{VCG}}$ of $A_2$ regarding $A_1$ are shown. Obviously, the VCGs differ. In $A_{1, \text{VCG}}$, the original clock is not reset during the action labeled transition. Therefore, the virtual clock values are the same as the value of $x_1$, which is equal to zero. In contrast, the virtual clock values in $A_{2, \text{VCG}}$ can be greater than zero and the VCGs of $A_1$ and $A_2$ are not the same.
\end{example}

Using Theorem \ref{theorem:virtual-clocks:CheckingForBoundedBisim:BoundedSimIsCorrect}, we can also check for timed bisimilarity of non-deterministic TA.

\begin{example}
\label{ex:virtual-clocks:VCG:non-determ}
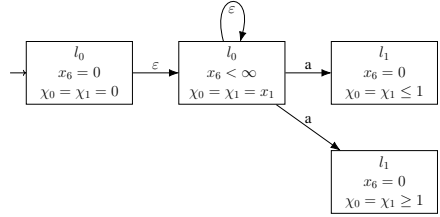
\begin{figure}
\centering
\scalebox{\VCGExampleAutomataScalingFactor}{
\begin{tikzpicture}
\tikzstyle{every node}=[font=\VCGTikzFontSize]
\tikzstyle{symstate} = [draw,rectangle,minimum width=3.5cm,inner sep=5pt,thick]
%
\node[symstate, align=center, initial left, initial text=] (06) {$l_0$\\$x_6 = 0$\\$\chi_0 = \chi_1 = 0$};
\node[symstate, align=center, right = 1.5cm of 06] (16) {$l_0$\\$x_6<\infty$\\$\chi_0 = \chi_1 = x_1$};
%
\node[symstate, align=center, right = 1.5cm of 16] (26) {$l_1$\\$x_6 = 0$\\$\chi_0 = \chi_1 \leq 1$};
\node[symstate, align=center, right = 1.5cm of 26] (36) {$l_1'$\\$x_6 = 0$\\$\chi_0 = \chi_1 \geq 1$};
\VCGExampleArrowDesc (06) --node[above, align=center]{$\varepsilon$} (16);
\VCGExampleArrowDesc (16) --node[above, align=center]{a} (26);
\VCGExampleArrowDesc (16) to[out=45, in=135] node[above, align=center]{a} (36);
\VCGExampleArrowDesc (16) to[loop above] node[below, align=center]{$\varepsilon$} (16);
\end{tikzpicture}
}
\caption{Extract of the VCG of $A_6$ regarding $A_2$}
\label{fig:virtual_clocks:VCG:extracts-of-the-vcgs-of-a2-reg-a6-and-a6-reg-a2}
\end{figure}
From Example \ref{ex:background:timed-bisim}, we know that $A_2$ and $A_6$ are timed bisimilar. Figure \ref{fig:virtual_clocks:VCG:extracts-of-the-vcgs-of-a2-reg-a6-and-a6-reg-a2} shows an extract of the VCG of $A_6$ regarding $A_2$. The VCG of $A_2$ regarding $A_6$ is equivalent to the VCG of $A_2$ regarding $A_1$ and, therefore, we do not redraw the extract but refer to Figure \ref{fig:virtual_clocks:VCG:extracts-of-the-vcgs-of-a1-reg-a2-and-a2-reg-a1}. The crucial part is the comparison of the symbolic states reached after the $\varepsilon$-transition. Since the symbolic states are virtually equivalent and $A_2A_6$-\synchronized, we know that applying the $\syncfunction$ function does not change the symbolic states. The outgoing $\varepsilon$-transitions of those symbolic states are self-loops and, obviously, we cannot find a contradiction using the $\varepsilon$-transitions if we cannot find a contradiction without using them (actually, we will use this fact in the next chapter). The symbolic state of $\VCG{A_{2,} {} }$ has a single outgoing transition $\VCGFullState{0}{\varepsilon, 2} \VCGTrans{a} \VCGFullState{1}{2}$, while the symbolic state of $\VCG{A_{6,} {} }$ has two outgoing transitions $\VCGFullState{0}{\varepsilon, 6} \VCGTrans{a} \VCGFullState{1}{6, \leq 1}$ and $\VCGFullState{0}{\varepsilon, 6} \VCGTrans{a} \VCGFullState{1}{6, \geq 1}$. We focus on the outgoing transition of $\VCG{A_{2, } {}}$.

According to Definition \ref{def:virtual-clocks:VCG:CheckingForBoundedBisim:BoundedBiSim}, we have to find a finite set of symbolic states $\{\VCGFullState{1}{0, 2}, \VCGFullState{1}{1, 2}, ...\}$ with $(\bigcup \Zone{i, 2}) = \Zone{2}$ and for any $\VCGFullState{1}{i, 2}$ exists an outgoing transition $\VCGFullState{0}{\varepsilon, 6} \VCGTrans{a} \VCGFullState{}{6}$ such that there exists a zone $\Zone{i, 6} \subseteq \Zone{6}$ with $\VCGBisimulation{\VCGFullState{1}{i, 2}}{\VCGFullState{}{i, 6}}{}$.

To do so, we split the zone $\Zone{2}$ into the subzones $\Zone{1, 2} = \Zone{2} \land \chi_0 \leq 1$ and $\Zone{2, 2} = \Zone{2} \land \chi_0 > 1$. We define $\Zone{1, 6} = \Zone{6, \leq 1}$ and $\Zone{2, 6} = (\Zone{6, \geq 1} \land \chi_0 > 1) \subseteq \Zone{6, \geq 1}$, such that the needed properties are satisfied (under the assumptions $\VCGBisimulation{\VCGFullState{1}{i, 2}}{\VCGFullState{1}{i, 6}}{}$, which actually hold).
\end{example}

We remind the reader that zone graphs can be infinite. This can also happen to VCGs. Given two TA $A$ and $B$ using the sets of clocks $C_A$ and $C_B$, we use the normalization function $k : C_A \cup C_B \cup \{\chi_0, ..., \chi_{|C_A| + |C_B| - 1}\} \rightarrow \mathbb{N}^{\geq 0}$, such that for any clock constraint $c \sim m$, which occurs in $A$ or $B$, $m < k(c)$, for any $i \in [0, |C_A| - 1] : k(C_A[i]) = k(\chi_i)$, and for any $i \in [0, |C_B| - 1] : k(C_B[i]) = k(\chi_{i + |C_A|})$ hold. Using this normalization function, it is implied by the definition of k-normalization as shown in \cite{Bengtsson2004} and the results of \u{C}er\={a}ns~\cite{Cerans1992} that two symbolic states are virtually bisimilar if and only if the k-normalized versions are virtually bisimilar. For more information see \cite{FullReport}. Therefore, we can use k-normalization during the check for virtual bisimulation and the VCGs become finite. In the next section, we describe an algorithm to effectively check virtual bisimulation on two input models.

%% file: sections/algorithm.tex
\section{Checking for Timed Bisimulation}
\label{sec:CheckingForBisimulation}

\newcommand{\AlgExampleAutomataScalingFactor}{0.5}
\newcommand{\AlgExampleAutomataSpaceBetween}{10mm}
\newcommand{\AlgTikzFontSize}{\Large}
\newcommand{\AlgExampleArrowDesc}{\draw[-{Latex[length=3mm]}]}

During the implementation of the construction described in Definition \ref{def:virtual-clocks:VCG:CheckingForBoundedBisim:BoundedBiSim}, two open questions remain to be solved. The first of these can be observed in Example \ref{ex:virtual-clocks:VCG:non-determ}. While for small TA, it is straightforward to see how to split the target zone, we have to solve this problem in general. The second problem we solve is to provide evidence that alternating sequences are indeed permissible.

\subsection{Virtual Constraints}

To solve the first open question, we use so-called \emph{contradictions}. Each contradiction describes a symbolic substate of either $\VCGFullState{A}{}$ or $\VCGFullState{B}{}$, which is not bisimilar to any substate of the other symbolic state. If the set of contradictions is empty, the symbolic states are virtually bisimilar. \emph{Virtual constraints} describe the contradictions to virtual bisimulation of the given pair of symbolic states. The following example shows the desired input/output behavior of our algorithm.
\begin{example}
\label{ex:virtual-clocks:virtual-constraint:initial-example}
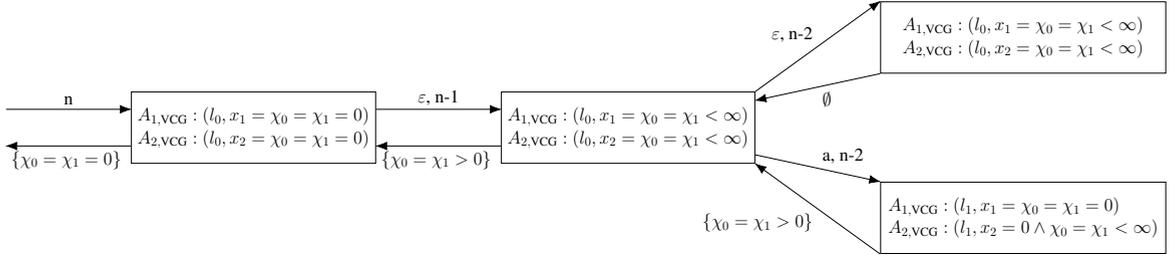
\begin{figure}
\centering
\scalebox{0.47}{
\begin{tikzpicture}
\tikzstyle{every node}=[font=\AlgTikzFontSize]
\tikzstyle{foo} = [draw,rectangle,minimum width=5cm, minimum height = 2cm, inner sep=5pt,thick, align=left]
\node[foo](calln) {$A_{1, \text{VCG}} : (l_0, x_1 = \chi_0 = \chi_1 = 0)$ \\ $A_{2, \text{VCG}} : (l_0, x_2 = \chi_0 = \chi_1 = 0)$};
\node[foo, right = 3.5cm of calln](callnminusone) {$A_{1, \text{VCG}} : (l_0, x_1 = \chi_0 = \chi_1 < \infty)$ \\$A_{2, \text{VCG}} : (l_0, x_2 = \chi_0 = \chi_1 < \infty)$};
\node[foo, above right=-0.2cm and 3.5cm of callnminusone, minimum width=8cm](epsilontransition) {$A_{1, \text{VCG}} : (l_0, x_1 = \chi_0 = \chi_1 < \infty)$ \\ $A_{2, \text{VCG}} : (l_0, x_2 = \chi_0 = \chi_1 < \infty)$};
\node[foo, below right=-0.2cm and 3.5cm of callnminusone, minimum width=8cm](atransition) {$A_{1, \text{VCG}} : (l_1, x_1 = \chi_0 = \chi_1 = 0)$ \\ $A_{2, \text{VCG}} : (l_1, x_2 = 0 \land \chi_0 = \chi_1 < \infty)$};
\AlgExampleArrowDesc ($(calln.north west) + (-3.5cm, -0.5cm)$) --node[above, align=center]{n} ($(calln.north west) + (0cm, -0.5cm)$);
\AlgExampleArrowDesc ($(calln.north east) + (0cm, -0.5cm)$) --node[above, align=center]{$\varepsilon$, n-1} ($(callnminusone.north west) + (0cm, -0.5cm)$);
\AlgExampleArrowDesc (callnminusone.north east) --node[above left, align=center]{$\varepsilon$, n-2} (epsilontransition.north west);
\AlgExampleArrowDesc (callnminusone) --node[above right, align=center]{a, n-2} (atransition.north west);
\AlgExampleArrowDesc (epsilontransition.south west) --node[below right, align=center]{ $ \emptyset$} (callnminusone);
\AlgExampleArrowDesc (atransition.south west) --node[below left, align=center]{$\{ \chi_0 = \chi_1 > 0 \}$} (callnminusone.south east);
\AlgExampleArrowDesc ($(callnminusone.south west) + (0cm, 0.5cm)$) --node[below, align=center] { $\{ \chi_0 = \chi_1 > 0 \}$ } ($(calln.south east) + (0cm, 0.5cm)$);
\AlgExampleArrowDesc ($(calln.south west) + (0cm, 0.5cm)$) --node[below, align=center] { $\{ \chi_0 = \chi_1 = 0 \}$ } ($(calln.south west) + (-3.5cm, 0.5cm)$);
\end{tikzpicture}
}
\caption{Illustration of Example \ref{ex:virtual-clocks:virtual-constraint:initial-example}}
\label{fig:algorithm:initial-example:Illustration}
\end{figure}

We reuse the extracts of the VCGs from Example~\ref{ex:virtual_clocks:VCG:observability-of-clk-resets} and apply our algorithm to the initial symbolic states to determine whether they are virtually bisimilar in order $n > 1$. This is illustrated on the left-hand side of Figure~\ref{fig:algorithm:initial-example:Illustration}, where the value $n$ is provided to the initial symbolic states. We have to check whether the zones of the initial symbolic states are virtually equivalent (which is indeed the case) and we have to check all outgoing transitions. As we are only considering transitions depicted in Figure~\ref{fig:virtual_clocks:VCG:extracts-of-the-vcgs-of-a1-reg-a2-and-a2-reg-a1}, we focus on the targets of the $\varepsilon$-transitions and check them for virtual bisimulation in order $n-1$, as illustrated by the rectangle in the centre of Figure~\ref{fig:algorithm:initial-example:Illustration}. Since the zones of the symbolic states are virtually equivalent, we have to check the outgoing transitions.

The outgoing $\varepsilon$-transitions (illustrated by the upper rectangle of Figure~\ref{fig:algorithm:initial-example:Illustration}) of the symbolic states are both self-loops. Consequently, to check whether virtual bisimulation in order $n-1$ holds for this pair of symbolic states, we have to check whether virtual bisimulation in order $n-2$ holds for the same pair of symbolic states. Should a contradiction arise within the subsequent $n-2$ transitions, it will inevitably be found when we check the symbolic states for virtual bisimulation in order $n-1$. Therefore, the application of the algorithm may be terminated at this point.

The transitions labeled with $a$ (illustrated by the lower rectangle of Figure~\ref{fig:algorithm:initial-example:Illustration}) lead into a pair of symbolic states that are not virtually bisimilar in order $n-2$, as they are not virtually equivalent. The virtual constraint $(\chi_0 = \chi_1 > 0)$ represents the non-overlapping part of the zones. Consequently, we expect the algorithm to return a set that contains exactly this constraint.

Since our checks for virtual bisimulation in order $n-2$ returned a contradiction, it is necessary to analyze which parts of the symbolic states in the centre of Figure~\ref{fig:algorithm:initial-example:Illustration} have a path leading to that contradiction. In $A_{2, \text{VCG}}$, any symbolic substate of $(l_0, (\chi_0 = \chi_1 = x_2) \land (\chi_0 = \chi_1 > 0))$ has an outgoing action transition to the contradiction we have already found. Consequently, $\{(\chi_0 = \chi_1 > 0)\}$ should be returned when our algorithm is applied to the symbolic states depicted in the centre of Figure~\ref{fig:algorithm:initial-example:Illustration}. Since the initial symbolic states have an outgoing transition to this contradiction, the set $\{(\chi_0 = \chi_1 = 0)\}$ is returned when the VCGs are checked for virtual bisimulation.
\end{example}

If the algorithm has found a contradiction, it stops. It should be noted that the algorithm returns a contradiction if one exists, but it does not guarantee the identification of all contradictions. The termination criterion described in Example~\ref{ex:virtual-clocks:virtual-constraint:initial-example} is analogous to the well-known termination criterion when checking for untimed bisimulation \cite{SangiorgiIntroToBisim}. When we check for virtual bisimulation in order $n$, we assume that virtual bisimulation in order $n-1$ holds for this pair of symbolic states. If there is no contradiction, this assumption is validated by Proposition \ref{prop:virtual-clocks:bisimInOrderPlusImpliesBisimInOrder}. Conversely, the discovery of a contradiction indicates that the assumption was an overfit and, thus, the identified contradiction remains valid.
The following definition introduces virtual constraints and an operator to extract the virtual constraint of a zone. All operators we define in this section can be implemented using DBMs and the operations described by Bengtsson and Yi~\cite{Bengtsson2004}. For more information see \cite{FullReport}.

\begin{definition}[$\text{extract-virtual-constraint}$ Operator]
\label{def:virtual-clocks:CheckingForBoundedBisim:Extract-Virtual-Constraint-Operator}
Assume a set of clocks $C$. A constraint $\phi \in \mathcal{B}(\{\chi_0, \chi_1, ..., \chi_{i-1}\})$ is called the virtual constraint of a zone $\Zone{} \in \mathcal{D}(C \cup \{\chi_0, ..., \chi_{i-1}\})$, if and only if the following conditions hold:
\begin{enumerate}
\item (Soundness) if $\ClockValuation{\text{virtual}} \models \phi$, then there exists a $\ClockValuation{} \in D$ such that $\text{virt}(\ClockValuation{}) = \text{virt}(\ClockValuation{\text{virtual}})$,
\item (Completeness) if $\ClockValuation{} \in D$, then $u \models \phi$,
\item (Simple Structure) $\phi$ consists of up to $(i + 1)^2$ conjugated atomic constraints of the form $\phi_{j, k} = (\chi_{j - 1} - \chi_{k - 1} \preccurlyeq_{j, k} n_{j, k})$, with $i, j \in [1, i]$, and $\phi_{j, 0} = \chi_{j-1} \preccurlyeq_{j, 0} n_{j, 0}$ respectively $\phi_{0, j} = - \chi_{j-1} \preccurlyeq_{0, j} n_{0, j}$ with $\forall j, k \in [0, i] : \preccurlyeq_{j, k} \in \{<, \leq\} \land n_{j, k} \in \mathbb{N}^{\geq 0}$ such that $\phi = \bigwedge_{j, k \in [0, i]} \phi_{j, k}$, and
\item (Canonical) none of the atomic constraints of a virtual constraint can be strengthened without changing the solution set.
\end{enumerate}
The operator $\text{extract-virtual-constraint}: \mathcal{D}(C \cup \{\chi_0, ..., \chi_{i-1}\}) \rightarrow \mathcal{B}(\{\chi_0, \chi_1, ..., \chi_{i-1}\})$ takes a zone $D \in \mathcal{D}(C \cup \{\chi_0, ..., \chi_{i-1}\})$ and returns the virtual constraint of $D$.
\end{definition}

$\VirtualEquivalence{A}{B}$ holds if and only if $[\text{extract-virtual-constraint}(\Zone{A})] = [\text{extract-virtual-constraint}( \allowbreak \Zone{B})]$.
We may use the shorthand notation $\text{extract-virtual-constraint}(\{\Zone{0}, ..., \Zone{m}\})$ instead of $\{\text{extract-} \allowbreak \text{virtual-} \allowbreak \text{constraint}(\Zone{0}), ..., \text{extract-virtual-constraint}(\Zone{m})\}$.
We apply this convention for any operator that we introduce from now on. It should be noted that any AB-\synchronized \ symbolic state is uniquely identified by its location and its virtual constraint.

No logical disjunction is defined for virtual constraints. This has practical reasons as they are stored using DBMs, which are unable to handle such operations. To solve this issue, sets of virtual constraints are used, where a set of virtual constraints represents a disjunction of all included virtual constraints. Furthermore, we introduce the shorthand notation $\phi_{A} \land \lnot \phi_B$, which describes a set of virtual constraints such that a clock valuation $\ClockValuation{}$ satisfies exactly one of the resulting virtual constraints if and only if $\ClockValuation{} \models \phi_{A}$ and $\ClockValuation{} \not\models \phi_{B}$ holds.

\subsection{Virtual Bisimulation in Order $n$}

Using virtual constraints, we can present the \textsc{check-for-virt-bisim-in-order-impl} function in Algorithm~\ref{alg:algorithm:check-for-virt-bisim-in-order}. We remind the reader of the termination criterion described in Example~\ref{ex:virtual-clocks:virtual-constraint:initial-example}. In order to apply this termination criterion, a set of pairs of symbolic states that we assume to be virtual bisimilar in order $n$ is required. This set is referred to as \textit{visited}.

\begin{algorithm}[H]
  \caption{\textsc{check-for-virt-bisim-in-order-impl} function}
  \label{alg:algorithm:check-for-virt-bisim-in-order}
  \begin{algorithmic}[1]
  \LeftComment{Let $\VCGFullState{A}{}$, $\VCGFullState{B}{}$ be AB-\semisynchronized{} symbolic states,}
  \LeftComment{$k : C_A \cup C_B \cup \{\chi_{0}, ..., \chi_{|C_A| + |C_B| - 1}\} \rightarrow \mathbb{N}^{\geq 0}$,$\text{visited}$ be a set, and $n \in \mathbb{N}^{\geq 0}$.}
  \LeftComment{The return value of \textsc{check-for-virt-bisim-in-order-impl} is a set of virtual constraints.}
  \Function{check-for-virt-bisim-in-order-impl${}_{\VCG{A}, \VCG{B}, k, \textup{visited}, n}$}{$(l_A, \Zone{A})$, $(l_B, \Zone{B})$}
    \If{$n = 0 \lor \lnot (\VirtualEquivalence{A}{B})$}
      \State \Return $\text{extract-virtual-constraint}(\Zone{A}) \land \lnot \ \text{extract-virtual-constraint}(\Zone{B}) \cup $
      \State \qquad \qquad \qquad \qquad $\text{extract-virtual-constraint}(\Zone{B}) \land \lnot \ \text{extract-virtual-constraint}(\Zone{A})$
    \EndIf
    \State $(\VCGFullState{A}{e, A}, \VCGFullState{B}{e, B}) \gets \syncfunction(\VCGFullState{A}{}, \VCGFullState{B}{})$
    \State
    \State $\VCGFullState{A}{\text{norm}, A} \gets (l_A, \text{norm}(\Zone{e, A}, k))$, $\VCGFullState{B}{\text{norm}, B} \gets (l_B, \text{norm}(\Zone{e, B}, k))$
    \IfThen{($(\VCGFullState{A}{\text{norm}, A}, \VCGFullState{B}{\text{norm}, B}) \in \text{visited}$)}{\Return $\emptyset$}
    \State $\text{new-visited} \gets \text{visited} \cup \{(\VCGFullState{A}{\text{norm}, A}, \VCGFullState{B}{\text{norm}, B})\}$
  
    \State $\textsc{func} = \textsc{check-for-virt-bisim-in-order-impl}_{\VCG{A}, \VCG{B}, k, \textup{new-visited}, n-1}$
    \State
    \LeftComment{Assume $\VCGFullState{A}{e, A} \VCGTrans{\varepsilon} \VCGFullState{A}{\varepsilon, A}$ and $\VCGFullState{B}{e, B} \VCGTrans{\varepsilon} \VCGFullState{B}{\varepsilon, B}$.}
    \State $\varepsilon\text{-result} \gets $\Call{func}{$(l_A, \Zone{\varepsilon, A})$, $(l_B, \Zone{\varepsilon, B})$}
    \State $\text{sync-cond} \gets \text{revert-}\varepsilon \text{-trans}(\Zone{e, A}, \Zone{\varepsilon, A}, \varepsilon\text{-result} \land \text{extract-virtual-constraint}(\Zone{\varepsilon, A})))$
    \State \qquad \qquad \qquad $\cup \ \text{revert-}\varepsilon \text{-trans}(\Zone{e, B}, \Zone{\varepsilon, B}, \varepsilon\text{-result} \land \text{extract-virtual-constraint}(\Zone{\varepsilon, B}))$
    
    \IfThen{($\text{sync-cond} \neq \emptyset$)}{\Return $\revertsyncfunction(\VCGFullState{A}{}, \VCGFullState{B}{}, \text{sync-cond})$}
    \State
    \ForAll{$\sigma \in \Sigma$}
      \LeftComment{$\text{out-trans}(\sigma, \VCGFullState{}{})$ is the set of all outgoing transitions of $\VCGFullState{}{}$ labeled with $\sigma$.}
      \State $\text{sync-cond} \gets \textsc{check-outgoing-transitions-impl}_{\VCG{A}, \VCG{B}, \textsc{func}} (\Zone{e, A}, \Zone{e, B},$
      \State \qquad \qquad \qquad \qquad $\text{out-trans}(\sigma, \VCGFullState{A}{e, A}), \text{out-trans}(\sigma, \VCGFullState{B}{e, B}))$
      \IfThen{($\text{sync-cond} \neq \emptyset$)}{\Return $\revertsyncfunction(\VCGFullState{A}{}, \VCGFullState{B}{}, \text{sync-cond})$}
    \EndFor
    \State \Return $\emptyset$
  \EndFunction{}
  \end{algorithmic}
  \end{algorithm}

The first if-statement addresses the scenarios where either $n=0$ holds or where the symbolic states are not virtually equivalent. In both cases, the non-overlapping parts of the zones are returned. If $n=0$ and the zones are virtually equivalent, this is the empty set.

If $n>0$ and the symbolic states are virtually equivalent, we proceed with the application of the $\syncfunction{}$ function according to Definition \ref{def:virtual-clocks:VCG:CheckingForBoundedBisim:BoundedBiSim}. Afterwards, we check whether the given symbolic states are part of the assumptions, which are represented by the elements of the visited set. For this, we remind the reader that the k-normalized symbolic states are virtually bisimilar if and only if the original symbolic states are virtually bisimilar. Since there is only a finite number of normalized symbolic states, visited has a finite maximum size. If the normalized symbolic states are part of the assumptions, no contradiction is returned. In the absence of such an assumption, the normalized symbolic states are added to $\text{visited}$ for the reasons we have previously discussed in Example~\ref{ex:virtual-clocks:virtual-constraint:initial-example}.

For enhanced readability, we introduce the shorthand notation $\textsc{func}$ for the recursive call with the extended $\text{visited}$ set and order $n-1$. We first check the outgoing $\varepsilon$-transitions by applying \textsc{func} to the targets of the $\varepsilon$-transitions. $\textsc{func}$ returns the contradictions for $\VCGFullState{A}{\varepsilon, A}$ and $\VCGFullState{B}{\varepsilon, B}$. Since we are interested in the resulting contradictions for $\VCGFullState{A}{e, A}$ and $\VCGFullState{B}{e, B}$, we have to find those parts of $\VCGFullState{A}{e, A}$ and $\VCGFullState{B}{e, B}$ that have an outgoing transition leading to the contradictions. We do this by utilizing the so-called $\text{revert-}\varepsilon\text{-transition}$ operator, which uses the established past operator \cite{Bengtsson2004}. If we have found a contradiction for $\VCGFullState{A}{e, A}$ and $\VCGFullState{B}{e, B}$, we need to revert the $\syncfunction{}$ operation and return the resulting contradictions. Since Proposition~\ref{prop:virtual-clocks:VCG:correctness-sync-function} shows that $\syncfunction{}$ is essentially a reset operation, we can revert it with the well-known $\text{free}$ operation \cite{Bengtsson2004}. If we have not found a contradiction, we go on with the outgoing action transitions.

To check the action transitions, we iterate through all available actions and provide all outgoing transitions of a certain action to the $\textsc{check-outgoing-transitions-impl}$ operator. This operator will be explained in the subsequent section. It returns the contradictions of the outgoing action transitions. If we find a contradiction, we revert the $\syncfunction{}$ and return the result. Otherwise, we proceed with the next action. If all actions are checked without finding a contradiction, no contradiction exists and the empty set is returned.
To check for virtual bisimulation of a certain pair of symbolic states, we use the $\textsc{check-for-virt-}$ $\textsc{bisim-in-order-impl}$ function with an empty $\text{visited}$ set. 
The following example illustrates the algorithm in practice.

\begin{example}
We extend Example~\ref{ex:virtual-clocks:virtual-constraint:initial-example} here. First, we calculate
\begin{align*}
\textsc{check-for-virt-bisim-} & \textsc{in-order-impl}_{A_{1, \text{VCG}}, A_{2, \text{VCG}}, k, \emptyset, n} (\\
& \qquad\qquad (l_0, x_1 = \chi_0 = \chi_1 = 0), (l_0, x_2 = \chi_0 = \chi_1 = 0)).
\end{align*}

Since we assumed $n > 1$ and since the symbolic states are virtually equivalent, we omit the first if-statement. The symbolic states are already AB-\synchronized{} and the $\syncfunction{}$ function does not affect them. All values are equal to zero and k-normalization does not change them. Due to the fact that $\text{visited}$ is empty, the second if-statement is also omitted and the initial symbolic states are added to $\text{new-visited}$. Since we only consider the transitions depicted in Figure~\ref{fig:algorithm:initial-example:Illustration}, the outgoing action transitions at the initial symbolic states are ignored and we only check the outgoing $\varepsilon$-transitions.

We now apply the recursive call with parameters $n-1$ and the new $\text{visited}$ set to the targets of the $\varepsilon$-transitions. Since we assumed $n > 1$, we can conclude $(n-1) > 0$. Furthermore, since the zones are virtually equivalent, we can omit the first if-statement. Since the symbolic states are already AB-\synchronized{}, the $\syncfunction{}$ function does not change the symbolic states. Once again, k-normalization has no impact here, and the symbolic states included in $\text{new-visited}$ are not equal to the symbolic states under test. Consequently, the second if-statement can be ignored. The algorithm now generates another visited set. For reasons of uniqueness, we call this set "$\text{last-visited}$". $\text{last-visited}$ contains the pair of initial symbolic states and the current pair of symbolic states. We first check the outgoing $\varepsilon$-transitions. We consider two scenarios for the function
\begin{align*}
\textsc{check-for-virt-bisim-} & \textsc{in-order-impl}_{A_{1, \text{VCG}}, A_{2, \text{VCG}}, k, \text{last-visited}, n-2}( \\
& \qquad\qquad (l_0, x_1 = \chi_0 = \chi_1 < \infty), (l_0, x_2 = \chi_0 = \chi_1 < \infty)).
\end{align*}
If the statement $n-2 = 0$ is true, the empty set is returned as the zones are virtually equivalent. If the statement $n-2 > 0$ is true, the first if-statement is bypassed and since the given pair is an element of $\text{last-visited}$, the empty set is returned by the second if-statement. Therefore, in both cases, the expected value as described in Example~\ref{ex:virtual-clocks:virtual-constraint:initial-example} is returned.

We now check the outgoing action transition by calculating
\begin{align*}
\textsc{check-for-virt-bisim-} & \textsc{in-order-impl}_{A_{1, \text{VCG}}, A_{2, \text{VCG}}, k, \text{last-visited}, n-2}( \\
& \qquad\qquad(l_1, x_1 = \chi_0 = \chi_1 = 0), (l_1, x_2 = 0 \land \chi_0 = \chi_1 < \infty)).
\end{align*}
Since the zones are not virtually equivalent, the statement
\begin{align*}
& (\text{extract-virtual-constraint}(x_1 = \chi_0 = \chi_1 = 0) \land \\
& \qquad \lnot \text{extract-virtual-constraint}(x_2 = 0 \land \chi_0 = \chi_1 < \infty)) \\
 \cup \ & (\text{extract-virtual-constraint}(x_2 = 0 \land \chi_0 = \chi_1 < \infty) \land \\
 & \qquad \lnot \text{extract-virtual-constraint}(x_1 = \chi_0 = \chi_1 = 0)) \\
= &\ \{\chi_0 = \chi_1 > 0\}
\end{align*}
is returned, which is the return value we expected in Example~\ref{ex:virtual-clocks:virtual-constraint:initial-example}. Consequently, the second function call finds a contradiction returned by $\textsc{check-outgoing-transitions-impl}$ and returns it to the initial function call. The $\text{revert-}\varepsilon \text{-transition}$ operator is then applied, resulting into the contradiction $\chi_0 = \chi_1 = 0$, as the returned contradiction can be reached from the new contradiction. Ultimately, since the $\syncfunction{}$ function had no impact on the symbolic states, the $\revertsyncfunction{}$ also has no impact on the returned contradiction and the initial function call returns the contradiction we expected in  Example~\ref{ex:virtual-clocks:virtual-constraint:initial-example}.
\end{example}

The following section describes the \textsc{check-outgoing-transitions-impl} operator.

\subsection{Checking Two Sets of Transitions}

The $\textsc{check-outgoing-transitions-impl}$ operator is presented in Algorithm~\ref{alg:virtual-clocks:check-outgoing-transitions}. The cases in which at least one of the sets is empty are skipped here, as they are straightforward. The operator takes as input two sets of outgoing action transitions, both using the same label, and a function that checks for virtual bisimulation in order $n-1$. If contradictions exist, at least one of them is returned. 

While implementing this operator is straightforward in the case of deterministic transitions, the non-deterministic case is more challenging, as we know from Example~\ref{ex:virtual-clocks:VCG:non-determ}. Definition~\ref{def:virtual-clocks:VCG:CheckingForBoundedBisim:BoundedBiSim} requires the existence of a split satisfying the specified properties. If such a split exists, it must be identified. However, if no such split exists, we have to show nonexistence. Both tasks are non-trivial in the non-deterministic case, since the algorithm presented in the previous section only guarantees finding a contradiction if one exists, but does not guarantee that all contradictions will be detected.

The main idea of the operator is to analyze each pair of transitions, with one transition drawn from each set, until a contradiction is found that cannot be resolved by analyzing a different split or until it is clear that no such contradiction exists. To achieve this, we initialize two matrices, each with as many rows as the first transition set has elements and as many columns as the second transition set has elements. The first matrix contains all contradictions found for each pair of transitions, while the second matrix contains a boolean for each pair of transitions, indicating whether all contradictions for that particular pair have already been found. Since we do not know anything about any pair in the beginning, we initialize each element of the first matrix with the empty set (no contradictions found yet) and each element of the second matrix with false.

\begin{algorithm}[!ht]
  \caption{\textsc{check-outgoing-transitions-impl} function}
  \label{alg:virtual-clocks:check-outgoing-transitions}
  \begin{algorithmic}[1]
  \LeftComment{Let $\text{trans}_A$ and $\text{trans}_B$ be two sets of transitions and}
  \LeftComment{$\textsc{func} : (L_A \times \mathcal{D}(\CCupChi{A})) \times (L_B \times \mathcal{D}(\CCupChi{B})) \rightarrow 2^{\mathcal{B}(\{\chi_0, ..., \chi_{|C_A| + |C_B| - 1}\})}$}
  \LeftComment{\textsc{check-outgoing-transitions-impl} returns a set of virtual constraints}
  \Function{check-outgoing-transitions-impl${}_{\VCG{A}, \VCG{B}, \textsc{func}}$}{$\Zone{A}$, $\Zone{B}$, $\text{trans}_A$, $\text{trans}_B$}
    \State $\text{found-cont} \gets \text{matrix-of-size}(|\text{trans}_A| \times |\text{trans}_B|, \emptyset)$ \Comment{matrix of empty sets of vc}
    \State $\text{finished} \gets \text{matrix-of-size}(|\text{trans}_A| \times |\text{trans}_B|, \text{false})$ \Comment{matrix of booleans, all set to false}
  
    \Repeat
      \ForAll{$\VCGFullState{A}{} \VCGTrans{\sigma} \VCGFullState{\sigma, A}{} \in \text{trans}_A$}
          \ForAll{$\VCGFullState{B}{} \VCGTrans{\sigma} \VCGFullState{\sigma, B}{} \in \text{trans}_B$}
\LeftComment{Let $i_A$ and $i_B$ denote the index of the transition from $A$ and $B$, respectively.}
              \IfThen{$\text{finished}[i_A][i_B]$}{$\text{continue with the next pair}$}
              \State $\Zone{\text{eq}, A} \gets \Zone{\sigma, A} \land \text{extract-virtual-constraint}(\Zone{\sigma, B})$
              \State $\Zone{\text{eq}, B} \gets \Zone{\sigma, B} \land \text{extract-virtual-constraint}(\Zone{\sigma, A})$
              \State $\text{cont} \gets \textsc{check-target-pair-impl}_{\VCG{A}, \VCG{B}, \textsc{func}}($
              \State \qquad \qquad $\VCGFullState{\sigma, A}{\text{eq}, A}, \VCGFullState{\sigma, B}{\text{eq}, B}, \text{found-cont}[i_A][i_B])$
              \State $\text{finished}[i_A][i_B] \gets (\text{cont} = \emptyset)$ \Comment{finished is set to true if cont is empty.}
              \State $\text{found-cont}[i_A][i_B] \gets \text{found-cont}[i_A][i_B] \cup \text{cont}$
        \EndFor
      \EndFor
      \State $\text{contradiction} \gets \text{search-contradiction}(\Zone{A}, \Zone{B}, \text{trans}_A, \text{trans}_B, \text{found-cont})$
      \IfThen{$\text{contradiction} \neq \emptyset$}{\Return \text{contradiction}}
    \Until{$\text{no-contradiction-possible}(\Zone{A}, \Zone{B}, \text{trans}_A, \text{trans}_B, \text{found-cont}, \text{finished})$}
    \State \Return $\emptyset$
  \EndFunction{}
  \end{algorithmic}
\end{algorithm}

The for-loops iterate through each pair of transitions. To identify the corresponding matrix entries, we assume that each transition has a unique index within its respective set, which is denoted as $i_A$ and $i_B$. If a pair is already marked as finished, we skip it. 

Otherwise, the targets are made virtually equivalent and the \textsc{check-target-pair-impl} function is used. This function takes the virtually equivalent symbolic states and the already found contradictions of this pair and returns a set of new contradictions, if there exist more, or the empty set if no further contradictions exist. The \textsc{check-target-pair-impl} function will be explained in the next section. If there are no more contradictions, we mark this pair as finished. Otherwise, the newly found contradictions are added to the contradictions of this pair.

The search for contradictions is repeated until either a contradiction is found that cannot be removed by analyzing a different split or no such contradiction can be found anymore. This is indicated by the outer loop, which repeats until $\text{no-contradiction-possible}$ returns $\text{true}$, and the last if-statement, which returns the contradictions if any are present. The $\text{search-contradiction}$ operator iterates through each transition of $\text{trans}_A$ and $\text{trans}_B$, building the intersection of the contradiction sets belonging to this transition. If the result is not empty, it contains the contradictions that are contained by all sets of contradictions regarding a certain transition. Therefore, any split, done accordingly to Definition~\ref{def:virtual-clocks:VCG:CheckingForBoundedBisim:BoundedBiSim}, will have these contradictions and we return them. The $\text{no-contradiction-possible}$ operator uses only those entries in $\text{found-cont}$ for which the corresponding $\text{finished}$ entry is $\text{true}$. Subsequently, the $\text{search-contradiction}$ operator is invoked. If the $\text{search-contradiction}$ operator cannot find any contradiction using only those entries for which $\text{finished}$ is $\text{true}$, we can conclude that no contradictions will be found in the subsequent iterations, as the contradiction sets for those entries stay the same. Consequently, we terminate the outer loop. It should be noted that the outer loop will only be executed more than once if both sets of transitions contain non-deterministic choices. If one of the sets contains only deterministic choices, then any pair checked will either return a contradiction that is directly a contradiction for the deterministic choice, or none of the pairs will return a contradiction.

The practical application of $\textsc{check-outgoing-transitions-impl}$ is illustrated in the following example. We use the nondeterministic TA $A_5$ from Figure~\ref{fig:examples} for this.

\begin{example}
\label{ex:algorithm:checking-two-sets:nondeterministic-ex}
\begin{figure}
\centering
\scalebox{\AlgExampleAutomataScalingFactor}{
\begin{tikzpicture}
\tikzstyle{every node}=[font=\AlgTikzFontSize]
\tikzstyle{foo} = [draw,rectangle,minimum width=3.5cm, minimum height = 2cm, inner sep=5pt,thick]
\node[foo, minimum height = 3cm, align=center](callnminusone) {$(l_0, x = y = \chi_0 = \chi_1 = \chi_2 = \chi_3 < \infty)$ \\ $(l_0, x = y = \chi_0 = \chi_1 = \chi_2 = \chi_3 < \infty)$};
\node[foo, align=center, above=2cm of callnminusone, minimum width=6cm](calln) {$(l_0, x = y = \chi_0 = \chi_1 = \chi_2 = \chi_3 = 0)$ \\ $(l_0, x = y = \chi_0 = \chi_1 = \chi_2 = \chi_3 = 0)$};
\node[foo, align=center, above right=3cm and 5cm of callnminusone, minimum width=6cm](bothleft) {$(l_1', x = y = 0 \land \chi_0 = \chi_1 = \chi_2 = \chi_3 < \infty)$ \\ $(l_1', x = y = 0 \land \chi_0 = \chi_1 = \chi_2 = \chi_3 < \infty)$};
\node[foo, align=center, right=2cm of bothleft, minimum width=6cm](bothleftdots) {$\dots$};
\node[foo, align=center, below = 1cm of bothleft, minimum width=6cm](leftright) {$(l_1', x = y = 0 \land \chi_0 = \chi_1 = \chi_2 = \chi_3 < \infty)$ \\ $(l_1, x = 0 \land y = \chi_0 = \chi_1 = \chi_2 = \chi_3 < \infty)$};
\node[foo, align=center, right=2cm of leftright, minimum width=6cm](leftrightdots) {$\dots$};
\node[foo, align=center, below = 1cm of leftright, minimum width=6cm](rightleft) {$(l_1, x = 0 \land y = \chi_0 = \chi_1 = \chi_2 = \chi_3 < \infty)$ \\ $(l_1', x = y = 0 \land \chi_0 = \chi_1 = \chi_2 = \chi_3 < \infty)$};
\node[foo, align=center, right=2cm of rightleft, minimum width=6cm](rightleftdots) {$\dots$};
\node[foo, align=center, below = 1cm of rightleft, minimum width=6cm](bothright) {$(l_1, x = 0 \land y = \chi_0 = \chi_1 = \chi_2 = \chi_3 < \infty)$ \\ $(l_1, x = 0 \land y = \chi_0 = \chi_1 = \chi_2 = \chi_3 < \infty)$};
\node[foo, align=center, right=2cm of bothright, minimum width=6cm](bothrightdots) {$\dots$};
\AlgExampleArrowDesc ($(calln.north) + (1cm, 1.5cm)$) --node[left, align=center]{n} ($(calln.north) + (1cm, 0)$);
\AlgExampleArrowDesc ($(calln.north) + (-1cm, 1.5cm)$) --node[left, align=center]{$\emptyset$} ($(calln.north) + (-1cm, 0)$);
\AlgExampleArrowDesc ($(calln.south) + (1cm, 0)$) --node[left, align=center]{$\varepsilon$, n-1} ($(callnminusone.north) + (1cm, 0)$);
\AlgExampleArrowDesc ($(callnminusone.north) + (-1cm, 0)$) --node[left, align=center]{$\emptyset$} ($(calln.south) + (-1cm, 0)$);
\AlgExampleArrowDesc ($(callnminusone.north east)$) --node[left, align=center, yshift=2cm, xshift=1.5cm]{a, n-2} ($(bothleft.north west)$);
\AlgExampleArrowDesc ($(bothleft.south west)$) --node[below right, align=center, yshift=1.2cm, xshift=1.5cm]{$\emptyset$} ($(callnminusone.north east) + (0, -0.5cm)$);
\AlgExampleArrowDesc ($(bothleft.east)$) --node[above, align=center]{...} ($(bothleftdots.west)$);
\AlgExampleArrowDesc ($(callnminusone.north east) + (0, -1cm)$) --node[right, align=center, xshift=0.2cm]{a, n-2} ($(leftright.north west)$);
\AlgExampleArrowDesc ($(leftright.south west)$) -- ($(callnminusone.north east) + (0, -1.5cm)$);
\node[below right = 0.1cm and -1cm of leftright.south west, align=center, xshift=0.1cm]{$\{ \chi_0 = \chi_1 = \chi_2 = \chi_3 < \infty \}$};
\AlgExampleArrowDesc ($(leftright.east)$) --node[above, align=center]{...} ($(leftrightdots.west)$);
\AlgExampleArrowDesc ($(callnminusone.south east) + (0, 1.5cm)$) --node[below, align=center, xshift=1cm]{a, n-2} ($(rightleft.north west)$);
\AlgExampleArrowDesc ($(rightleft.south west)$) -- ($(callnminusone.south east) + (0, 1cm)$);
\node[below right = -0.1cm and -0.5cm of rightleft.south west, align=center, xshift=-0.2cm]{$\{ \chi_0 = \chi_1 = \chi_2 = \chi_3 < \infty \}$};
\AlgExampleArrowDesc ($(rightleft.east)$) --node[above, align=center]{...} ($(rightleftdots.west)$);
\AlgExampleArrowDesc ($(callnminusone.south east) + (0, 0.5cm)$) -- ($(bothright.north west)$);
\node[below left = -0.1cm and 0.2cm of bothright.north west, align=center]{a, n-2};
\AlgExampleArrowDesc ($(bothright.south west)$) --node[below left, align=center]{$\emptyset$} ($(callnminusone.south east)$);
\AlgExampleArrowDesc ($(bothright.east)$) --node[above, align=center]{...} ($(bothrightdots.west)$);
\end{tikzpicture}
}
\caption{Illustration of Example~\ref{ex:algorithm:checking-two-sets:nondeterministic-ex}}
\label{fig:algorithm:checking-two-sets:nondeterministic-ex}
\end{figure}
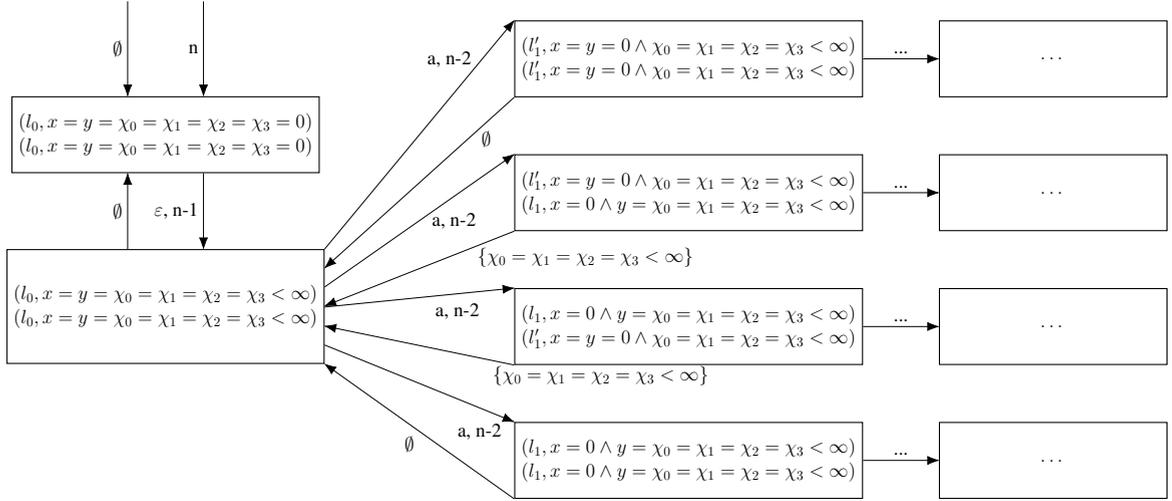
We use the $\textsc{check-for-virt-bisim-in-order-impl}$ function to check whether $A_5$ is bisimilar in order $n > 6$ to itself (which is obviously the case). The recursive calls are illustrated in Figure~\ref{fig:algorithm:checking-two-sets:nondeterministic-ex}. For the initial symbolic state, we ignore all outgoing action transitions. After the first $\varepsilon$-transition, we reach a pair that consists two times of the symbolic state $(l_0, x = y = \chi_0 = \chi_1 = \chi_2 = \chi_3 < \infty)$. Therefore, we apply the $\textsc{check-for-virt-bisim-in-order-impl}$ function to this pair. After the check for virtual equivalence and the check of the outgoing $\varepsilon$-transitions, we analyze the outgoing action transitions. Since both symbolic states have two outgoing transitions labeled with $a$, $\text{found-cont}$ and $\text{finished}$ are both 2x2-matrices. In accordance with the illustration in Figure~\ref{fig:algorithm:checking-two-sets:nondeterministic-ex}, we know that after the for-loops have been executed, the matrices have the following entries:
\begin{equation*}
\begin{pmatrix}
\emptyset & \{ \chi_0 = \chi_1 = \chi_2 = \chi_3 < \infty \} \\
\{ \chi_0 = \chi_1 = \chi_2 = \chi_3 < \infty \} & \emptyset
\end{pmatrix} \text{ and }
\begin{pmatrix}
\text{true} & \text{false} \\
\text{false} & \text{true}
\end{pmatrix}
\end{equation*}
Since the conjunction of the disjunction of the contradiction sets of any row or column is equal to false, no contradictions have been found for any transition. From the second matrix, we know that the upper left and lower right elements will remain unaltered upon the completion of another iteration. After we have replaced all other entries by $\text{true}$, it becomes evident that the $\text{search-contradiction}$ operator is still unable to find any contradictions. Therefore, no other iteration will result into any contradictions and the outer loop is terminated, which results into the return of the empty set.
\end{example}

We show the \textsc{check-target-pair-impl} operator in the next section.

\subsection{Checking a Pair of Targets}

In the deterministic case, checking a pair of target symbolic states is straightforward: We simply apply the \textsc{check-for-virt-bisim-in-order} function with order $n-1$. In the non-deterministic case it is necessary to have the ability to remove previously found contradictions, ensuring that only those parts of the target symbolic states are compared for which no contradiction has been found.

Algorithm \ref{alg:virtual-clocks:check-target-pair} shows the \textsc{check-target-pair-impl} function. If $\text{found-cont}$ is empty, the operator is trivial: We invoke $\textsc{func}$ on the given symbolic states and return the result (we use the convention $\bigvee \emptyset = \emptyset$ here). If $\text{found-cont}$ is not empty, it is necessary to extract the parts of the target zones without an already found contradiction. This is achieved by negating the contradictions, which results in a different set of virtual constraints representing the parts of the zones, for which no contradiction has been found, and iterate through that set. In the case that a contradiction is encountered, we return it.

\begin{algorithm}
\caption{\textsc{check-target-pair-impl} function}
\label{alg:virtual-clocks:check-target-pair}
\begin{algorithmic}[1]
\LeftComment{Let $\VCGFullState{\sigma, A}{}$ be a symbolic state of $\VCG{A}$, $\VCGFullState{\sigma, B}{}$ be a symbolic state of $\VCG{B}$,}
\LeftComment{$\text{found-cont}$ be a set of contradictions, and}
\LeftComment{$\textsc{func} : (L_A \times \mathcal{D}(\CCupChi{A})) \times (L_B \times \mathcal{D}(\CCupChi{B})) \rightarrow 2^{\mathcal{B}(\{\chi_0, ..., \chi_{|C_A| + |C_B| - 1}\})}$.}
\LeftComment{\textsc{check-target-pair-impl} returns a set of virtual constraints}
\Function{check-target-pair-impl${}_{\VCG{A}, \VCG{B}, \textsc{func}}$}{$(l_{\sigma, A}, \Zone{\sigma, A})$, $(l_{\sigma, B}, \Zone{\sigma, B})$, $\text{found-cont}$}

  \State $\text{without-cont} \gets \text{extract-virtual-constraint}(\Zone{\sigma, A}) \land \lnot (\bigvee \text{found-cont})$
  \State $\text{contradictions} \gets \emptyset$

  \For{$\phi \in \text{without-cont}$}
    \State $\text{new-cont} \gets \textsc{func}( (l_{\sigma, A}, \Zone{\sigma, A} \land \phi), (l_{\sigma, B}, \Zone{\sigma, B} \land \phi))$
    \State $\text{contradictions} = \text{contradictions} \cup \text{new-cont}$
  \EndFor
  \State \Return $\text{contradictions}$

\EndFunction{}
\end{algorithmic}
\end{algorithm}

We now show that our algorithm is correct.

\subsection{Correctness of the Algorithm}

Previously, we have explained the desired behavior of our algorithm. We formalize these requirements by the following definition.

\begin{definition}[$\text{check-for-virt-bisim-in-order}$ function]
  \label{def:appendix:check-for-virt-bisim-in-order:check-for-virt-bisim-in-order-function}
  Assume two TA $A$, $B$, using the sets of clocks $C_A$, $C_B$, the AB-\semisynchronized \ symbolic states $\VCGFullState{A}{}$ of the VCG $\VCG{A}$ of $A$ regarding $B$ and $\VCGFullState{B}{}$ of the VCG $\VCG{B}$ of $B$ regarding $A$. Moreover, let $n \in \mathbb{N}^{\geq 0}$, $\text{\textit{visited}}$ be a set of pairs of symbolic states, and $k : C_A \cup C_B \cup \{\chi_0, ..., \chi_{|C_A| + |C_B| - 1}\} \rightarrow \mathbb{N}^{\geq 0}$ be a function, such that for any clock constraint $c \sim n$, which occurs in $A$ or $B$, $n < k(c)$, for any $i \in [0, |C_A| - 1] : k(C_A[i]) = k(\chi_i)$, and for any $i \in [0, |C_B| - 1] : k(C_B[i]) = k(\chi_{i + |C_A|})$. We define a function 
  \begin{equation*}
    \text{check-for-virt-bisim-in-order}_{n, \VCG{A}, \VCG{B}, k, \text{visited}}(\VCGFullState{A}{}, \VCGFullState{B}{}) = \{\phi_0, ..., \phi_m\}
  \end{equation*}
  with $\phi_0, ..., \phi_m$ being virtual constraints, such that
  \begin{enumerate}
    \item (Usability) $\phi \in \{\phi_0, ..., \phi_m\} : [\phi] \neq \emptyset$, $\forall i, j \in [0, m] : [\phi_i \land \phi_j] \neq \emptyset \text{ implies } i = j$, and
    \begin{equation*}
      \forall \phi \in \{\phi_0, ..., \phi_m\} : [\phi] \subseteq [\text{extract-virtual-constraint}(\Zone{A})] \cup [\text{extract-virtual-constraint}(\Zone{B})],
    \end{equation*}
    \item (Soundness) for any non-empty zone $\Zone{\text{sub}, A} \subseteq \Zone{A}$ with $\exists \phi \in \{\phi_0, ..., \phi_m\} : \Zone{\text{sub}, A} \land \phi \neq \emptyset$ exists no $\Zone{\text{sub}, B} \subseteq \Zone{B}$ such that $\VCGBisimulation{\VCGFullState{A}{\text{sub}, A}}{\VCGFullState{B}{\text{sub}, B}}{n}$ and vice versa, and
    \item (Completeness) if $\forall (\VCGFullState{v, A}{}, \VCGFullState{v, B}{}) \in \text{visited} : \VCGBisimulation{\VCGFullState{v, A}{}}{\VCGFullState{v, B}{}}{n}$ and $\text{check-for-virt-bisim-in-order}_{n, \VCG{A}, \VCG{B}, k, \text{visited}}(\VCGFullState{A}{}, \VCGFullState{B}{}) = \emptyset$, then for any $\Zone{\text{sub}, A} \subseteq \Zone{A}$ with $\Zone{\text{sub}, A} \neq \emptyset$ the statement $\VCGBisimulation{\VCGFullState{A}{\text{sub}, A}}{(l_B, \Zone{B} \land \text{extract-virtual-constraint}(\Zone{\text{sub}, A}))}{n}$ and the analog for $B$ hold.
  \end{enumerate}
\end{definition}

We now show that \textsc{check-for-virt-bisim-in-order-impl} is such a function. To do so, we first have to show some properties of \textsc{check-outgoing-transitions-impl} of Algorithm~\ref{alg:virtual-clocks:check-outgoing-transitions}.

\begin{proposition}[\textsc{check-outgoing-transitions-impl}]
  \label{prop:appendix:check-for-virt-bisim-in-order:check-outgoing-transitions}
  Assume two diagonal-free TA $A$, $B$, using the sets of clocks $C_A$, $C_B$, the VCG $\VCG{A}$ of $A$ regarding $B$, the VCG $\VCG{B}$ of $B$ regarding $A$ and two AB-\synchronized{} symbolic states $\VCGFullState{A}{}$, $\VCGFullState{B}{}$ with $\VirtualEquivalence{A}{B}$. Assume a $\sigma \in \Sigma$ and let $\text{trans}_A \in 2^{\VCGTrans{}}$ be the set of all outgoing transitions of $\VCGFullState{A}{}$ labeled with $\sigma$ and let $\text{trans}_B \in 2^{\VCGTrans{}}$ be the set of all outgoing transitions of $\VCGFullState{B}{}$ labeled with $\sigma$. We assume a $n \in \mathbb{N}^{\geq 0}$ and $\textsc{func}$ to be a $\text{check-for-virt-bisim-in-order}_{n, \VCG{A}, \VCG{B}, k, \text{visited}}$ function. We denote
  \begin{align*}
    \textsc{check-outgoing-transitions-impl}_{\VCG{A}, \VCG{B}, \textsc{func}} (\Zone{A}, \Zone{B}, \text{trans}_A, \text{trans}_B) = \{\phi_0, ..., \phi_m\}.
  \end{align*}
  $\{\phi_0, ..., \phi_m\}$ are virtual constraints such that
  \begin{enumerate}
    \item (Usability) $\forall \phi \in \{\phi_0, ..., \phi_m\} : [\phi] \neq \emptyset$, $\forall i, j \in [0, m] : [\phi_i \land \phi_j] \neq \emptyset$ implies $i = j$, and
    \begin{equation*}
    \forall \phi \in \{\phi_0, ..., \phi_m\} : [\phi] \subseteq [\text{extract-virtual-constraint}(\Zone{A})] ( = [\text{extract-virtual-constraint}(\Zone{B})]),
    \end{equation*}
    \item (Soundness) for any non-empty zone $\Zone{\text{sub}, A} \subseteq \Zone{A}$ with $\exists \phi \in \{\phi_0, ..., \phi_m\} : \Zone{\text{sub}, A} \land \phi \neq \emptyset$ either exists a transition $\VCGFullState{A}{\text{sub}, A} \VCGTrans{\sigma} \VCGFullState{\sigma, A}{\text{sub}, \sigma, A}$ such that there exists \textit{no} finite set of symbolic states $\{\VCGFullState{\sigma, A}{0, \text{sub}, \sigma, A}, \VCGFullState{\sigma, A}{1, \text{sub}, \sigma, A}, ...\}$ with $(\bigcup \Zone{i, \text{sub}, \sigma, A}) = \Zone{\sigma, A}$ such that for any $\VCGFullState{\sigma, A}{i, \text{sub}, \sigma, A}$ there exists a transition $\VCGFullState{B}{} \VCGTrans{\sigma} \VCGFullState{\sigma, B}{}$ such that there exists a zone $\Zone{\text{sub}, \sigma, B} \subseteq \Zone{\sigma, B}$ with $\textsc{func}(\VCGFullState{\sigma, A}{i, \sigma, A}, \VCGFullState{\sigma, B}{\text{sub}, \sigma, B}) = \emptyset$ or the analog for $B$, and
    \item (Completeness) $\{\phi_0, ..., \phi_m\} = \emptyset$ implies that for any transition  $\VCGFullState{A}{} \VCGTrans{\sigma} \VCGFullState{\sigma, A}{}$ exists a finite set of symbolic states $\{\VCGFullState{\sigma, A}{0, \sigma, A}, \VCGFullState{\sigma, A}{1, \sigma, A}, ...\}$ with $(\bigcup \Zone{i, \sigma, A}) = \Zone{\sigma, A}$ such that for any $\VCGFullState{\sigma, A}{i, \sigma, A}$ there exists a transition $\VCGFullState{B}{} \VCGTrans{\sigma} \VCGFullState{\sigma, B}{}$ such that there exists a zone $\Zone{\text{sub}, \sigma, B} \subseteq \Zone{\sigma, B}$ with
    \begin{equation*}
    \textsc{func}(\VCGFullState{\sigma, A}{i, \sigma, A}, \VCGFullState{\sigma, B}{\text{sub}, \sigma, B}) = \emptyset
    \end{equation*}
    and the analog statement for $B$.
  \end{enumerate}
\end{proposition}

\begin{proof}
The for-loops terminate as we defined TA to be finite. The search-contradiction operator can easily handle the case in which one of the transition sets is empty. Therefore, we assume this not to be the case. Moreover, we assume \textsc{check-target-pair-impl} to return the empty set if and only if $\text{found-cont}$ contains all contradictions of $(l_{\sigma, A}, \Zone{\text{eq}, A})$ and $(l_{\sigma, B}, \Zone{\text{eq}, B})$. Otherwise, it returns a new valid contradiction. This assumption can be shown straightforward.

A zone can be split into its regions (however, the region might not be fully included in the zone). This number of regions is finite. Since virtual constraints are clock constraints and, therefore, can only use natural numbers for comparison, $\text{extract-virtual-constraint}(\Zone{A}) \land \lnot \ \text{extract-virtual-constraint}(\Zone{B})$ (and analog for $B$) always returns at least one of those regions (or that part of the region that is contained in the zone). Therefore, there are finitely many contradictions possible for each pair of symbolic states. We will use this later to show termination of the outer loop. We now analyze the two for-loops. We skip the proof for Usability.

\noindent \textbf{Precondition:} For any valid indices $i_A$ and $i_B$, any $\phi \in \text{found-cont}[i_A][i_B]$ is a valid contradiction for the corresponding symbolic states. $\text{finished}[i_A][i_B] = \text{true}$ implies $\textsc{check-target-pair-impl}((l_{\sigma, A}, \allowbreak \Zone{\text{eq}, A}), (l_{\sigma, B}, \Zone{\text{eq}, B}), \text{found-cont}[i_A][i_B]) = \emptyset$. \\
\textbf{Postcondition:} The precondition still holds and for any valid indices $i_A$ and $i_B$ either $\text{finished}[i_A][i_B] \allowbreak = \text{true}$ or $[ \ \text{found-cont}[i_A][i_B] \ ]$ is larger than before.

Given the precondition, the postcondition can be shown by using Definition~\ref{def:appendix:check-for-virt-bisim-in-order:check-for-virt-bisim-in-order-function}. 

Since $\text{found-} \allowbreak \text{cont}[i_A][i_B]$ cannot grow infinitely, more and more entries of $\text{finished}$ will become $ \text{true}$ until the outer loop terminates (at the latest when all finished entries are set to true). Therefore, we only have to show that after the outer loop terminates, Soundness and Completeness hold.

To show Soundness, we assume a $\Zone{\text{sub}, A} \subseteq \Zone{A}$,
\begin{align*}
  &\ \textsc{check-outgoing-transitions-impl}_{\VCG{A}, \VCG{B}, \textsc{func}} (\Zone{A}, \Zone{B}, \text{trans}_A, \text{trans}_B) \\
= &\ \{\phi_0, ..., \phi_m\} \neq \emptyset,
\end{align*}
and $\exists \phi \in \{\phi_0, ..., \phi_m\} : \Zone{\text{sub}, A} \land \phi \neq \emptyset$. By definition of search-contradiction and the postcondition of the for-loops, this can only be the case, if there is a transition $\VCGFullState{A}{A} \VCGTrans{\sigma} \VCGFullState{\sigma, A}{\sigma, A}$, a $\phi_{\sigma}$ such that for any transition $\VCGFullState{B}{B} \VCGTrans{\sigma} \VCGFullState{\sigma, B}{\sigma, B}$, $\phi_{\sigma}$ is a valid contradiction for $(\VCGFullState{\sigma, A}{\sigma, A}, \VCGFullState{\sigma, B}{\sigma, B})$, and there is a transition $(l_A, \Zone{\text{sub}, A} \land \phi) \VCGTrans{\sigma} \VCGFullState{\sigma, A}{\sigma, A}$ with $\Zone{\sigma, A} \land \phi_{\sigma} \neq \emptyset$. In any split $\{\VCGFullState{\sigma, A}{0, \text{sub}, \sigma, A}, \VCGFullState{\sigma, A}{1, \text{sub}, \sigma, A}, ...\}$ with $(\bigcup \Zone{i, \text{sub}, \sigma, A}) = \Zone{\sigma, A}$, there exists a $\VCGFullState{\sigma, A}{i, \text{sub}, \sigma, A}$ with $\Zone{\sigma, A} \land \phi_{\sigma} \neq \emptyset$ and, therefore, there is no $\VCGFullState{B}{} \VCGTrans{\sigma} \VCGFullState{\sigma, B}{}$ such that there exists a zone $\Zone{\text{sub}, \sigma, B} \subseteq \Zone{\sigma, B}$ with $(l_A, \Zone{\text{sub}, A} \land \phi) \VCGBisimulationBare{n} (l_B, \Zone{\text{sub}, B})$ and Soundness holds. Soundness for $B$ can be shown analogously.

To show Completeness, we assume any transition $\VCGFullState{A}{A} \VCGTrans{\sigma} \VCGFullState{\sigma, A}{\sigma, A}$ and
\begin{align*}
  \textsc{check-outgoing-transitions-impl}_{\VCG{A}, \VCG{B}, \textsc{func}} (\Zone{A}, \Zone{B}, \text{trans}_A, \text{trans}_B) = \emptyset.
\end{align*}
In this case, we can split $\VCGFullState{\sigma, A}{\sigma, A}$ with index $i_A$ into $\{\VCGFullState{\sigma, A}{0, \sigma, A}, \VCGFullState{\sigma, A}{1, \sigma, A}, ...\}$ with $(\bigcup \Zone{i, \sigma, A}) = \Zone{\sigma, A}$ such that for any $\VCGFullState{\sigma, A}{i, \sigma, A}$ there exists an $i_B$ with transition $\VCGFullState{B}{} \VCGTrans{\sigma} \VCGFullState{i, \sigma, B}{},$ $\text{finished}[i_A][i_B] = \text{true}$, and $\text{found-cont}[i_A][i_B] \land \Zone{i, \sigma, A} = \emptyset$, as we know that the intersection of those contradiction sets is empty by the fact that no-contradiction-possible returned true. This implies that there exists a subzone $\Zone{\text{sub}, i, \sigma, B} \subseteq \Zone{i, \sigma, B}$ such that $(l_{\sigma, A}, \Zone{i, \sigma, A}) \VCGBisimulationBare{n} \VCGFullState{\text{sub}, i, \sigma, B}{}$ and Completeness holds. Completeness for $B$ can be shown analogously.
\end{proof}

Using Proposition~\ref{prop:appendix:check-for-virt-bisim-in-order:check-outgoing-transitions}, we can show that \textsc{check-for-virt-bisim-in-order-impl} is a check-for-virt-bisim-in-order function.

\begin{proposition}[\textsc{check-for-virt-bisim-in-order-impl} is Correct]
  \label{prop:appendix:check-for-virt-bisim-in-order-impl:check-for-virt-bisim-in-order-impl-is-correct}
  \textsc{check-for-virt-bisim-in-order-impl} is a $\text{check-for-virt-bisim-in-order}$ function.
\end{proposition}

\begin{proof}
Since $\Sigma$ is finite and with every recursion step, $n$ is decreased by one and, therefore, will eventually become zero, we skip the termination proof. For every recursion step, we define a precondition and an invariant. \\
\textbf{Precondition: } We assume the preconditions of Definition~\ref{def:appendix:check-for-virt-bisim-in-order:check-for-virt-bisim-in-order-function}. \\
\textbf{Invariant: } After the function returns, the properties of Definition~\ref{def:appendix:check-for-virt-bisim-in-order:check-for-virt-bisim-in-order-function} hold.

\noindent We show the invariant by induction.

\noindent\textbf{Base Case: } If $n = 0$, the condition of the first if statement evaluates to true. Therefore,
\begin{align*}
  & \text{extract-virtual-constraint}(\Zone{A}) \land \lnot \text{extract-virtual-constraint}(\Zone{B}) \cup \\
  & \qquad \qquad \text{extract-virtual-constraint}(\Zone{B}) \land \lnot \text{extract-virtual-constraint}(\Zone{A})
\end{align*}
is returned. Usability holds straightforward. We show Soundness only for a non-empty zone $\Zone{\text{sub}, A} \subseteq \Zone{A}$ since the statement for a subzone of $\Zone{B}$ can be shown analogously. 
\begin{align*}
  & \exists \phi \in \text{extract-virtual-constraint}(\Zone{A}) \land \lnot \text{extract-virtual-constraint}(\Zone{B}) \cup \\
  & \qquad \text{extract-virtual-constraint}(\Zone{B}) \land \lnot \text{extract-virtual-constraint}(\Zone{A}) : \Zone{\text{sub}, A} \land \phi \neq \emptyset 
\end{align*}
implies $\exists \phi \in \lnot \text{extract-virtual-constraint}(\Zone{B}) : \Zone{\text{sub}, A} \land \phi \neq \emptyset$ since a subzone of $\Zone{A}$ cannot fulfill any virtual constraint of $\lnot \text{extract-virtual-constraint}(\Zone{A})$.
By Definition~\ref{def:virtual-clocks:CheckingForBoundedBisim:Extract-Virtual-Constraint-Operator}, there cannot be any subzone of $\Zone{B}$ which is virtually equivalent to $\Zone{\text{sub}, A}$, which is required for virtual bisimulation in order $0$ by Definition~\ref{def:virtual-clocks:VCG:CheckingForBoundedBisim:BoundedBiSim}. Therefore, Soundness holds.
The return value can only be the empty set if and only if $\VirtualEquivalence{A}{B}$ holds. Therefore, Completeness also holds.

\noindent\textbf{Induction Step: } If $\Zone{A}$ and $\Zone{B}$ are not virtual equivalent, we can show all statements analogously to the base case. Therefore, we assume the zones to be virtually equivalent. 

If the pair of normalized symbolic states is element of $\text{visited}$, the condition of the second if statement evaluates to true, $\emptyset$ is returned, and the induction step holds by the precondition regarding the visited set.
If the pair of normalized symbolic states is not element of $\text{visited}$, either $\emptyset$ or the return value of $\text{revert-sync}$ are returned.
Usability holds straightforward. To show Soundness, we assume a non-empty zone $\Zone{\text{sub}, A} \subseteq \Zone{A}$ with
\begin{align*}
  \exists \phi \in & \textsc{check-for-virt-bisim-in-order-impl}_{\VCG{A}, \VCG{B}, k, \text{visited}, n}(\VCGFullState{A}{}, \VCGFullState{B}{}) : \\
  & \qquad \Zone{\text{sub}, A} \land \phi \neq \emptyset.
\end{align*}
$\Zone{\text{sub}, B} = \Zone{B} \land \text{extract-virtual-} \allowbreak \text{constraint}(\Zone{\text{sub}, A})$ is the only subzone of $\Zone{B}$ being virtual equivalent to $\Zone{\text{sub}, A}$. We denote $(\VCGFullState{A}{\text{sub}, e, A}, \VCGFullState{B}{\text{sub}, e, B}) = \text{sync}(\VCGFullState{A}{\text{sub}, A}, \VCGFullState{B}{\text{sub}, B})$. By definition, $\Zone{\text{sub}, e, A} \subseteq \Zone{e, A}$ and $\Zone{\text{sub}, e, B} \subseteq \Zone{e, B}$. Since the revert-sync operator returned a set of virtual constraints, we know that $\exists \phi_e \in \text{sync-cond} : \Zone{e, \text{sub}, A} \land \phi_e \neq \emptyset$. Since the returned value is not $\emptyset$, either 
\begin{align*}
  \exists \sigma \in \Sigma:\ & \textsc{check-outgoing-transitions-impl}_{\VCG{A}, \VCG{B}, \textsc{func}}( \\
  & \qquad \Zone{e, A}, \Zone{e, B}, \text{out-trans}(\sigma, \VCGFullState{A}{e, A}), \text{out-trans}(\sigma, \VCGFullState{B}{e, B})) \neq \emptyset
\end{align*}
or $\textsc{func}(\VCGFullState{A}{\varepsilon, A}, \VCGFullState{B}{\varepsilon, B}) \neq \emptyset$ holds.
In the first case, Soundness holds by Proposition~\ref{prop:appendix:check-for-virt-bisim-in-order:check-outgoing-transitions} and the induction hypothesis, which states that $\textsc{func}$ is a check-for-virt-bisim-in-order function for n-1. Therefore, we focus on the second case. 

We denote $\textsc{func}(\VCGFullState{A}{\varepsilon, A}, \VCGFullState{B}{\varepsilon, B}) = \{\phi_{\varepsilon, 0}, ..., \phi_{\varepsilon, p}\} \neq \emptyset$.
We denote the target of the $\varepsilon$-transition of $\VCGFullState{A}{e, \text{sub}, A}$ with $\VCGFullState{A}{\varepsilon, \text{sub}, A}$ and remind the reader of the fact that $\Zone{e, \text{sub}, A} \subseteq \Zone{e, A}$ implies $\Zone{\varepsilon, \text{sub}, A} \subseteq \Zone{\varepsilon, A}$ and analogous for $B$. By the given assumptions, $\exists \phi_{\varepsilon} \in \varepsilon\text{-result} : \Zone{\varepsilon, \text{sub}, A} \land \phi_{\varepsilon} \neq \emptyset$.
By the induction hypothesis, we know that $\textsc{func}$ is a check-for-virt-bisim-in-order function for order n-1. Therefore, for $\Zone{\varepsilon, \text{sub}, A}$ exists no $\Zone{\text{sub}, \varepsilon, B} \subseteq \Zone{\varepsilon, e, B}$ such that $\VCGFullState{A}{\text{sub}, \varepsilon, A} \VCGBisimulationBare{n} \VCGFullState{B}{\text{sub}, \varepsilon, B}$.
Therefore, $\VCGFullState{A}{\text{sub}, e, A}$ has an outgoing $\varepsilon$-transition, such that the target is not virtual bisimilar to any symbolic substate of the target of the $\varepsilon$-transition of $\VCGFullState{B}{\text{sub}, e, B}$ and Soundness holds by Definition~\ref{def:virtual-clocks:VCG:CheckingForBoundedBisim:BoundedBiSim}.
  
To show Completeness, we assume that $\textsc{check-for-virt-bisim-in-order-impl}$ returns $\emptyset$. Therefore,
\begin{align*}
  \forall \sigma \in \Sigma:\ & \textsc{check-outgoing-transitions-impl}_{\VCG{A}, \VCG{B}, \textsc{func}}( \\
  & \qquad \Zone{e, A}, \Zone{e, B}, \text{out-trans}(\sigma, \VCGFullState{A}{}), \text{out-trans}(\sigma, \VCGFullState{B}{})) = \emptyset
\end{align*}
and $\textsc{func}(\VCGFullState{A}{\varepsilon, A}, \VCGFullState{B}{\varepsilon, B}) = \emptyset$ hold. We assume any non-empty $\Zone{\text{sub}, A} \subseteq \Zone{A}$. Since $\Zone{\text{sub}, A} \VirtualEquivalenceBare \Zone{B} \land \text{extract-virtual-constraint}(\Zone{\text{sub}, A})$ obviously holds, we know by Definition~\ref{def:virtual-clocks:VCG:CheckingForBoundedBisim:BoundedBiSim} that we have to analyze the outgoing transitions of
\begin{equation*}
  (\VCGFullState{A}{\text{sub}, e, A}, \VCGFullState{B}{\text{sub}, e, B}) = \text{sync}(\VCGFullState{A}{\text{sub}, A}, (l_B, \Zone{B} \land \text{extract-virtual-constraint}(\Zone{\text{sub}, A}))).
\end{equation*}
For any outgoing action transition of either $\VCGFullState{A}{\text{sub}, e, A}$ or $\VCGFullState{B}{\text{sub}, e, B}$, the existing of a corresponding outgoing transition of the other symbolic state can be shown by the induction hypothesis and Proposition~\ref{prop:appendix:check-for-virt-bisim-in-order:check-outgoing-transitions}. We denote the target of the $\varepsilon$-transition of $\VCGFullState{A}{\text{sub}, e, A}$ with $\VCGFullState{A}{\text{sub}, \varepsilon, A}$ and the target of the $\varepsilon$-transition of $\VCGFullState{B}{\text{sub}, e, B}$ with $\VCGFullState{B}{\text{sub}, \varepsilon, B}$. We can show $\VirtualEquivalence{\text{sub}, \varepsilon, A}{\text{sub}, \varepsilon, B}$ and, therefore, the induction hypothesis implies $\VCGBisimulation{\VCGFullState{A}{\text{sub}, \varepsilon, A}}{\VCGFullState{B}{\text{sub}, \varepsilon, B}}{n-1}$ due to $\textsc{func}(\VCGFullState{A}{\varepsilon, A}, \VCGFullState{B}{\varepsilon, B}) = \emptyset$. Therefore, Completeness holds.
\end{proof}

Therefore, we can use $\textsc{check-for-virt-bisim-in-order-impl}$ to check for virtual bisimulation.

\begin{corollary}[$\textsc{check-for-virt-bisim-in-order-impl}$ is correct]
\label{cor:algorithm:algorithm:check-for-virt-bisim-in-order-impl-is-correct}
Assume two TA $A$, $B$, using the sets of clocks $C_A$, $C_B$, the initial symbolic states $\VCGFullState{0, A}{}$ of the VCG $\VCG{A}$ of $A$ regarding $B$ and $\VCGFullState{0, B}{}$ of the VCG $\VCG{B}$ of $B$ regarding $A$. Let $n \in \mathbb{N}^{\geq 0}$ and $k$ be a normalization function. We denote
\begin{equation*}
\textsc{check-for-virt-bisim-in-order-impl}_{\VCG{A}, \VCG{B}, k, \emptyset, n}(\VCGFullState{0, A}{}, \VCGFullState{0, B}{}) = \{ \phi_0, ..., \phi_m \}.
\end{equation*}
$\{ \phi_0, ..., \phi_m \} = \emptyset$ holds if and only if $\TLTSBisimulation{A}{B}{n}$ holds.
\end{corollary}

\begin{shortproof}
This is a direct consequence of Proposition~\ref{prop:appendix:check-for-virt-bisim-in-order-impl:check-for-virt-bisim-in-order-impl-is-correct}.
\end{shortproof}

Therefore, the problem on how to split the target zone is solved. However, the algorithm still suffers from the problem described in Example~\ref{ex:background:alternating-sequences}.

\subsection{Virtual Bisimulation}

In the preceding section, we used the \textsc{check-for-virt-bisim-in-order-impl} function to check for virtual bisimulation in any order $n$. However, from \cite{Bengtsson2004}, we know that the number of k-normalized symbolic states is finite and there is an upper bound. Let $N$ be the upper bound for the number of pairs of k-normalized symbolic states. Since in each recursion step one element is added to $\text{visited}$, we can follow that the recursion terminates after no more than $N$ recursion steps. Therefore, if we check for virtual bisimulation in order $N+1$, the order will never become zero and the algorithm will not change its output if any larger value for $n$ is used. The validity of Proposition~\ref{prop:appendix:check-for-virt-bisim-in-order-impl:check-for-virt-bisim-in-order-impl-is-correct} remains intact and, therefore, this function is able to check for virtual bisimulation. However, it suffers from the problem described in Example~\ref{ex:background:alternating-sequences}.

From Yi et al.~\cite{Yi1995}, we know that for any outgoing $\varepsilon$-transitions $\VCGFullState{A}{} \VCGTrans{\varepsilon} \VCGFullState{A}{A, \varepsilon}$ the statement $\Zone{A} \subseteq \Zone{A, \varepsilon}$ and for any outgoing $\varepsilon$-transitions $\VCGFullState{A}{A, \varepsilon} \VCGTrans{\varepsilon} \VCGFullState{A}{A, \varepsilon, \varepsilon}$ holds. Due to $\Zone{A} \subseteq \Zone{A, \varepsilon}$ (and analogously for B), it suffices to check the targets of the outgoing $\varepsilon$-transitions. Any contradiction for the original pair is still valid for the new pair (according to the hypotheses used to prove Theorem~\ref{theorem:virtual-clocks:TSVCSimIffTASim}). Therefore, if there exists a contradiction for $\VCGFullState{A}{}$ and $\VCGFullState{B}{}$, a contradiction for $\VCGFullState{A}{A, \varepsilon}$ and $\VCGFullState{B}{B, \varepsilon}$ will be found (and we have to check the outgoing $\varepsilon$-transitions anyway). This reasoning, however, only holds under the assumption that $n$ never reaches zero. $\VCGFullState{A}{A, \varepsilon} = \VCGFullState{A}{A, \varepsilon, \varepsilon}$ (and analogously for B) implies that once the outgoing $\varepsilon$-transition of a pair of symbolic states that is already the target of an outgoing $\varepsilon$-transition is checked, it will always return $\emptyset$, as it belongs to the visited set. Therefore, we are allowed to use alternating sequences.

We now present the \textsc{check-for-virt-bisim-impl} function in Algorithm~\ref{alg:virtual-clocks:bisim-no-n}, which checks for virtual bisimulation and uses alternating sequences. Furthermore, in the event that an $\varepsilon$-transition is deemed superfluous, the algorithm proceeds without it. \textsc{check-for-virt-bisim-impl} mainly differs to \textsc{check-for-virt-bisim-in-order-impl} by the missing $n$ and the fact that either the $\varepsilon$-transitions or the action transitions are checked. After the application of the $\syncfunction{}$ function, we check whether the pair of targets of the $\varepsilon$-transitions is equivalent to the current pair of symbolic states.

\begin{algorithm}[!ht]
  \caption{\textsc{check-for-virt-bisim-impl} function}
  \label{alg:virtual-clocks:bisim-no-n}
  \begin{algorithmic}[1]
  \LeftComment{Let $\VCGFullState{A}{}$, $\VCGFullState{B}{}$ be AB-\semisynchronized{} symbolic states,}
  \LeftComment{$k : C_A \cup C_B \cup \{\chi_{0}, ..., \chi_{|C_A| + |C_B| - 1}\} \rightarrow \mathbb{N}^{\geq 0}$, and $\text{visited}$ be a set.}
  \LeftComment{The return value of \textsc{check-for-virt-bisim-impl} is a set of virtual constraints}
  \Function{check-for-virt-bisim-impl${}_{\VCG{A}, \VCG{B}, k, \textup{visited}}$}{$(l_A, \Zone{A})$, $(l_B, \Zone{B})$}
    \If{$\lnot (\VirtualEquivalence{A}{B})$}
      \State \Return $\text{extract-virtual-constraint}(\Zone{A}) \land \lnot \ \text{extract-virtual-constraint}(\Zone{B}) \cup $
      \State \qquad \qquad \qquad \qquad $\text{extract-virtual-constraint}(\Zone{B}) \land \lnot \ \text{extract-virtual-constraint}(\Zone{A})$
    \EndIf
    \State $(\VCGFullState{A}{e, A}, \VCGFullState{B}{e, B}) \gets \syncfunction(\VCGFullState{A}{}, \VCGFullState{B}{B})$
    \State
    \LeftComment{Assume $\VCGFullState{A}{e, A} \VCGTrans{\varepsilon} \VCGFullState{A}{\varepsilon, A}$ and $\VCGFullState{B}{e, B} \VCGTrans{\varepsilon} \VCGFullState{B}{\varepsilon, B}$.}
    \If{$(\Zone{e, A}, \Zone{e, B}) \neq (\Zone{\varepsilon, A}, \Zone{\varepsilon, B})$}
      \State $\varepsilon\text{-result} \gets $ \Call{check-for-virt-bisim-impl${}_{\VCG{A}, \VCG{B}, k, \textup{visited}}$}{$(l_A, \Zone{\varepsilon, A})$, $(l_B, \Zone{\varepsilon, B})$}
      \State $\text{contra} \gets \text{revert-}\varepsilon \text{-trans}(\Zone{e, A}, \Zone{\varepsilon, A}, \varepsilon\text{-result} \land \text{extract-virtual-constraint}(\Zone{\varepsilon, A})))$
      \State \qquad \qquad \qquad $\cup \text{revert-}\varepsilon \text{-trans}(\Zone{e, B}, \Zone{\varepsilon, B}, \varepsilon\text{-result} \land \text{extract-virtual-constraint}(\Zone{\varepsilon, B}))$
  
      \State \Return $\revertsyncfunction(\VCGFullState{A}{}, \VCGFullState{B}{}, \text{contra})$
  
    \EndIf
    \State
    \State $\VCGFullState{A}{\text{norm}, A} \gets (l_A, \text{norm}(\Zone{e, A}, k))$, $\VCGFullState{B}{\text{norm}, B} \gets (l_B, \text{norm}(\Zone{e, B}, k))$
    \IfThen{($(\VCGFullState{A}{\text{norm}, A}, \VCGFullState{B}{\text{norm}, B}) \in \text{visited}$)}{\Return $\emptyset$}
    \State $\text{new-visited} \gets \text{visited} \cup \{(\VCGFullState{A}{\text{norm}, A}, \VCGFullState{B}{\text{norm}, B})\}$
    \State $\textsc{func} = \textsc{check-for-virt-bisim-impl}_{\VCG{A}, \VCG{B}, k, \textup{new-visited}}$
    \State
    \ForAll{$\sigma \in \Sigma$}
      \LeftComment{$\text{out-trans}(\sigma, \VCGFullState{}{})$ is the set of all outgoing transitions of $\VCGFullState{}{}$ labeled with $\sigma$.}
      \State $\text{contradiction} \gets \textsc{check-outgoing-transitions-impl}_{\VCG{A}, \VCG{B}, \textsc{func}}(\Zone{e, A}, \Zone{e, B}$
      \State \qquad \qquad \qquad \qquad $\text{out-trans}(\sigma, \VCGFullState{A}{e, A}), \text{out-trans}(\sigma, \VCGFullState{B}{e, B}))$
      \IfThen{($\text{contradiction} \neq \emptyset$)}{\Return $\revertsyncfunction(\VCGFullState{A}{e, A}, \VCGFullState{B}{e, B}, \text{contradiction})$}
    \EndFor
    \State \Return $\emptyset$
  \EndFunction{}
  \end{algorithmic}
\end{algorithm}

If this is not the case, we check only the outgoing $\varepsilon$-transitions. By Proposition~\ref{prop:virtual-clocks:VCG:outgoing-transitions-of-symbolic-syncd-states}, we know that these targets are AB-\synchronized{} and, therefore, if the targets are virtually equivalent, the $\syncfunction{}$ function has no impact. Therefore, after an $\varepsilon$-transition, the condition in the second if-statement is false. This implies that at least at every second recursion step, the pair of targets of the $\varepsilon$-transitions is equivalent to the current pair of symbolic states.

If this is the case, the following part of the algorithm is identical to the corresponding part of the aformentioned \textsc{check-for-virt-bisim-in-order-impl} function. However, the check of the outgoing $\varepsilon$-transitions is omitted, as it would inevitably yield an empty set.

\begin{proposition}[$\textsc{check-for-virt-bisim-impl}$ is Correct]
\label{prop:virtual-clocks:implementing-check-for-virt-bisim}
For $n$ larger than the maximum size of the visited set, \NoHyper Proposition~\ref{prop:appendix:check-for-virt-bisim-in-order-impl:check-for-virt-bisim-in-order-impl-is-correct}\endNoHyper{} also holds for \textsc{check-for-virt-bisim-} \textsc{impl}.
\end{proposition}

\begin{shortproof}
Follows directly from Proposition~\ref{prop:appendix:check-for-virt-bisim-in-order-impl:check-for-virt-bisim-in-order-impl-is-correct}.
\end{shortproof}

Since the region-based approach of \u{C}er\={a}ns~\cite{Cerans1992} has an exponential behavior in the number of symbolic states when applied to linear TA, our algorithm improves the complexity of the bisimulation check significantly.

\subsection{Comparison to the Region-Based Approach}

Since \u{C}er\={a}ns~\cite{Cerans1992} uses parallel timer processes (PTPs) and only shows the construction of a finite set but not an actual algorithm, we have to do some interpretation for the comparison. While it may be possible to do slightly better using the region-based construction, the main point remains valid. Moreover, to the best of our knowledge, there is no tool implementing the approach and we need to do the calculations by hand.

The region-based approach proposed by \u{C}er\={a}ns~\cite{Cerans1992} uses a product graph. It thus follows that, given two TA $A$ and $B$ with sets of clocks $C_A$ and $C_B$, the analyzed symbolic states have $|C_A| + |C_B|$ many clocks. Since the symbolic states contained in the $\text{visited}$ set in Algorithm \ref{alg:virtual-clocks:bisim-no-n} are AB-\synchronized, it follows that the original clocks have the same values as the corresponding virtual clocks. Consequently, only the $|C_A| + |C_B|$ virtual clocks are relevant for the number of recursion steps and, therefore, the number of relevant clocks is the same in both constructions. 
It should be noted, however, that the region-based approach is subject to an analog problem as described in Example \ref{ex:background:alternating-sequences}.

\begin{example}

\newcommand{\AlgExampleRegionScaling}{1cm}
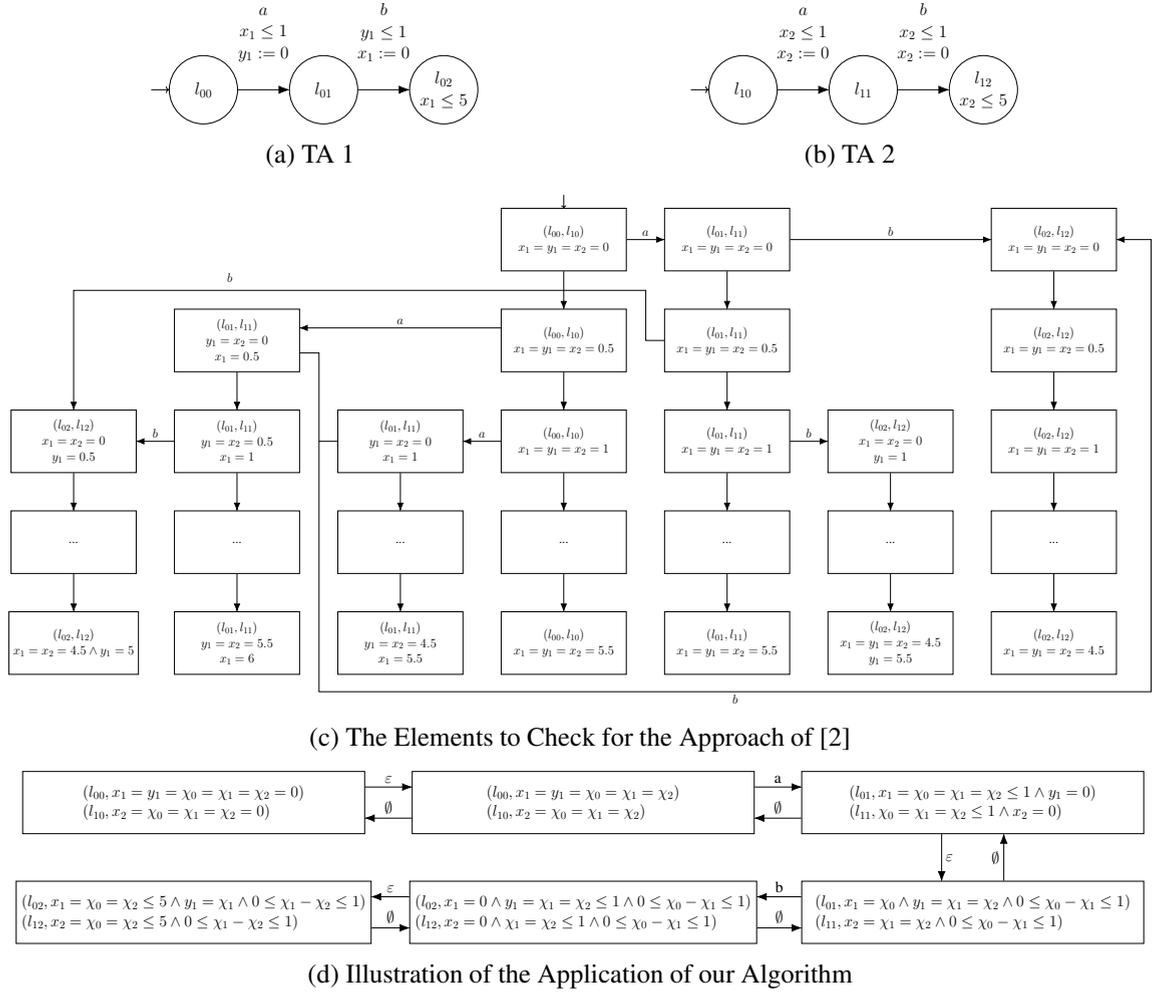
\begin{figure}[h!]
\centering
\begin{subfigure}{0.45\textwidth}
\centering
\scalebox{0.45}{
\begin{tikzpicture}
\tikzstyle{every node}=[font=\backgroundTikzFontSize]
\tikzstyle{state} = [draw,circle,minimum size=2cm,inner sep=0pt,semithick]

\node[state, align=center, initial left, initial text=] (0) {$l_{00}$};
\node[state, align=center, right = 1.5cm of 0] (1) {$l_{01}$};
\node[state, align=center, right = 1.5cm of 1] (2) {$l_{02}$ \\ $x_1 \leq 5$};
\backgroundExampleTAArrowDesc (0) --node[above = 0.7cm, align=center]{$a$\\$x_1 \leq 1$\\$y_1:=0$} (1);
\backgroundExampleTAArrowDesc (1) --node[above = 0.7cm, align=center]{$b$\\$y_1 \leq 1$\\$x_1:=0$} (2);
\end{tikzpicture}
}
\caption{TA 1}
\label{subfig:algorithm:example:comparison:TA1}
\end{subfigure}
\begin{subfigure}{0.45\textwidth}
\centering
\scalebox{0.45}{
\begin{tikzpicture}
\tikzstyle{every node}=[font=\backgroundTikzFontSize]
\tikzstyle{state} = [draw,circle,minimum size=2cm,inner sep=0pt,semithick]

\node[state, align=center, initial left, initial text=] (0) {$l_{10}$};
\node[state, align=center, right = 1.5cm of 0] (1) {$l_{11}$};
\node[state, align=center, right = 1.5cm of 1] (2) {$l_{12}$ \\ $x_2 \leq 5$};
\backgroundExampleTAArrowDesc (0) --node[above = 0.7cm, align=center]{$a$\\$x_2 \leq 1$\\$x_2:=0$} (1);
\backgroundExampleTAArrowDesc (1) --node[above = 0.7cm, align=center]{$b$\\$x_2 \leq 1$\\$x_2:=0$} (2);
\end{tikzpicture}
}
\caption{TA 2}
\label{subfig:algorithm:example:comparison:TA2}
\end{subfigure}
\begin{subfigure}{\textwidth}
\centering
\scalebox{0.33}{
\begin{tikzpicture}
\tikzstyle{every node}=[font=\backgroundTikzFontSize]
\tikzstyle{symstate} = [draw,rectangle,minimum width=5cm,minimum height=2.5cm,inner sep=5pt,thick]
%
\node[symstate, align=center, initial above, initial text=] (00) {$(l_{00}, l_{10})$\\$x_1=y_1=x_2=0$};
\node[symstate, align=center, below = \AlgExampleRegionScaling of 00] (01) {$(l_{00}, l_{10})$\\$x_1=y_1=x_2=0.5$};
\node[symstate, align=center, below = \AlgExampleRegionScaling of 01] (02) {$(l_{00}, l_{10})$\\$x_1=y_1=x_2=1$};
\node[symstate, align=center, below = \AlgExampleRegionScaling of 02] (03) {...};
\node[symstate, align=center, below = \AlgExampleRegionScaling of 03] (04) {$(l_{00}, l_{10})$\\$x_1=y_1=x_2=5.5$};
%
\node[symstate, align=center, right = \AlgExampleRegionScaling of 00] (10) {$(l_{01}, l_{11})$\\$x_1=y_1=x_2=0$};
\node[symstate, align=center, below = \AlgExampleRegionScaling of 10] (11) {$(l_{01}, l_{11})$\\$x_1=y_1=x_2=0.5$};
\node[symstate, align=center, below = \AlgExampleRegionScaling of 11] (12) {$(l_{01}, l_{11})$\\$x_1=y_1=x_2=1$};
\node[symstate, align=center, below = \AlgExampleRegionScaling of 12] (13) {...};
\node[symstate, align=center, below = \AlgExampleRegionScaling of 13] (14) {$(l_{01}, l_{11})$\\$x_1=y_1=x_2=5.5$};
%
\node[symstate, align=center, right = \AlgExampleRegionScaling of 12] (60) {$(l_{02}, l_{12})$\\$x_1=x_2=0$\\$y_1 = 1$};
\node[symstate, align=center, below = \AlgExampleRegionScaling of 60] (61) {...};
\node[symstate, align=center, below = \AlgExampleRegionScaling of 61] (62) {$(l_{02}, l_{12})$\\$x_1=y_1=x_2=4.5$\\$y_1=5.5$};
%
\node[symstate, align=center, right = \AlgExampleRegionScaling of 60] (42) {$(l_{02}, l_{12})$\\$x_1=y_1=x_2=1$};
\node[symstate, align=center, above = \AlgExampleRegionScaling of 42] (41) {$(l_{02}, l_{12})$\\$x_1=y_1=x_2=0.5$};
\node[symstate, align=center, above = \AlgExampleRegionScaling of 41] (40) {$(l_{02}, l_{12})$\\$x_1=y_1=x_2=0$};
\node[symstate, align=center, below = \AlgExampleRegionScaling of 42] (43) {...};
\node[symstate, align=center, below = \AlgExampleRegionScaling of 43] (44) {$(l_{02}, l_{12})$\\$x_1=y_1=x_2=4.5$};
%
\node[symstate, align=center, left = \AlgExampleRegionScaling of 02] (32) {$(l_{01}, l_{11})$\\$y_1= x_2 = 0$ \\ $x_1=1$};
\node[symstate, align=center, below = \AlgExampleRegionScaling of 32] (33) {...};
\node[symstate, align=center, below = \AlgExampleRegionScaling of 33] (34) {$(l_{01}, l_{11})$\\$y_1= x_2 = 4.5$ \\ $x_1=5.5$};
%
\node[symstate, align=center, left = \AlgExampleRegionScaling of 32] (22) {$(l_{01}, l_{11})$\\$y_1= x_2 = 0.5$ \\ $x_1=1$};
\node[symstate, align=center, above = \AlgExampleRegionScaling of 22] (21) {$(l_{01}, l_{11})$\\$y_1=x_2 = 0$\\$x_1=0.5$};
\node[symstate, align=center, below = \AlgExampleRegionScaling of 22] (23) {...};
\node[symstate, align=center, below = \AlgExampleRegionScaling of 23] (24) {$(l_{01}, l_{11})$\\$y_1= x_2 = 5.5$ \\ $x_1=6$};
%
\node[symstate, align=center, left = \AlgExampleRegionScaling of 22] (51) {$(l_{02}, l_{12})$\\$x_1=x_2=0$\\$y_1=0.5$};
\node[symstate, align=center, below = \AlgExampleRegionScaling of 51] (52) {...};
\node[symstate, align=center, below = \AlgExampleRegionScaling of 52] (53) {$(l_{02}, l_{12})$\\$x_1=x_2=4.5 \land y_1 = 5$};
%
%
\backgroundExampleTAArrowDesc (00) -- (01);
\backgroundExampleTAArrowDesc (01) -- (02);
\backgroundExampleTAArrowDesc (02) -- (03);
\backgroundExampleTAArrowDesc (03) -- (04);
%
\backgroundExampleTAArrowDesc (10) -- (11);
\backgroundExampleTAArrowDesc (11) -- (12);
\backgroundExampleTAArrowDesc (12) -- (13);
\backgroundExampleTAArrowDesc (13) -- (14);
%
\backgroundExampleTAArrowDesc (60) -- (61);
\backgroundExampleTAArrowDesc (61) -- (62);
%
\backgroundExampleTAArrowDesc (40) -- (41);
\backgroundExampleTAArrowDesc (41) -- (42);
\backgroundExampleTAArrowDesc (42) -- (43);
\backgroundExampleTAArrowDesc (43) -- (44);
%
\backgroundExampleTAArrowDesc (32) -- (33);
\backgroundExampleTAArrowDesc (33) -- (34);
%
\backgroundExampleTAArrowDesc (21) -- (22);
\backgroundExampleTAArrowDesc (22) -- (23);
\backgroundExampleTAArrowDesc (23) -- (24);
%
\backgroundExampleTAArrowDesc (51) -- (52);
\backgroundExampleTAArrowDesc (52) -- (53);
%
%
%
\backgroundExampleTAArrowDesc (00) --node[above, align=center]{$a$} (10);
\backgroundExampleTAArrowDesc (10) --node[above, align=center]{$b$} (40);
\backgroundExampleTAArrowDesc (11) -| ($(11)+(-3.25cm, 2cm)$) -|node[above right = 0.2cm and 6cm, align=center]{$b$} (51);
\backgroundExampleTAArrowDesc (12) --node[above, align=center]{$b$} (60);
%
%
\backgroundExampleTAArrowDesc (02) --node[above, align=center]{$a$} (32);
\backgroundExampleTAArrowDesc (32) -| ($(32)+(-3.25cm, -9cm)$) --node[below, align=center]{$b$} ($(32) + (28cm, -9cm)$) |- (40.east);
%
%
\backgroundExampleTAArrowDesc ($(01.west) + (0, 0.5cm)$) --node[above, align=center]{$a$} ($(21.east) + (0, 0.5cm)$);
\draw ($(21.east) + (0, -0.5cm)$) -| ($(32)+(-3.25cm, 0)$);
\backgroundExampleTAArrowDesc (22) --node[above, align=center]{$b$} (51);
\end{tikzpicture}
}
\caption{The Elements to Check for the Approach of \cite{Cerans1992}}
\label{subfig:algorithm:example:comparison:Region}
\end{subfigure}
\begin{subfigure}{\textwidth}
\centering
\scalebox{0.41}{
\begin{tikzpicture}
\tikzstyle{every node}=[font=\AlgTikzFontSize]
\tikzstyle{foo} = [draw,rectangle,minimum width=11cm, minimum height = 2cm, inner sep=5pt,thick, align=left]
\node[foo](init) {$(l_{00}, x_1 = y_1 = \chi_0 = \chi_1 = \chi_2 = 0)$ \\ $(l_{10}, x_2 = \chi_0 = \chi_1 = \chi_2 = 0)$};
\node[foo, right = 1.5cm of init](initeps) {$(l_{00}, x_1 = y_1 = \chi_0 = \chi_1 = \chi_2)$ \\ $(l_{10}, x_2 = \chi_0 = \chi_1 = \chi_2)$};
\node[foo, right= 1.5cm of initeps](second) {$(l_{01}, x_1 = \chi_0 = \chi_1 = \chi_2 \leq 1 \land y_1 = 0)$ \\ $(l_{11}, \chi_0 = \chi_1 = \chi_2 \leq 1 \land x_2 = 0)$};
\node[foo, below= 1cm of second](secondeps) {$(l_{01}, x_1 = \chi_0 \land y_1 = \chi_1 = \chi_2 \land 0 \leq \chi_0 - \chi_1 \leq 1 )$ \\ $(l_{11}, x_2 = \chi_1 = \chi_2 \land 0 \leq \chi_0 - \chi_1 \leq 1)$};
\node[foo, below= 1cm of initeps](third) {$(l_{02}, x_1 = 0 \land y_1 = \chi_1 = \chi_2 \leq 1 \land 0 \leq \chi_0 - \chi_1 \leq 1)$ \\ $(l_{12}, x_2 = 0 \land \chi_1 = \chi_2 \leq 1 \land 0 \leq \chi_0 - \chi_1 \leq 1)$};
\node[foo, below= 1cm of init](thirdeps) {$(l_{02}, x_1 = \chi_0 = \chi_2 \leq 5 \land y_1 = \chi_1 \land 0 \leq \chi_1 - \chi_2 \leq 1 )$ \\ $(l_{12}, x_2 = \chi_0 = \chi_2 \leq 5 \land 0 \leq \chi_1 - \chi_2 \leq 1)$};%
\AlgExampleArrowDesc ($(init.east) + (0, 0.5cm)$) --node[above, align=center]{$\varepsilon$} ($(initeps.west) + (0, 0.5cm)$);
\AlgExampleArrowDesc ($(initeps.west) + (0, -0.5cm)$) --node[above, align=center]{$\emptyset$} ($(init.east) + (0, -0.5cm)$);
\AlgExampleArrowDesc ($(initeps.east) + (0, 0.5cm)$) --node[above, align=center]{a} ($(second.west) + (0, 0.5cm)$);
\AlgExampleArrowDesc ($(second.west) + (0, -0.5cm)$) --node[above, align=center]{$\emptyset$} ($(initeps.east) + (0, -0.5cm)$);
\AlgExampleArrowDesc ($(second.south) + (-1cm, 0)$) --node[right, align=center]{$\varepsilon$} ($(secondeps.north) + (-1cm, 0)$);
\AlgExampleArrowDesc ($(secondeps.north) + (1cm, 0)$) --node[left, align=center]{$\emptyset$} ($(second.south) + (1cm, 0)$);
\AlgExampleArrowDesc ($(secondeps.west) + (0, 0.5cm)$) --node[above, align=center]{b} ($(third.east) + (0, 0.5cm)$);
\AlgExampleArrowDesc ($(third.east) + (0, -0.5cm)$) --node[above, align=center]{$\emptyset$} ($(secondeps.west) + (0, -0.5cm)$);
\AlgExampleArrowDesc ($(third.west) + (0, 0.5cm)$) --node[above, align=center]{$\varepsilon$} ($(thirdeps.east) + (0, 0.5cm)$);
\AlgExampleArrowDesc ($(thirdeps.east) + (0, -0.5cm)$) --node[above, align=center]{$\emptyset$} ($(third.west) + (0, -0.5cm)$);
\end{tikzpicture}
}
\caption{Illustration of the Application of our Algorithm}
\label{subfig:algorithm:example:comparison:VCG}
\end{subfigure}
\caption{Comparison between \cite{Cerans1992} and our Algorithm}
\label{fig:algorithm:example:comparison}
\end{figure}
Figures~\ref{subfig:algorithm:example:comparison:TA1} and \ref{subfig:algorithm:example:comparison:TA2} show two simple, deterministic, and bisimilar TA. As shown in Figure~\ref{subfig:algorithm:example:comparison:Region}, the approach from \cite{Cerans1992} now takes a state from each region of the product of both TA. Since this results in multiple uses of the first transition, and some of the paths result in multiple uses of the second transition, this results into an exponential number of nodes to be checked. This is not the case when using our Algorithm as shown in Figure~\ref{subfig:algorithm:example:comparison:VCG}.
\end{example}

The next section shows that our algorithm can check realistic TA taken from community benchmarks within an acceptable amount of time.

%% file: sections/evaluation.tex
\section{Experimental Evaluation and Comparison}
\label{sec:evaluation}

In this section, we demonstrate the practical usability of our algorithm and show that our tool is currently the one to use when checking for timed bisimilarity. To achieve this, we compare our tool to \textsc{Caal}, the only currently available tool for checking timed bisimilarity we are aware of. Since \textsc{Caal} accepts processes written in \textit{Timed Calculus of Communicating Systems (TCCS)}, we translate the TA under test into TCCS.

In \textsc{Caal}, the action $\tau$ has a special meaning. If a transition labeled with $\tau$ is enabled, no delay can occur, while a transition with any other action can be delayed by any amount of time. Therefore, we use transitions labeled with $\tau$ to translate invariants of TA into the language of \textsc{Caal}.

As far as we understand \textsc{Caal}, the tool assumes a discrete time semantics and enumerates the states of the resulting TLTS according to Definition~\ref{def:background:sim-and-bisim:Strong-Timed-Bisimulation}. If the TLTS becomes too large, it is either cut-off, which implies that false positives can occur, or a \textit{Too Much Recursion Error (TMRE)} occurs. The first behavior is demonstrated by the following example.

\begin{example}
\begin{figure}[h!]
  \begin{tabular}{p{0.4\textwidth}p{0.4\textwidth}}
    \begin{minipage}{0.45\textwidth}%
      \centering
      \scalebox{0.45}{
        \begin{tikzpicture}
        \tikzstyle{every node}=[font=\backgroundTikzFontSize]
        \tikzstyle{state} = [draw,circle,minimum size=2cm,inner sep=0pt,semithick]
        
        \node[state, align=center, initial left, initial text=] (0) {$l_{0}$\\$x \leq p$};
        \node[state, align=center, above right = -0.3cm and 1.5cm of 0] (1) {$l_{1}$\\$x \leq 1$};
        \node[state, align=center, below right = -0.3cm and 1.5cm of 0] (2) {$l_{2}$\\$x \leq 1$};
        \backgroundExampleTAArrowDesc (0) --node[above left, align=center, xshift=0.35cm, yshift=0.3cm]{$a$\\$x:=0$} (1);
        \backgroundExampleTAArrowDesc (0) --node[below left, align=center, xshift=0.2cm, yshift=-0.2cm]{$\tau$\\$x \geq p$\\$x := 0$} (2);
        \backgroundExampleTAArrowDesc (1) to[loop right] node[right, align=center]{$\tau$\\$x \geq 1$\\$x:=0$} (1);
        \backgroundExampleTAArrowDesc (2) to[loop right] node[right, align=center]{$\tau$\\$x \geq 1$\\$x:=0$} (2);
      \end{tikzpicture}
      }
    \end{minipage}%
    &%
    \begin{minipage}{0.45\textwidth}%
      \begin{lstlisting}
L0 = a.L1 + p.tau.L2;
L1 = 1.tau.L1;
L2 = 1.tau.L2;
      \end{lstlisting}
    \end{minipage}
  \end{tabular}
\caption{TA (left) and TCCS (right) of our Synthetic Example.}
\label{fig:evaluation:Synthetic_Example:Template}
\end{figure}
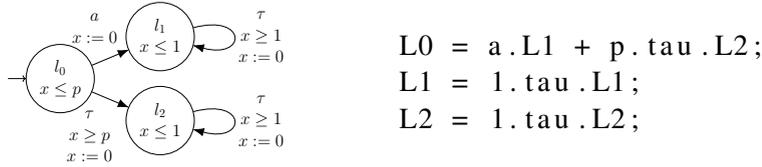
Figure~\ref{fig:evaluation:Synthetic_Example:Template} shows the templates for this example. On the left hand side, the template of our TA can be seen. $l_0$ requires $x$ to be lower or equals to a parameter $p$ and has two outgoing transitions, one with the action $a$ and one with $\tau$, and the transition labeled with $\tau$ can only be used if $x = p$. The locations $l_1$ and $l_2$ both have self-loops which use the action $\tau$ and reset clock $x$.

The right hand side of Figure~\ref{fig:evaluation:Synthetic_Example:Template} shows the template for \textsc{Caal}. The initial process is L0, which has the options to use the action $a$ and switch to process L1 or to wait for $p$ time steps and switch to process L2. The processes L1 and L2 wait a timestep and then become itself again using a $\tau$-transition.

If we instantiate the templates with $p = 100$ and $p = 101$, resulting in two different processes that are not timed bisimilar, \textsc{Caal} returns a false positive, which does not happen when using our tool.
\end{example}

The main reason for this behavior is that if the TLTS were not cut off, they would quickly become very large. Therefore, any TLTS-based approach suffers from this problem either by having infinite large runtimes or by producing false positives. Moreover, any approach based on TLTS requires discrete-time modelling, which is not the case for our tool.

We now compare both tools by evaluating frequently used community examples as our subject systems.

\subsection{Evaluation of our Subject Systems}
\label{subsec:evaluation:fre_used_examples}

To evaluate our tool in practice, we utilize three different TA models, which are frequently used in the evaluation of TA analysis techniques. Collision Avoidance (CA) \cite{CollisionAvoidance} is a model of a protocol for the avoidance of collisions on an Ethernet-like broadcast medium, the Root Contention Protocol (RCP) \cite{IEEE1394RCP} is a model of the IEEE1394 root contention protocol, and Audio/Video Components (AVC) \cite{AVProtocol} is a model of a protocol used in the industry for the purpose of controlling the transmission of messages between audio and video components over a shared bus.

\begin{table}
\centering
\begin{tabular}{ | c c c c |}
\hline
Name & \# Locations & \# Switches & \# Clocks \\
\hline
CA & 6 & 13 & 1 \\
RCP & 10 & 26 & 2 \\
AVC & 18 & 30 & 1\\
\hline
\end{tabular}
\caption{Statistics of the models to benchmark}
\label{tab:eval:statistics-of-the-models-to-benchmark}
\end{table}

The statistics of the models can be seen in Table~\ref{tab:eval:statistics-of-the-models-to-benchmark}. CA is the smallest model, with six locations, 13 switches, and a single clock. RCP has ten locations, 26 switches, and two clocks, while AVC has 18 locations, 30 switches, and a single clock.
All these models were originally published as \textsc{Uppaal} models. We used the uppaal-to-tchecker converter \cite{UppaalToTchecker} to generate the corresponding TChecker models. For each of these TA, we generated four mutants to obtain a set of similar but different models, using a mutation testing approach described in \cite{MutationTesting}. The first mutant is always timed bisimilar, while the other mutants are not. The first mutant is created by adding a reset (except for the RCP example, where this was not possible and, therefore, we doubled a transition), the second mutant is created by changing an invariant, the third mutant is created by changing a guard, and the fourth mutant is created by removing a reset.

Unfortunately, the found community benchmarks do not include any non-determinism. Therefore, we used another operation of the mutant testing approach that changes the action of a transition. Sometimes this leads to non-deterministic mutants. For each of our models, we picked a non-deterministic mutant and repeated our evaluation procedure, which allows us to analyze the impact of a single non-deterministic choice.

To compare the tools, we applied them first to the model and its one-to-one copy and afterwards to the model and the corresponding mutants. The results are shown in Table~\ref{tab:eval:real_world_TA:benchmark_tchecker} and Table~\ref{tab:eval:real_world_TA:benchmark_tchecker_caal}. The tools were run eleven times for each pair (always with a reboot in between), and the time values shown are average values. The variances between measurements are shown next to the time values. Table~\ref{tab:eval:real_world_TA:benchmark_tchecker} also shows the number of checked pairs of symbolic states, while Table~\ref{tab:eval:real_world_TA:benchmark_tchecker_caal} also shows whether the returned result was correct, which is always the case for our tool (the ground truth was obtained manually).

The evaluation was conducted on a workstation equipped with an Intel i7-6700K processor and 64GB main memory, running a Linux Mint 21.2 ("Victoria") operating system.

\addtolength{\tabcolsep}{-0.4em}
\begin{table}
  \centering
  \begin{tabular}{| c | c c c c c c c c c |}
  \hline
    & \multicolumn{3}{c}{CA} & \multicolumn{3}{c}{RCP} & \multicolumn{3}{c |}{AVC} \\
  Determ. & \# Pairs & t [ms] & Var. & \# Pairs & t [ms] & Var & \# Pairs & t [ms] & Var \\
  \hline
  One-To-One & 16 & 0.3 & 0.00 & 698 & 18.3 & 3.81 & 808 & 14.4 & 11.32 \\
  Bisim & 16 & 0.3 & 0.00 & 732 & 20.0 & 4.50 & 808 & 13.4 & 0.01 \\
  Changed G. & 12 & 0.3 & 0.00 & 5 & 0.4 & 0.01 & 138 & 2.7 & 0.53 \\
  Changed Inv. & 2 & 0.1 & 0.00 & 3 & 0.2 & 0.00 & 5 & 0.2 & 0.00\\
  Rm. Reset & 2 & 0.1 & 0.00 & 3 & 0.2 & 0.00 & 20 & 0.6 & 0.00 \\  
  \hline
  Non-Determ. & & & & & & & & & \\
  \hline
  One-To-One & 18 & 0.4 & 0.38 & 744 & 20.1 & 2.26 & 1784 & 32.4 & 18.14 \\
  Bisim & 18 & 0.4 & 0.0 & 778 & 21.9 & 8.05 & 1784 & 30.9 & 9.82 \\
  Changed G. & 14 & 0.4 & 0.01 & 8 & 0.4 & 8.80 & 595 & 10.3 & 6.45 \\
  Changed Inv. & 2 & 0.1 & 0.0 & 6 & 0.3 & 0.00 & 5 & 0.2 & 0.00\\
  Rm. Reset & 2 & 0.1 & 0.0 &6 & 0.4 & 0.01 & 59 & 1.5 & 0.02\\  
  \hline
  \end{tabular}
  \caption{Benchmark Results of our Tool}
  \label{tab:eval:real_world_TA:benchmark_tchecker}  
\end{table}

\begin{table}
  \centering
  \begin{tabular}{| c | c c c c c c c c |}
  \hline
    & \multicolumn{3}{c}{CA} & \multicolumn{2}{c}{RCP} & \multicolumn{3}{c |}{AVC} \\
  Determ.  & Correct & t [ms] & Var. & Correct & t [ms] & Correct & t [ms] & Var. \\
  \hline
  One-To-One & yes & 45.8 & 111.56 & - & TMRE & yes & 73.6 & 312.05 \\
  Bisim. & yes & 34.6 & 174.25 & - & TMRE & yes & 80.3 & 225.62 \\
  Changed Guard & yes & 48.2 & 184.56 & - & TMRE & yes & 80.1 & 232.09 \\
  Changed Inv. & yes & 42.0 & 172.00 & - & TMRE & no & 77.73 & 184.82 \\
  Removed Reset & yes & 46.6 & 111.85 & - & TMRE & yes & 100.8 & 0.56 \\
  \hline
  Non-Determ. & & & & & & & &  \\
  \hline
  One-To-One & yes & 39.3 & 296.62 & - & TMRE & yes & 75.7 & 401.42 \\
  Bisim. & yes & 41.7 & 176.82 & - & TMRE & yes & 75.5 & 125.47 \\
  Changed Guard & yes & 34.6 & 295.67 & - & TMRE & yes & 73.18 & 54.76 \\
  Changed Inv. & yes & 30.3 & 110.62 & - & TMRE & no & 70.8 & 106.36 \\
  Removed Reset & yes & 37.2 & 185.96 & - & TMRE & yes & 107.5 & 135.67 \\
  \hline
  \end{tabular}
  \caption{Benchmark Results of \textsc{Caal}}
  \label{tab:eval:real_world_TA:benchmark_tchecker_caal}  
\end{table}

For our tool, all average time values are below 33ms, which demonstrates the practical usability of our algorithm. Since the algorithm terminates upon the discovery of a contradiction, the number of symbolic states to be examined is reduced for non-bisimilar pairs, resulting in a faster termination compared to bisimilar pairs. Furthermore, the presence of two clocks in RCP, as opposed to the single clock in CA and AVC, leads to a higher computation time, sometimes even if the number of pairs is lower. As expected, a non-deterministic choice can lead to a significant increase in the number of pairs needed, as happened especially for the AVC example. The variance of our measurements is low and all measurements were within an acceptable interval.

For \textsc{Caal}, we can see that for CA, the time required is larger by a factor of about 100 and that \textsc{Caal} is unable to handle the RCP model due to its size. For AVC, the fact that a false positive occurs proves that the state space is not fully explored. Since we do not know the size of the ignored part of the state space, the time values shown are only a lower bound with an unclear upper bound. Nevertheless, all measured time values are above the corresponding values in Table~\ref{tab:eval:real_world_TA:benchmark_tchecker}. The variance is high, indicating that we usually had a wide range of measurement values.

Our tool performs significantly better than \textsc{Caal} for all three examples. In particular, the fact that RCP is rejected and that there is a false positive for AVC should be taken into account when choosing the tool to be used for checking timed bisimilarity.

\subsection{Threats to Validity}

From Section~\ref{subsec:evaluation:fre_used_examples}, we can conclude that the approach is sufficiently effective for the models. The scope of our experimental setting is limited to the class of safety TA. Our approach could be rendered inapplicable by any non-trivial TA extension. It is possible to extend the computation time by incorporating additional paths, for example by adding non-deterministic choices. Due to our approach of generating mutants, we only consider small and locally restricted changes. Nevertheless, we used proofs and exhaustive testing to ensure correctness of our theory and our tool implementation. The used tests and benchmarks are provided within the examples directory of our tool.

We compared our tool to the only tool we know with a similar functionality, namely \textsc{Caal}. Nevertheless, \textsc{Caal} works on TCCS instead of TA and, therefore, we had to translate the models. While there might be small improvements possible, the overall outcome will not be changed by using different translations. Timing Measurements are always subject to noise. We run all measurements eleven times and took the average value to reduce the impact.

The final threat we consider is the relatively small set of subject systems. The examples from community benchmarks are of reasonable size and complexity. The benchmarks are often used in experiments involving analysis techniques for TA. In future, further case studies should be considered.

%% file: sections/conclusion.tex
\section{Conclusion}\label{chap:conclusion}

We presented virtual clock graphs, an extension of zone graphs, with the objective of verifying timed bisimilarity. We used this formalism to develop an algorithm, which we implemented in the open-source tool TChecker. Our experimental evaluation demonstrates that the tool is fit for purpose in practice, particularly in deterministic scenarios, but also in non-deterministic cases. To the best of our knowledge, this is the first practically usable tool for checking timed bisimilarity that uses an algorithm with a publicly available proof of its correctness.

As future work, we intend to utilize virtual clock graphs to develop algorithms that check for different comparison definitions such as timed refinement or weak timed bisimilarity. Furthermore, we plan to improve our cut-off criterion to reduce computation time and to introduce a cache system. Finally, we would like to develop a certificate such that it allows the user to use a graphical user interface that allows for a better understanding of the result.